\PassOptionsToPackage{table}{xcolor} 
\documentclass[a4paper,onecolumn,11pt,accepted=2023-08-21]{quantumarticle}
\pdfoutput=1
\usepackage[utf8]{inputenc}
\usepackage[english]{babel}
\usepackage[T1]{fontenc}
\usepackage{amsmath}
\usepackage{hyperref}

\usepackage{tikz}
\usepackage{lipsum}

\usepackage{appendix}
\usepackage{amsmath, amsthm, amssymb,commath}
\usepackage{bm}
\usepackage{optidef}
\usepackage{booktabs}

\usepackage{multirow} 
\usepackage{dcolumn} 
\usepackage[table]{xcolor}

\newcommand{\sgn}{\mathop{\rm sgn}\,}

\newcommand{\I}{\mathrm{i}}

\def \be {\begin{equation}}
\def \ee {\end{equation}}
 
\newcommand{\tr}{\mathrm{Tr}}
\newcommand{\Tr}{\mathrm{Tr}}

\def \Re{\mathrm{Re}\,}
\def \Im{\mathrm{Im}\,}

\def \del{\partial}

\def \argmax{\mathop{\rm argmax}}

\def \cM{{\cal M}}
\def \cH{{\cal H}}
\def \cX{{\cal X}}
\def \sofh{{\cal S}({\cal H})}

\def \bbr{{\mathbb R}}

\def \sofc2{{\cal S}({\mathbb C}^2)}

\def \lofh{{\cal L}({\cal H})}

\def\>{\rangle}
\def\<{\langle}

\newtheorem{theorem}{Theorem}
\newtheorem{lemma}[theorem]{Lemma}
\newtheorem{proposition}[theorem]{Proposition}

\newtheorem{remark}[theorem]{Remark}
 
\def\Label{\label} 
 
\begin{document}

\title{Tight Cram\'{e}r-Rao type bounds for multiparameter quantum metrology through conic programming}

\author{Masahito Hayashi}\email{hmasahito@cuhk.edu.cn, masahito@math.nagoya-u.ac.jp}
\affiliation{School of Data Science, The Chinese University of Hong Kong,
Shenzhen, Longgang District, Shenzhen, 518172, China}
\affiliation{International Quantum Academy (SIQA), Futian District, Shenzhen 518048, China}
\affiliation{Graduate School of Mathematics, Nagoya University, Nagoya, 464-8602, Japan}
\author{Yingkai Ouyang}
\email{y.ouyang@sheffield.ac.uk}
\affiliation{Department of Physics \& Astronomy, University of Sheffield, Sheffield, S3 7RH, United Kingdom}
\orcid{0000-0003-1115-0074}

\begin{abstract}
In the quest to unlock the maximum potential of quantum sensors, it is of paramount importance to have practical measurement strategies that can estimate incompatible parameters with best precisions possible.
However, it is still not known how to find practical measurements with optimal precisions, even for uncorrelated measurements over probe states.
Here, we give a concrete way to find uncorrelated measurement strategies with optimal precisions.
We solve this fundamental problem by introducing a framework of conic programming that unifies the theory of precision bounds for multiparameter estimates for uncorrelated and correlated measurement strategies under a common umbrella.
Namely, we give precision bounds that arise from linear programs on various cones defined on a tensor product space of matrices, including a particular cone of separable matrices.
Subsequently, our theory allows us to develop an efficient algorithm that calculates both upper and lower bounds for the ultimate precision bound for uncorrelated measurement strategies, where these bounds can be tight.
In particular, the uncorrelated measurement strategy that arises from our theory saturates the upper bound to the ultimate precision bound.
Also, we show numerically that there is a strict gap between the previous efficiently computable bounds and the ultimate precision bound.
\end{abstract}

\maketitle

\section{Introduction}

Quantum sensors, by employing quantum resources, promise to estimate physical parameters with unprecedented precision beyond what is possible using classical resources. 
Quantum metrology is a research field that studies quantum sensors.
Quantum metrology schemes require the ability to both prepare parameter-dependent quantum probe states and perform quantum measurements on these states.
Armed with the statistics of the measurement outcomes, one can thereafter estimate the underlying parameters. 
A central question in quantum metrology is to find measurement strategies with the ultimate precision for these multiparameter estimates.
In the simplest scenario of estimating a single-parameter, the ultimate precision, given by the quantum ({\textsf{CR}) Cram\'er-Rao bound \cite{HELSTROM1967101,helstrom,holevo,nagaoka89,HM08}, along with the corresponding optimal measurement strategy,
are efficient to compute with knowledge of both the probe state and its dependence on the single parameter.

The theory of multiparameter quantum metrology\footnote{See 
\cite{D_and_D_2020}
and \cite[Section V]{sidhu2020geometric} for a recent review.} is considerably richer than the single parameter setting. For example, parameters can be fundamentally incompatible, as is often the case in quantum systems.
If we want to design the best quantum sensors that can simultaneously estimate incompatible parameters, we must find optimal practical measurement strategies for multiparameter quantum metrology.
However, even after decades of research, the question of how to determine these optimal measurement strategies remains unanswered and unknown.

In lieu of determining these optimal measurement strategies, the field has focused on determining bounds for the ultimate precision of quantum sensors that estimate incompatible parameters simultaneously.
Since most optimizations for precision bounds \cite{holevo,HM08,Albarelli2019_PRL,sidhu2019tight,nagaoka91,Haya,CSLA} are not based directly on measurement strategies, even if we can calculate the best precision bounds from these optimizations, we still will not know what the optimal measurement strategies are.

Regarding the theory of precision bounds, the theory of single parameter estimation differs substantially from the multiparameter case.
Namely, while the \textsf{CR} bound is tight for both correlated and uncorrelated measurement strategies in single-parameter estimation, this is not the case for multiple parameters.
This is because 
the SLD Cram\'er Rao (\textsf{SLD}) bound does not give the tight bound
in the multiple-parameter case
nevertheless it gives the tight bound
in the single parameter case.
That is, for multiparameter quantum metrology, 
the Holevo-Nagaoka (\textsf{HN}) bound \cite{holevo,nagaoka89,
HM08,Albarelli2019_PRL,sidhu2019tight} is efficient to compute
and always tight for correlated measurement strategies across multiple probe states.
Although it is often called the Holevo Cram\'er Rao bound,
it is called Holevo-Nagaoka bound in this paper and its reason is explained later.
Such correlated measurement strategies however require a large quantum device
across multiple probe states, this bound is not practical.
To accomplish state estimation in a practical way,
we need to design uncorrelated measurement strategies across multiple probe states,
using only measurement devices that access individual probe states. (See Fig. \ref{fig:figures}.)
The Nagaoka-Hayashi (\textsf{NH}) \cite{nagaoka89,nagaoka91,Haya,CSLA} bound
is not only efficient to compute, but also addresses such uncorrelated measurement strategies.
However, the \textsf{NH} bound might not be always tight for uncorrelated measurement strategies.
Hence, an efficient way to determine the ultimate precision bound for uncorrelated strategies remains unknown.

The ultimate precision bound for uncorrelated strategies in the multiparameter setting, or simply the \textsf{tight} bound, 
was formulated more than two decades ago as the optimal value of an 
infinite-dimensional optimization program with linear objective and constraint functions on a topological vector space \cite{hayashi97-2}. 
The advantage of this formulation of \textsf{tight} bound is that it is a direct optimization over measurement strategies in contrast to all other multiparameter precision bounds.

There however remain many outstanding questions pertaining to the \textsf{tight} bound.
\begin{description}
\item[(1)] How to efficiently determine the \textsf{tight} bound?
\item[(2)] How to determine the optimal uncorrelated measurement strategy for multiparameter quantum metrology that saturates the tight bound?
\item[(3)] What is the relationship between the \textsf{SLD} bound, 
the \textsf{HN} bound, the \textsf{NH} bound, and the \textsf{tight} bound?
\item[(4)] Is there a gap between the \textsf{tight} bound and the \textsf{NH} bound?
\end{description}

\begin{figure}
\centering
    \includegraphics[width=\textwidth]{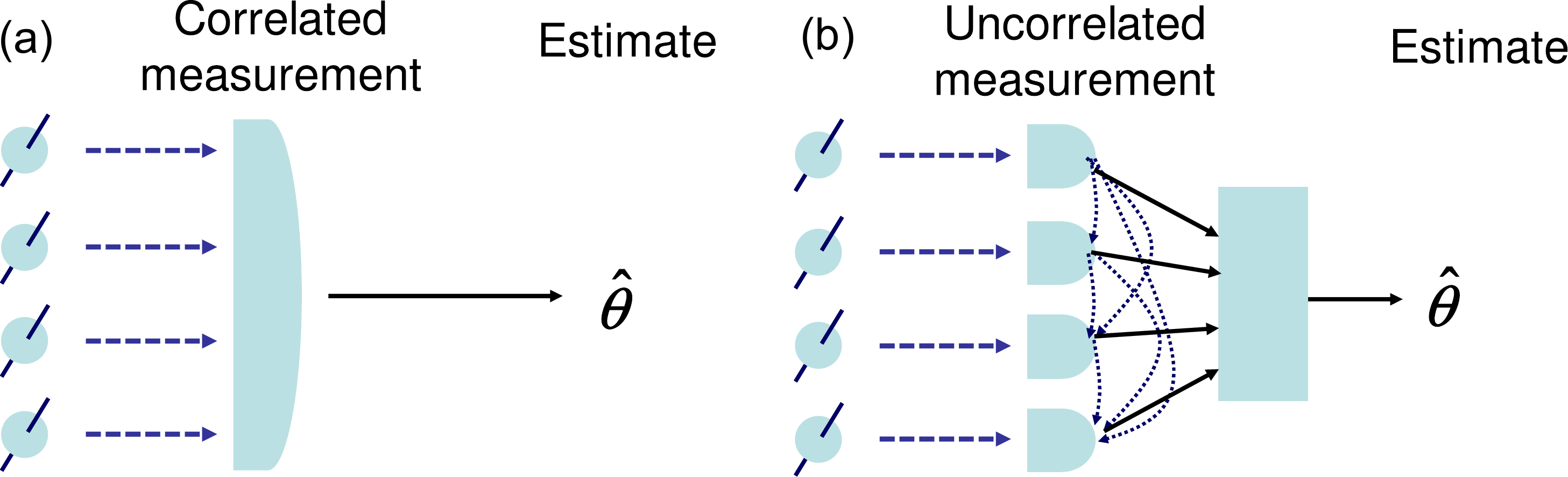}
\caption{Measurement strategies in parameter estimation using quantum probe states can be either correlated (a), or uncorrelated (b) across identical copies of the probe states.
In (a), the measurement device collectively measures the input states.
That is, the measurement device needs to access a large quantum system.
In (b) the measurement device individually measures the input states, but 
classical feedback is allowed to improve the measurement.
This strategy is composed of measurement device to access a single system.
}
\label{fig:figures}
\end{figure}

With regards to question (3), we unify the theory of 
the \textsf{SLD}, \textsf{HN}, \textsf{NH} and \textsf{tight} 
bounds under a common umbrella.
Remarkably, these bounds can alternatively be formulated as conic programs with the same linear objective and constraint functions. The only difference between these programs is the different choices of their cones. Namely, the cone for the \textsf{tight} bound is a strict subset of the cone for the \textsf{NH} bound, and the cone for the \textsf{NH} bound is a strict subset of the cone for the \textsf{HN} and \textsf{SLD} bounds.
This solves the open problem (3).

Regarding question (2), using our reformulation of the \textsf{tight} bound, we also construct an efficient algorithm that calculates the \textsf{tight} bound.
From the dual program of our reformulation of the \textsf{tight} bound,
we construct an associated semidefinite program (\textsf{SDP}), 
and show how to use its optimal solution to 
approximately solve the \textsf{tight} bound.
Namely, we provide an efficient algorithm that calculates both upper and lower bounds to the \textsf{tight} bound, thereby solving open problem (4).
Regarding question (4), using our algorithm, we numerically demonstrate that the \textsf{tight} bound can be strictly tighter than the \textsf{NH} bound.

We also solve questions (1) and (2), where we show concretely how we can efficiently compute optimal uncorrelated measurement strategies that saturate the \textsf{tight} bound along with the \textsf{tight} bound.

Our numerically tight estimates to the \textsf{tight} bounds are useful beyond multiparameter estimation theory.
More abstractly, we develop an approach to optimize over the separable cone for bipartite systems. Hence we expect our theory can pave the way towards gaining new insights into optimizations over separable bipartite states, which could be useful for entanglement theory \cite{Gurvits,bruss2002characterizing} and more general quantum resource theories \cite{uola2019conic,takagi2019prl,takagi2019prx}.

Now we sketch the organization of our paper. We also give the structure of our paper visually in Fig.~\ref{fig:structure}. In Table \ref{table:notations}, we list the important notations we used in our paper along with their meanings.
\begin{figure*}
         \centering
    \includegraphics[width=\textwidth]{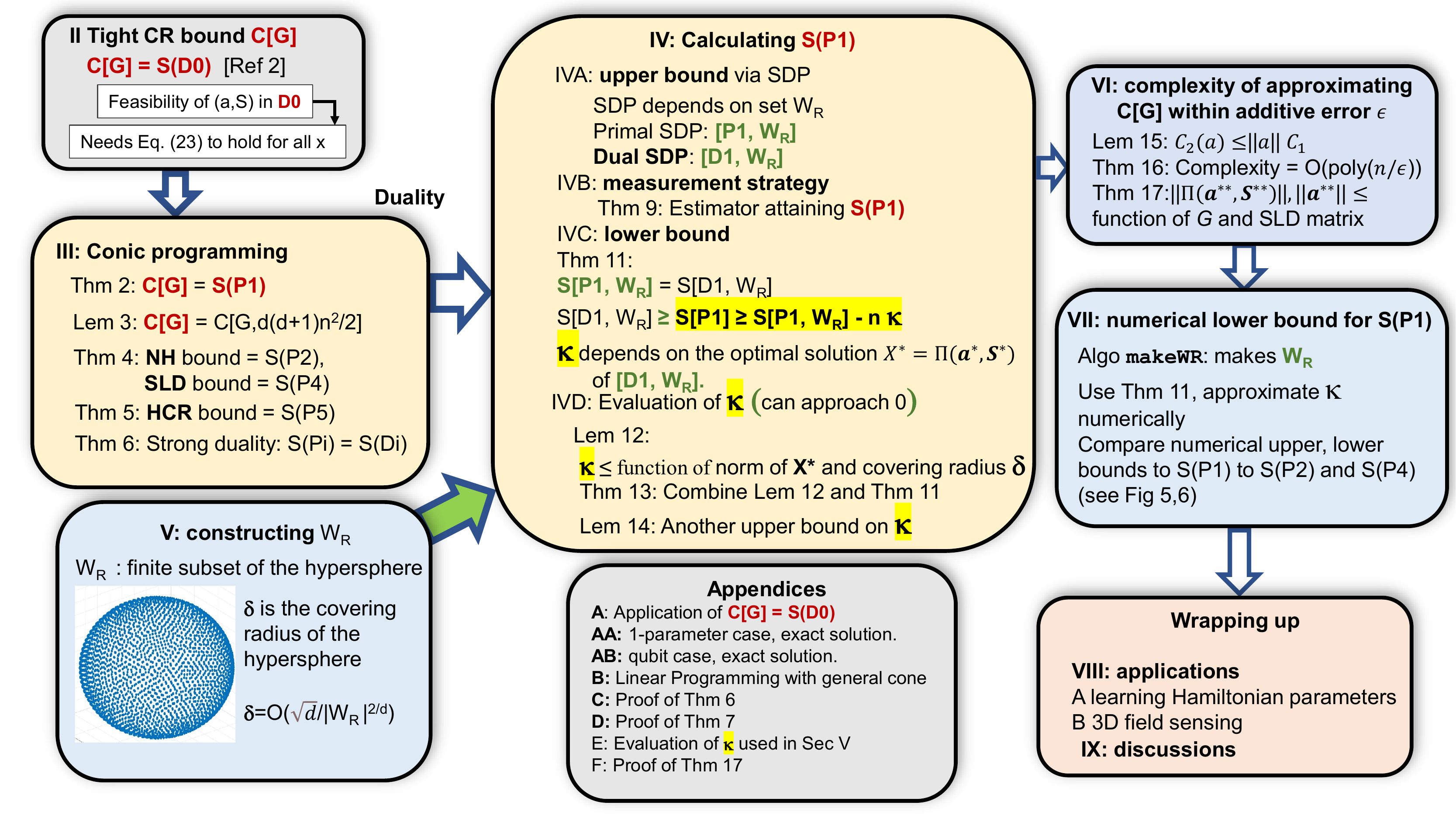}
    \caption{Structure of our paper.}
    \label{fig:structure}
\end{figure*}
In Section \ref{S2} we review various \textsf{CR}-type bounds in multiparameter quantum metrology and define 
the \textsf{tight} bound.
In Section \ref{S3}, we answer question (3), where we unify various \textsf{CR}-type bounds via conic linear programming with various cones on a space that is the tensor product of real symmetric matrices and complex Hermitian matrices. 
Here, we reformulate the \textsf{tight} bound, (i.e.Theorem \ref{TH2}), the \textsf{HN} bound (Theorem \ref{TH99}) and the \textsf{NH} bound (Theorem \ref{TH99}). 
In Section \ref{S4}, we develop our theory of how to calculate 
upper and lower bounds on the \textsf{tight} bound 
by applying an \textsf{SDP} with constraints labeled by unit vectors from a real vector space. 
We also derive an \textsf{SDP} with optimal value equal to the upper bound, and which is directly optimized over uncorrelated measurement strategies. 
From the solution for this \textsf{SDP}, we derive a concrete uncorrelated measurement strategy that has its precision given by our upper bound to the \textsf{tight} bound. Since our two-sided bounds to the \textsf{bound} can be tight, 
we therefore have derived a near-optimal measurement strategy
We thereby can compute optimal uncorrelated measurement strategies in the asymptotic limit, and this answers question (2).
In Section \ref{S5}, we construct a strategic subset of 
the above unit vectors from design theory. 
In Section \ref{S6}, based on the results of Section \ref{S5}, we give the calculation complexity of the \textsf{tight} bound
within an additive error of $\epsilon$.
This thereby answers question (1).
In Section \ref{S7}, we answer question (4), where we describe how we obtain our numerical lower bounds for the \textsf{tight} bound and illustrate its results.
In Section \ref{S-apps}, we explain how quantum multiparameter estimation theory is applicable in learning parameters of Hamiltonian models, 
and also in 3D-field sensing.
In Section \ref{S8}, we discuss our results and its implications in more detail.

\section{Formulation and review of existing results}\label{S2}
\subsection{Various lower bounds for tight CR bound}
We describe the formulation of quantum state estimation, and briefly review existing results on quantum parameter estimation.
We refer readers to references \cite{helstrom,holevo,ANbook,hayashi,petz,Suzuki_2020} for more details.

\begin{table*}[htbp]
    \centering
    \rowcolors{2}{white}{gray!15}
    \label{tab:symbols}
    \begin{tabular}{>{$}c<{$} l}
    \toprule
    \textbf{Symbol} & \textbf{Meaning} \\
    \midrule
    d & the number of parameters to be estimated \\
    \theta = (\theta^1, \dots, \theta^d) & the parameters' true value  \\
     \Theta \subseteq \mathbb R^d & set of all possible parameter vectors \\
    \mathcal M & model $\{ \rho_\theta: \theta \in \Theta \}$  \\
    \mathcal H & the Hilbert space for the quantum probe state  \\
    n & the number of dimensions $\mathcal H$ has \\
    \mathcal T_{sa}(\mathcal H) & trace-class self-adjoint operators on $\mathcal H$ \\
    \mathcal B_{sa}(\mathcal H) & bounded self-adjoint operators on $\mathcal H$ \\
    \mathcal S(\mathcal H) & set of density operators on $\mathcal H$ \\
    \hat \theta= (\hat \theta^1, \dots, \hat \theta^d) & an estimator of $\theta$ \\
    \rho_\theta, \rho & a probe state that is parametrized by $\theta$ \\
    \Pi & a POVM, a measurement \\
    \Pi_x & positive semidefinite operators in $\Pi$ \\
    \hat \Pi = (\Pi, \hat \theta) & an estimator based on $\Pi$ and $\hat \theta$\\
    V_\theta[\hat \Pi],V[\hat \Pi] & mean-square error (MSE) matrix of $\hat \Pi$\\
    {\rm l.u. at }\theta & locally unbiased at $\theta$\\
    C_\theta[G], C[G] & fundamental precision bound: $\min_{\hat \Pi\ {\rm l.u. at }\ \theta}\tr G V_\theta[\hat \Pi]$\\
    a & size $d$ real square matrix\\
    S & Hermitian operator on $
    \mathcal H$\\
    D0 & Maximization of $
    \tr a + 
    \tr S$ subject to \eqref{MML} for all $x \in 
    \mathbb R^d$\\
    (a^{**},S^{**}) & the optimal solution of $D0$\\
    S(D0) = \tr a^{**} + \tr S^{**} & the optimal value of $D0$, (equal to $C[G]$)\\
    J_\theta(\Pi) & Fisher information matrix\\
    \mathcal X & set of labels $x$ corresponding to $\Pi_x$\\
    D_j = \frac{\partial}{\partial \theta^j} \rho_\theta & $j$th partial derivative of $\rho_\theta$\\
    C[G,m] & 
    denotes $C[G]$ when $|\mathcal X| =m$ and $m < \infty$\\
    L_i & SLD corresponding to $D_i$ \\
    C^S[G] & SLD CR bound \\
    C^{HN}[G] & Holevo-Nagaoka (HN) bound \\
    C_\theta^N[G], C^N[G] & 
    Nagaoka bound \\
    C^{NH}[G] & Nagaoka-Hayashi (NH) bound \\
    \bottomrule
    \end{tabular}
    \caption{Notations for Section \ref{S2}. 
    We define the quantum parameter estimation problem in terms of a quantum model $\mathcal M$ using the notation $\theta$, $\Theta$, $\mathcal H$, $d$ and $n$.
The vector $\theta$ encapsulates the true value of the $d$ parameters, and we estimate $\theta$ using an estimator $\hat \theta$ that is obtained from measurements $\Pi$.
The estimator $\hat \Pi$ encapsulates information about both $\Pi$ and how to construct $\hat \theta$ from the measurements. 
For quantum parameter estimation, we like to minimize the MSE $V_\theta[\hat \Pi]$.
For uncorrelated measurement strategies the minimum MSE for locally unbiased $\hat \Pi$ is given by the optimal value $S(D0)$ of a optimization program $D0$.
In fact $S(D0)$ is the tight bound $C[G]$, where $G$ denotes the weight matrix for the $d$-parameter estimation problem.
Evaluating $C[G]$ is non-trivial. 
There is a plethora of upper bounds to the tight bound $C[G]$, given by the SLD CR bound, $ C^S[G]$, the Holevo-Nagaoka (HN) bound  $C^{HN}[G]$,
the Nagaoka bound $C^N[G]$ and the Nagaoka-Hayashi (NH) bound $C^{HN}[G]$.} 
\label{table:notations}
\end{table*}

A {\it quantum system} is represented by a Hilbert space $\cH$. 
In the following, when ${\cal H}$ is finite-dimensional, 
we can ignore the word ``bounded'' and ``trace class'', 
and can replace ``self-adjoint operator'' by Hermitian matrix.
Let ${\cal B}_{sa}({\cal H})$ be the set of bounded self-adjoint operators on ${\cal H}$,
which is a real vector space.
Let ${\cal T}_{sa}({\cal H})$ be the real vector space composed of 
set of trace class self-adjoint operators on ${\cal H}$.
A {\it quantum state} $\rho$ is a positive semi-definite matrix on $\cH$ with unit trace. 
The set of all quantum states on $\cH$ is denoted by $\sofh:=\{\rho\,|\,\rho\ge0,\tr{\rho}=1 \}\subset {\cal T}_{sa}({\cal H})$.

A {\it measurement} $\Pi$ can be represented mathematically as a positive operator-valued measurement (POVM) which is a set of positive semidefinite matrices 
$\Pi=\{\Pi_x\}_{x\in\cX}$ that satisfies a completeness condition.
When the set of measurement outcomes is finite, $\cX$ is a finite set and the completeness condition is $\sum_{x\in\cX}\Pi_x=I$.   
When the set of measurement outcomes is a continuum, such as when $\cX = \mathbb R$, the completeness condition is $\int_{\cX}\Pi_x dx=I$.
For notational simplicity, for any POVM $\Pi$, we employ the notation for discrete-valued POVMs, and we always use $\sum_{x\in\cX}\Pi_x =I$ to represent the completeness condition.

When one performs a POVM $\Pi$ on $\rho$, the \emph{Born rule}
\begin{equation}\label{born}
p_\rho(x|\Pi)=\tr{\big[\rho\Pi_x\big]}
\end{equation}
gives the probability of getting an outcome $x \in \cal X$.
We are interested in a model of quantum parameter estimation given by a parametric family of quantum states on $\cH$: 
\be
\cM:=\{\rho_{\theta}\,|\,{\theta}\in\Theta\}\subset \sofh, 
\ee
where $\Theta\subset\bbr^d$ denotes the set of parameters.
Here, $d$ denotes the number of parameters that we want to simultaneously estimate.
To avoid mathematical subtleties, we impose regularity 
conditions; we require $\rho_\theta$ to be differentiable sufficiently many times,  assume $\del\rho_\theta/\del\theta_i$ to be linearly independent, and for the sake of clarity only consider full-rank states here \footnote{We refer the reader to \cite{fn95,fn99} for problems in the pure-state model.}. 

Now given a measurement $\Pi$ and an estimator $\hat{\bm{\theta}}$, we denote  
$\hat{\Pi}=(\Pi , \hat{{\bm{\theta}}})$ as an {\it estimator}.
We define the mean-square error (MSE) matrix for the estimator $\hat{\Pi}$ as  
\begin{align}\nonumber
V_{{\bm{\theta}}}[\hat{\Pi}]
&=\left[ \sum_{x\in\cX} \tr{
\big[
\rho_{\bm{\theta}}\Pi_x\big]}
({\hat{\theta}^i}(x)-\theta^i)({\hat{\theta}^j}(x)-\theta^j)  \right]\\
&=\left[ E_{\bm{\theta}}\big[({\hat{\theta}^i}(x)-\theta^i)({\hat{\theta}^j}(x)-\theta^j)|\Pi\big]  \right].
\end{align}
where $E_{\bm{\theta}}[f(X) |\Pi]$ denotes the expectation of a random variable $f(X)$ with respect to the probability distribution $p_{\rho_{\bm{\theta}}}(x|\Pi)=\tr{\rho_{\bm{\theta}}\Pi_x}$ obtained from the Born rule.
In multiparameter quantum metrology, the objective is to 
find an optimal estimator $\hat{\Pi}=(\Pi , \hat{{\bm{\theta}}})$ that in some sense minimizes the MSE matrix.

Since the minimization of an MSE matrix is not properly defined, we seek precision bounds where we minimize $\Tr{[GV_{{\bm{\theta}}}[\hat{\Pi}]]}$, which is the weighted trace of the MSE matrix according to a given positive matrix $G$. 
Here, $G$ is a {\it weight matrix} and quantifies the trade-off between estimating different vector components of the parameter ${\bm{\theta}}$. 
For instance, when $G$ is the size $d$ identity matrix $I_d$, minimizing $\Tr{GV_{{\bm{\theta}}}[\hat{\Pi}]}$ corresponds to minimizing the average variance of estimators.

A problem fundamental to quantum metrology is that of finding the ultimate precision bound under reasonable assumptions on the estimators we use.
We say that an estimator $\hat{\Pi}$ is {\it unbiased} if for all ${\bm{\theta}}= (\theta^1, \dots, \theta^d)\in\Theta$, we have
\[
E_{\bm{\theta}}\big[{\hat{\theta}^i}(X)|\Pi\big]
=\sum_{x\in\cX} {\hat{\theta}^i}(x) 
\tr{\big[\rho_{{\bm{\theta}}}\Pi_x\big]}=\theta^i \quad(\forall i=1,2,\dots,d).
\]
However, such an unbiased estimator invariably does not exist. Hence one often relaxes this unbiasedness condition to a {\it locally unbiased} condition in the neighborhood of a chosen point ${\bm{\theta}}$. To have the locally unbiased condition to hold, we require that for all parameter indices $i,j\in\{1,2,\dots,d\}$, we have the equations 
\begin{align}
E_{\bm{\theta}}\big[{\hat{\theta}^i}(X)|\Pi\big]&
=\sum_{x\in\cX} {\hat{\theta}^i}(x) \tr{\big[\rho_{{\bm{\theta}}}\Pi_x\big]}
=\theta^i, \label{MK}\\ 
\frac{\del}{\del\theta^j}E_{\bm{\theta}}\big[{\hat{\theta}^i}(X)|\Pi\big]&=
\sum_{x\in\cX} {\hat{\theta}^i}(x)
\tr{\Big[\frac{\del}{\del\theta^j}\rho_{{\bm{\theta}}}\Pi_x \Big] }
=\delta_i^j \label{M1}.
\end{align}
Note that we can derive this condition by applying the Taylor expansion to the usual unbiasedness condition at a point ${\bm{\theta}}$ to first order.

Then, we introduce the fundamental precision limit by
\be\label{qcrbound}
C_{\bm{\theta}}[G]:=
\min_{\hat{\Pi}\mathrm{\,:l.u.at\,}{\bm{\theta}}}\Tr{ \big[G V_{\bm{\theta}}[\hat{\Pi}]\big]}, 
\ee
where the minimization is carried out for all possible estimators under the locally unbiasedness condition, 
which is indicated by l.u.~at ${\bm{\theta}}$. 
In this paper, any lower bound for the weighted trace of the MSE matrix $V_{\bm{\theta}}[\hat{\Pi}]$
is referred to as the {\it CR type bound}. 
When a CR type bound equals to the fundamental precision limit $C_{\bm{\theta}}[G]$ as in \eqref{qcrbound}, it is called the
{\it tight CR} bound in our discussion.  
That is, $C_{\bm{\theta}}[G]$ is called the tight CR bound.
In the following, we discuss some CR type and tight CR bounds.

In fact, the set of MSE matrix $V_{\bm{\theta}}[\hat{\Pi}]$
under the local unbiasedness condition is characterized as follows.
\begin{align}
&\{ V_{\bm{\theta}}[\hat{\Pi}]|
\hat{\Pi} \hbox{ is l.u.at }{\bm{\theta}} \}
= \{  J_{\bm{\theta}}(\Pi)^{-1}| \Pi \hbox{ is a POVM }\},
\end{align}
where $J_{\bm{\theta}}(\Pi)$ is the Fisher information matrix
of the distribution family $\{  P_{\bm{\theta},\Pi} \}_{\bm{\theta}}$
and
$ P_{\bm{\theta},\Pi}(x):= \tr [\rho_{\bm{\theta}}\Pi_x]$ \cite[Exercise 6.44]{hayashi2016quantum}.

When multiple copies of the unknown state are prepared, 
only individual measurements for each copy 
are allowed, and the error is measured by weighted sum of mean square error
with the weight matrix $G$,
it is impossible to realize estimation precisions exceeding
$C_{\bm{\theta}}[G]$.
Although the measurement to achieve this bound depends on the true parameter
${\bm{\theta}}$, 
when classical adaptive improvement for the choice of measurement is allowed, 
it is possible to achieve the bound $C_{\bm{\theta}}[G]$
\cite{MH05,PhysRevA.61.042312,Hayashi11,YCH18}.
Therefore, we can consider that 
the tight CR bound $C_{\bm{\theta}}[G]$
expresses the 
ultimate precision of the optimal estimator in the asymptotic limit of infinitely many probe states
when only adaptive individual measurements 
for each copy are allowed.
This setting corresponds to the strategy A2 in Section 3.2 of \cite{Suzuki_2020}.
Furthermore, when $n$ copies of the unknown state are given,
even when 
any separable measurement over the $n$-fold system is allowed,
it is impossible to overcome the precision error $C_{\bm{\theta}}[G]$ 
although there exists a separable measurement that requires quantum correlation over the $n$-fold system
\cite[Exercise 6.42]{hayashi2016quantum}.

In the latter discussion, we focus on this problem only at one point $\bm{\theta} \in \Theta$.
Hence, we omit the subscript $\theta$ later, using $C[G]$ to denote $C_{\bm{\theta}}[G]$. Furthermore,
we simplify $\rho_\theta$, $ \frac{\del}{\del\theta^j}\rho_{{\bm{\theta}}}$, $V_{{\bm{\theta}}}[\hat{\Pi}]$
to $\rho$, $D_j$, and $V[\hat{\Pi}]$, respectively.
Further, when ${\cal H}$ is infinite-dimensional,
$\rho$ and $D_j$ are assumed to be trace-class operators.
Note that, even when we remove \eqref{MK}, the minimum value of 
$\Tr{[G V[\hat{\Pi}]]}$ is not changed due to the following reason.
Given an estimator $\hat{\Pi}=(\Pi , \hat{{\bm{\theta}}})$ satisfying \eqref{M1},
the new estimator $\hat{\Pi}'=(\Pi , \hat{{\bm{\theta}}}-
E_{\bm{\theta}}\big[{\hat{\theta}}(X)|\Pi\big]+\bm{\theta})$
satisfies the relation $V[\hat{\Pi}] = V[\hat{\Pi}']+
(E_{\bm{\theta}}\big[{\hat{\theta}}(X)|\Pi\big]-\bm{\theta})
(E_{\bm{\theta}}\big[{\hat{\theta}}(X)|\Pi\big]-\bm{\theta})^T
\ge V[\hat{\Pi}']$.
Hence, we ignore the condition \eqref{MK}.

We denote the minimum of \eqref{qcrbound} when 
our measurement is limited to measurement with discrete value and 
the number of elements in ${\cal X}$ is $m$
by $C[G,m]$.
Clearly, $C[G,m]\ge C[G,m']$ for $m' \ge m$.
In addition,
since our interest is the minimization \eqref{qcrbound}, 
without loss of generality, we can assume that $\bm{\theta}$ is zero.

To get a CR type bound,
we often focus on the SLD $L_{i}$, which is defined as any Hermitian matrix that satisfies
\begin{align}
D_i= 
\frac{1}{2}
\big(L_{i} \rho+ \rho L_{i}\big).
\label{DEFSLD}
\end{align}
The SLD Fisher information matrix $J=(J_{i,j})$ is defined as
\begin{align}
J_{i,j}:=\frac{1}{2}
\tr{\big[L_{i} \big(L_{j} \rho + \rho L_{j}\big)\big]}.
\label{DEFSLDF}
\end{align}
Here, when $\rho$ is strictly positive, 
the choice of Hermitian matrix $L_{i}$ is unique.
Otherwise, it is not unique.
However, the definition of 
the SLD Fisher information matrix $J$ in \eqref{DEFSLDF}
does not depend on the choice of Hermitian matrix $L_{i}$
under the condition \eqref{DEFSLD}.
Under the locally unbiasedness condition at ${\bm{\theta}}$, 
we have SLD CR inequality \cite{helstrom}
\begin{align}
V[\hat{\Pi}]
\ge J^{-1}.
\label{CRSLDF}
\end{align}
For the proof, see \cite{helstrom,holevo}, \cite[Section 6.6]{hayashi2016quantum},
\cite[Appendix B]{Suzuki_2020} for more details. 
When we can choose SLDs $L_i$ for $i=1, \ldots, d$ 
such that these SLDs $L_{i}$ all commute,
the equality in \eqref{CRSLDF} can be achieved by a local unbiased estimator constructed by their simultaneous spectral decomposition.
In the choice of SLDs $L_{i}$, 
extending the Hilbert space is allowed.
However, when $\rho$ is a strictly positive density matrix,
it is sufficient to check for the commutativity of 
SLDs $L_{i}$ without extending the Hilbert space.
In general, there is a possibility that 
the equality in \eqref{CRSLDF} be achieved
only with an extending Hilbert space.


Taking a weighted trace in \eqref{CRSLDF}, we obtain the following bound.
\begin{itemize}
\item The SLD CR bound, which is the tight CR bound for any one-parameter model \cite{helstrom}:
\be \label{Eq:sld_crbound}
C^{\rm S}[G]:=\Tr{[G J^{-1}]},
\ee
where $J$ denotes the SLD Fisher information matrix about the model $\cM$.
\end{itemize}

To get a tighter bound than SLD bound, 
we focus on the vector of self-adjoint operators 
$\vec{Z}=(Z^1,\ldots, Z^d)$ that satisfies the condition
\begin{align}
\tr{[D_j Z^i]}=\delta_i^j \hbox{ for } 
i,j=1, \ldots, d. \label{AMO}
\end{align}
Then, we define the Hermitian matrix ${\cal Z}(\vec{Z})$ whose ($i,j$) component is
$\tr  \rho Z^i Z^j$.
When an estimator $\hat{\Pi}=(\Pi , \hat{{\bm{\theta}}})$ satisfies the condition:
\begin{align}
\sum_{x\in\cX} {\hat{\theta}^i}(x) \Pi_x=Z^i,\label{ASC}
\end{align}
we have \cite[(6.7.73)]{holevo}
\begin{align}
V[\hat{\Pi}] \ge {\cal Z}(\vec{Z}).
\end{align}
Using this relation, we have \cite{nagaoka89}
\begin{align}
\Tr{G V_{\bm{\theta}}[\hat{\Pi}]}
\ge \Tr{G\Re {\cal Z}(\vec{Z})}+\Tr{|{G}^{\frac12} {\Im} {\cal Z}(\vec{Z}) {G}^{\frac12} |},
\Label{AML}
\end{align}
where the operator $ |X|$ is defined as $ \sqrt{X^\dagger X}$.

Therefore, we obtain the following bound; 

\begin{itemize}
\item 
The Holevo-Nagaoka (\textsf{HN}) bound \cite{holevo,nagaoka89}
(Nagaoka \cite{nagaoka89} derived this bound by using 
several formulas obtained in Holevo \cite{holevo}, and 
called this bound the Holevo bound.
Recently, the reference \cite{Sammy} called this bound the Holevo Cramer Rao bound,
and cited Nagaoka \cite{nagaoka89} as its formulation.
However, this bound does not appear in Holevo \cite{holevo}.
It should be called Holevo-Nagaoka CR bound.
For its detailed history, see the latest review \cite{masahito2023-review}.) 
\be\label{holevo-bound}
C^{HN}[G]:=
\min_{\vec{Z}=(Z^1,\ldots, Z^d)}
\Tr{[G\Re {\cal Z}(\vec{Z})]}+\Tr{[|{G}^{\frac12} {\Im} {\cal Z}(\vec{Z}) {G}^{\frac12} |]},
\ee
where the minimization takes 
the vector of self-adjoint operators  
$\vec{Z}=(Z^1,\ldots, Z^d)$
to satisfy the condition \eqref{AMO}.

Notice that the minimum \eqref{holevo-bound} is achieved 
when the vector of Hermitian matrices $\vec{Z}$ satisfies
the condition $\tr{\rho Z^i}=0$ for $i=1, \ldots, d$.
When the model is composed of pure states,
the equality in inequality $C[G]\ge C^{HN}[G]$
holds \cite{KM}.
\end{itemize}

The \textsf{HN} bound $C^{HN}[G]$ improves the SLD CR bound $C^S[G]$,
and gives the asymptotic limit of precision of the minimum estimation error 
when any quantum correlation is allowed in measurement apparatus
and multiple copies of unknown states are prepared \cite{GK06,HM08,KG09,YFG13,YCH18}.

In addition, the \textsf{HN} bound $C^{HN}[G]$ satisfies the additivity condition, i.e., 
this value with the $m$-copy setting equals the value with the one-copy setting divided by $m$ \cite{HM08}.
The reference \cite[Eq. (11)]{Albarelli2019_PRL} 
derived a calculation formula based on \textsf{SDP} for
the bound $C^{HN}(G)$.

When $d=2$, the lower bound in \eqref{AML} is written as
\begin{align}
& \Tr{[G\Re {\cal Z}(\vec{Z})]}+\Tr{[|{G}^{\frac12} {\Im} {\cal Z}(\vec{Z}) {G}^{\frac12} |
]}\nonumber\\
=&
{G}_{1,1}\tr {[Z^1 \rho Z^1]}+{G}_{2,2}\tr {[Z^2 \rho Z^2]}\nonumber\\
&+{G}_{1,2} \tr{[\rho (Z^1 Z^2+Z^2 Z^1) ] }
\nonumber\\
&+2 \sqrt{\det {G}}|\tr{[ \rho [Z^1,Z^2]  ]}|.\label{AMV}
\end{align}
To improve the \textsf{HN} bound, 
Nagaoka \cite{nagaoka89,nagaoka91} derived the following inequality with $d=2$: 
\begin{align}
 \Tr{[G V_{\bm{\theta}}[\hat{\Pi}]]} 
 \ge &
{G}_{1,1}\tr {[Z^1 \rho Z^1]}+{G}_{2,2}\tr {[Z^2 \rho Z^2]}
\nonumber \\
&+{G}_{1,2} \tr{[\rho (Z^1 Z^2+Z^2 Z^1) ] }
\nonumber\\
&+2 \sqrt{\det{G}}
\tr{[| \rho^{1/2} [Z^1,Z^2] \rho^{1/2} |]}\label{AMV-B}
\end{align}
when an estimator $\hat{\Pi}=(\Pi , \hat{{\bm{\theta}}})$ satisfies the condition \eqref{ASC}.
Since $\tr{[| \rho^{1/2} [Z^1,Z^2] \rho^{1/2} |]}
\ge |\tr{ [\rho [Z^1,Z^2] ] }|$, the inequality \eqref{AMV} implies \eqref{AML}.

Using \eqref{AMV}, we obtain the following bound.
\begin{itemize}
\item The Nagaoka bound \cite{nagaoka89,nagaoka91}, which is given only in the case with $d=2$,
and is tighter than the \textsf{HN} bound.
\begin{align}
\label{nagaoka-bound}
C_{\bm{\theta}}^N[G]:=&
\min_{\vec{Z}=(Z^1, Z^2)}
{G}_{1,1}\tr {[Z^1 \rho Z^1]}+{G}_{2,2}\tr {[Z^2 \rho Z^2]}\nonumber  \\
&+{G}_{1,2} \tr{[\rho (Z^1 Z^2+Z^2 Z^1) ] }
\nonumber\\
&+2 \sqrt{\det{G}} \tr{[| \rho^{1/2} [Z^1,Z^2] \rho^{1/2} |]},
\end{align}
where the minimization takes the vector of self-adjoint operators 
$\vec{Z}=(Z^1, Z^2)$ under the condition \eqref{AMO}.
\end{itemize}
In the qubit case, the Nagaoka bound $C_{\bm{\theta}}^N[G]$ equals 
the tight CR bound $C[G]$ for a two-parameter model ($d=2$).

To generalize the above improvement, we consider the tensor product between 
the Hilbert space ${\cal H}$ and 
the real vector space ${\cal R}'$ spanned by $\{ |1\rangle , \ldots , |d\rangle\}$.
Then, we define ${\cal B}'$ as
\begin{align}
{\cal B}'&:=\Big\{\sum_{j=1}^d \sum_{k=1}^d
\Big| |k\rangle \langle j| \otimes X^{k,j}|
X^{k,j} \in {\cal B}_{sa}({\cal H}), X^{k,j}=X^{j,k} \Big\} .
\end{align}
To generalize \eqref{AMV}, for the vector of self-adjoint operators 
$\vec{Z}=(Z^1,\ldots, Z^d)$,
the paper \cite{Haya} focuses on $\Pi(\vec{Z})\in {\cal B}'$ with components 
$\Pi(\vec{Z})^{i,j}= Z^i Z^j$.
When an estimator $\hat{\Pi}=(\Pi , \hat{{\bm{\theta}}})$ satisfies the condition
\eqref{ASC}, the paper \cite{Haya} derived the inequality
\begin{align}
\Tr{[G V_{\bm{\theta}}[\hat{\Pi}]]}
\ge \min_{X' \in {\cal B}'}\{
\Tr{ [({G} \otimes \rho)  X']}|X' \ge \Pi(\vec{Z})\}.
\end{align}

Using \eqref{AMV}, the paper \cite{CSLA} introduced the following bound.
\begin{itemize}
\item The Nagaoka-Hayashi (\textsf{NH}) bound \cite{CSLA}, which is tighter than 
the \textsf{HN} bound.
\be\label{NH-bound}
C^{NH}[G]:=
\min_{\vec{Z}}
\min_{X' \in {\cal B}'}\{\Tr{ [({G} \otimes \rho)  X']}|X' \ge \Pi(\vec{Z})\},
\ee
where the minimization takes the vector of self-adjoint operators 
$\vec{Z}=(Z^1,\ldots, Z^d)$ under the condition \eqref{AMO}.
\end{itemize}
The reference \cite[(22)]{CSLA} derived a calculation formula based on \textsf{SDP} for the bound $C^{NH}(G)$.
Overall, 
the \textsf{HN} and \textsf{NH} bounds
focused on the relation between 
the difficulty of joint measurement
and the multiparameter estimation.
This kind of relation was also discussed in the recent paper \cite{PhysRevLett.126.120503}.

\subsection{Another expression of tight CR bound \texorpdfstring{by \cite{hayashi97,hayashi97-2}}{}}
To get another form of $C[G]$,
the papers \cite{hayashi97,hayashi97-2} treated the minimization $C[G]$ as a minimization
with respect to a POVM $\hat{\Pi}$ over $\mathbb{R}^d$.
In this case, a POVM $\hat{\Pi}$ is considered as an element of convex cone.
Then, it focused on the following maximization problem.

For a $d\times d $ real matrix $a$ and a self-adjoint operator $S$ on the Hilbert space ${\cal H}$,
we consider the condition:
\begin{align}
(x^T G x) \rho - \sum_{i,j} a_i^j x^i D_j  -S \ge 0 \Label{MML}
\end{align}
for $x=(x^1,\ldots,x^d)\in \mathbb{R}^d$.
We consider the maximization:
\begin{align}
S(D0):= \max_{a,S} \sum_{i} a_i^i  +\Tr S \Label{XLO},
\end{align}
where the maximization takes 
the pair $(a,S)$ to satisfy the condition \eqref{MML}.
Here, $D0$ shows the maximization problem presented in \eqref{XLO},
and the maximum value is denoted by $S(D0)$.
Then, the paper \cite{hayashi97-2} showed the following proposition:

\begin{proposition}[\protect{\cite[Theorem 6]{hayashi97-2}}]\Label{TH1}
\begin{align}
C[G]=S(D0). \Label{AML-B}
\end{align}
\end{proposition}

In fact, the maximization \eqref{XLO} can be regarded as the dual problem 
when the minimization $C[G]$ is considered as a conic linear programming
with respect to a POVM $\hat{\Pi}$ over $\mathbb{R}^d$.
Hence, we call the maximization \eqref{XLO} the dual problem $D0$.

To see this, we focus on the relation
\begin{align}
\int_{\mathbb{R}^d} \Tr \Big( (x^T G x) \rho - \sum_{i,j} a_i^j x^i D_j -S)  \Pi(dx)\Big)
\ge 0,
\Label{MGT}
\end{align}
where
$x=(x^1,\ldots,x^d)$ takes values in $\mathbb{R}^d$.
Using the above relation, we have
\begin{align}
& \int_{\mathbb{R}^d} (x^T G x) \Tr [\rho\Pi(dx) ]\nonumber\\
\ge &\int_{\mathbb{R}^d} \sum_{i,j} a_i^j x^i \Tr [D_j  \Pi(dx)] 
+\int_{\mathbb{R}^d}  \Tr [ S \Pi(dx)] \nonumber\\
= &\sum_{i,j} a_i^j \int_{\mathbb{R}^d}  x^i \Tr [D_j  \Pi(dx)] +\Tr [S \Pi(\mathbb{R}^d)]
\nonumber\\
= &\sum_{i,j} a_i^j \delta_{j}^i +\Tr [S I]
= \sum_{i} a_i^i  +\Tr S .
\end{align}
Hence, we can easily see the inequality $\ge $ in \eqref{AML-B}.

To show the opposite inequality, one needs to show the non-existence of the duality gap, i.e., 
the minimum of the primal problem equals the maximum of the dual problem 
in the framework of a conic linear programming, as explained in Appendix \ref{Ap2}.
Although the non-existence of the duality gap holds
for a conic linear program in a finite-dimensional vector space,
the set of POVMs $\hat{\Pi}$ over $\mathbb{R}^d$ is a subset of an infinite-dimensional  space, 
because the number of elements in $\mathbb{R}^d$ is not finite,
even when ${\cal H}$ is finite-dimensional.
Hence, the paper \cite{hayashi97-2} showed the non-existence of the duality gap 
in this problem setting by discussing a complicated issue related to topological vector space.

In fact, the formula \eqref{AML-B} has various merits.
Appendix \ref{Ap1} summarizes its two applications.
For example, use of the relation \eqref{AML-B} enables us to derive the minimum MSE under the locally unbiasedness condition under the one-parameter case.
As another example, using the relation \eqref{AML-B},
the papers \cite{hayashi97,hayashi97-2} derived 
the tight CR bound $C[G]$ for a three-parameter model ($d=3$) in the qubit system.

\begin{table*}[htbp]
    \centering
    \rowcolors{2}{white}{gray!15}
    \label{tab:symbols}
    \begin{tabular}{>{$}c<{$} l}
    \toprule
    \textbf{Symbol} & \textbf{Meaning} \\
    \midrule
    \mathcal R & real vector space spanned by $|0\>, |1\>,\dots , |d\>$\\
    \mathcal R_C & complex vector space spanned by $|0\>, \dots , |d\>$\\
    \mathcal R' & real vector space spanned by $|1\>,\dots , |d\>$\\
    \mathcal R_C' & complex vector space spanned by $|1\>,\dots , |d\>$\\
    \mathcal M_{rs}(\mathcal R) & set of real symmetric matrices on $\mathcal R$ \\
    \mathcal M_{rs,+}(\mathcal R) & positive semidefinite matrices in $\mathcal M_{rs}(\mathcal R)$ \\
    \mathcal B & Equal to $\mathcal M_{rs}(\mathcal R) \otimes \mathcal B_{sa}(\mathcal H)$\\
    \mathcal T & Equal to $\mathcal M_{rs}(\mathcal R) \otimes \mathcal T_{sa}(\mathcal H)$\\
    \mathcal B', \mathcal B'' & both are extensions of $\mathcal B$\\
    \mathcal B''_\rho & $\rho$-dependent subspace of $\mathcal B''$\\
    \mathcal S_P, \mathcal S_{SEP}& cones in $\mathcal B$ \\
    X = \Pi(a,S) & an operator in $\mathcal T$ given by \eqref{def:X.AS}  \\
    \mathcal S^*_P, \mathcal S^*_{SEP}& dual cones\\
    P1,P2 &
    conic programs with cones 
    $\mathcal S_{SEP}$, $\mathcal S_P$\\
    \mathcal S_3 = \mathcal S(\mathcal R_C \otimes H)_{PPT} &  PPT cone \\
    \mathcal S_4 = \mathcal S(\mathcal R_C \otimes H)_P & 
    cone $S(\mathcal R_C \otimes H)_P$ \\
    Pi, \quad i=3,4 &
    conic program with variable in $\mathcal B''$ and $\mathcal S_i$\\
    P5 & conic program with variable in $\mathcal B''_\rho$ and $\mathcal S_4$\\ 
    S(Pi) & optimal value of $Pi$, for $i=1,2,3,4,5$\\
    Di     & dual program of $Pi$ for $i=1,2,3,4,5$\\
\mbox{Objective function of }Di & $\tr a + \tr S$\\
\mbox{Constraints of }Di & $\Pi(a,S)$ lies in appropriate dual cone\\
    S(Di)             &     optimal value of $Di$, for $i=0,1,2,3,4,5$\\
    \mathcal W_R & a finite set of norm 1 vectors in $\mathcal R$ \\
    \mathcal W_H & a finite set of norm 1 vectors in $\mathcal H$ \\
{[}P1,\mathcal W_R{]} & Primal SDP \\
    {[}D1,\mathcal W_R{]} & Dual SDP of $[P1,\mathcal W_R]$ \\
    S{[}P1,\mathcal W_R{]} & optimal value of $[P1,\mathcal W_R]$ \\
    S{[}D1,\mathcal W_R{]} & optimal value of $[D1,\mathcal W_R]$ \\
	C_2(a) & real number depending on $a$, $D_j$ and $
\rho$ (see \eqref{def:C2a})\\
	C_1 & real number depending on $D_j$ and $
\rho$ (see \eqref{AAS})\\
\delta(\mathcal W_R) & covering radius of $\mathcal W_R$\\
\kappa & related to the covering radius of $\mathcal W_R$\\
    \bottomrule
    \end{tabular}
    \caption{Notations from Section \ref{S3} onwards.
For our conic programming framework, we consider primal conic programs $P1,P2,P3,P4,P5$ and their dual conic programs $D1,D2,D3,D4,D5$. Primal programs are minimizations over cones, while dual programs are maximizations over the dual cones (also cones).
    There is no duality gap between the primal programs and the dual programs, so $Pi$ has the same optimal value as $Di$, so that $S(Pi) = S(Di)$.
The optimizations for the primal programs are over variables of the form $X$, which are matrices in the tensor product space $\mathcal R_C \otimes \mathcal H$. 
The primal conic programs all have objective functions $\tr (G\otimes \rho) X$ and the dual conic programs all have objective functions $\tr a + \tr S$, where $a$ is a size $d$ real matrix, and $S$ is a Hermitian matrix.
Now $S(P1)$ and $D(P1)$ are equal to the tight bound $C[G]$, and can be approximated using the semidefinite program (SDP) $[P1,\mathcal W_R]$, where we can take $\mathcal W_R$ to be a collection of points on a norm-1 hypersphere in a real vector space. 
Namely $C[G] = S(P1) = D(P1) \approx S[P1,\mathcal W_R]$. 
The smaller the covering radius $\delta(\mathcal W_R)$, the better $S[P1,\mathcal W_R]$ approximates 
$C[G]$. 
To approximate $C[G]$, for a chosen set $\mathcal W_R$, we solve the primal SDP $[P1,\mathcal W_R]$ and the dual SDP $[P1,\mathcal W_R]$, denoting the optimal solutions as $X^*$ and $(a^*,S^*)$ respectively. 
Using $(a^*,S^*)$ we can calculate $C_2(a^{*})$. 
Then we find that (in Theorem \ref{thm:D0-lower-bound}) $ S[D1,\mathcal W_R] - n \|X^*\|(1+C_2(a^{*})^2 )\delta(\mathcal W_R) \le C[G] \le  S[D1,\mathcal W_R]$.}
\label{table:notations}
\end{table*}

\section{Conic programming with various cones on tensor product}
\label{S3}
When ${\cal H}$ is finite-dimensional,
Proposition \ref{TH1} guarantees that the tight CR bound
$C[G]$
is given as the maximum value $S(D0)$ with the variables $(a,S)$ of finite-dimension.
However, even with a finite-dimensional space ${\cal H}$,
the calculation of the maximum value $S(D0)$ is not so easy,

To get a more computable form for $C[G]$,
we introduce the real vector space ${\cal R}$
spanned by $\{|0\rangle, |1\rangle , \ldots , |d\rangle\}$.
Let ${\cal M}_{rs}({\cal R})$ be the set of real symmetric matrices on ${\cal R}$.
The weight matrix $G$ can be considered as an element of ${\cal M}_{rs,+}({\cal R})\subset {\cal M}_{rs}({\cal R})$, 
where ${\cal M}_{rs,+}({\cal R})$ denotes the set of positive semidefinite matrices in ${\cal M}_{rs}({\cal R})$.
Here, we denote $G = \sum_{1\le j,i \le d}G_{i,j}|i\rangle \langle j|$.

This section aims to derive various CR bounds 
as conic linear programming problems on the vector spaces 
${\cal B}:= {\cal M}_{rs}({\cal R})\otimes {\cal B}_{sa}({\cal H}) $ and
${\cal T}:= {\cal M}_{rs}({\cal R})  \otimes {\cal T}_{sa}({\cal H})$.
That is, ${\cal B}$ is defined as
\begin{align}
{\cal B}&:=\Big\{\sum_{j=0}^d \sum_{k=0}^d
|k\rangle \langle j| \otimes X^{k,j} \Big|
X^{k,j} \in {\cal B}_{sa}({\cal H}), X^{k,j}=X^{j,k} \Big\} .
\end{align}
${\cal T}$ is similarly defined.

Since any real symmetric matrix can be regarded as a Hermitian matrix on 
the complexified space ${\cal R}_C$ of ${\cal R}$,
any element of ${\cal B}$
can be considered as a self-adjoint operator on ${\cal R}_C\otimes {\cal H}$.
Then, we define several cones in the space ${\cal B}$ as follows.
Let ${\cal B}_{sa,+}({\cal H})$ be the set of bounded positive semi-definite operators on ${\cal H}$.
Then, we define the cone ${\cal S}_{SEP}$ in ${\cal B}$
as 
${\cal S}_{SEP}:=conv ({\cal M}_{rs,+}({\cal R}) \otimes {\cal B}_{sa,+}({\cal H}))$,
where $conv$ is the convex hull.
Then, we define the cone ${\cal S}_{SEP}^*$ in ${\cal B}$ as the dual cone of
 ${\cal S}_{SEP}$.
That is,
\begin{align}
{\cal S}_{SEP}^*&:=\{X \in {\cal T}| \Tr XY\ge 0 \hbox{ for } Y \in   {\cal S}_{SEP}\} .
\end{align}
In the context of entanglement theory,
${\cal S}_{SEP}^*$ is often called entanglement
witnesses \cite{HORODECKI19961,TERHAL200161,Lewenstein}.

Also, we define the cone ${\cal S}_{P}$ as 
\begin{align}
{\cal S}_P&:=\{X \in {\cal B}| \langle v| X|v\rangle \ge 0 \hbox{ for } v \in  {\cal R}_C\otimes {\cal H}\} .
\end{align}
Then,  we have
\begin{align}
{\cal S}_P^*
&:=\{X \in {\cal T}| \Tr XY  \ge 0 \hbox{ for } Y \in {\cal S}_P\} \nonumber \\
&\supset \{X \in {\cal T}| \Tr [X |v\rangle\langle v| ]\ge 0 \hbox{ for } v \in  {\cal R}_C\otimes {\cal H}\} \nonumber \\
&=\{X \in {\cal T}| \langle v| X |v\rangle \ge 0 \hbox{ for } v \in  {\cal R}_C\otimes {\cal H}\}.
\end{align}
Hence, 
when ${\cal H}$ is finite-dimensional,
we have
\begin{align}
{\cal S}_P\subset  {\cal S}_P^*.
\end{align}

Relaxing the condition of ${\cal B}$
we extend the space ${\cal B}$ as
\begin{align}
{\cal B}''&:=\Big\{\sum_{j=0}^d \sum_{k=0}^d
|k\rangle \langle j| \otimes X^{k,j} \Big|
X^{k,0} \in {\cal B}_{sa}({\cal H}), X^{k,j}=(X^{j,k})^\dagger \Big\} .
\end{align}
We define the set 
${\cal S}({\cal R}_C\otimes {\cal H})_P$
(${\cal S}({\cal R}_C\otimes {\cal H})_{PPT}$)
of positive semi-definite (positive partial transpose) self-adjoint operators on ${\cal R}_C\otimes {\cal H}$.
Since any component of an element $X$ of ${\cal S}_{P}$ is a self-adjoint operator,
the element $X$ satisfies the conditions of ${\cal B}''$.
Also, the operator transposed on the system ${\cal R}_C$ of $X$ 
is positive.
Hence, the element $X$ satisfies the conditions of ${\cal B}''$
and ${\cal S}({\cal R}_C\otimes {\cal H})_{PPT}$.
Therefore, since 
the set ${\cal S}_{SEP}$ is generated by separable pure states in 
${\cal R}_C\otimes {\cal H}$, we find the relations
\begin{align}
{\cal S}_{SEP} \subset {\cal S}_P
\subset {\cal S}({\cal R}_C\otimes {\cal H})_{PPT} \cap {\cal B}''
\subset {\cal S}({\cal R}_C\otimes {\cal H})_P \cap {\cal B}''
\Label{MBP}.
\end{align}

To derive various CR bounds as conic linear programming problems on the vector spaces, 
we discuss the relation between our estimator and ${\cal S}_{SEP}$.
For any estimator $(\Pi,\hat{\theta}):=
(\{\Pi_x\}_{x \in \mathcal X} , \hat{\theta}(x) )$, we define 
the operator $X(\Pi,\hat{\theta}) \in {\cal S}_{SEP}$ as
\begin{align}
&X(\Pi,\hat{\theta})\nonumber \\
:=& \sum_{
{x\in\mathcal X}
}
\Big(|0\rangle+ \sum_{i=1}^d \hat{\theta}^i(
{x}
) |i\rangle\Big) 
\Big(\langle 0|+ \sum_{i=1}^d \langle i| \hat{\theta}^i(
{x}
)\Big)
\otimes 
\Pi_{x}.
\end{align}
Although the references \cite{Miguel,Demkowicz}
considered a matrix with the separable form from a POVM,
our matrix $X(\Pi,\hat{\theta})$ with the separable form  
is different from their matrix with the separable form
in the following point.
Their matrix with the separable form
is composed of the tensor product of the POVM element and 
the density matrix of their guess.
We consider the tensor product of the POVM element and 
the one-dimensional operator given by the superposition of  
$|0\rangle$ and the parameter of our estimate. 
The component $|0\rangle\langle 0|$ enables us to
check the condition for the POVM.

Since $\Pi$ satisfies the condition of POVM, 
$X(\Pi,\hat{\theta})$ satisfies the following condition:
\begin{align}
\Tr_{{\cal R}} \big[(|0\rangle \langle 0| \otimes I_{{\cal H}}) 
X(\Pi,\hat{\theta})\big]&=I_{{\cal H}} \Label{c1}.
\end{align}
The condition \eqref{M1} for a locally unbiased estimator guarantees 
the following condition:
\begin{align}
\Tr \Big[(
\frac{1}{2}(|0\rangle \langle i| +|i\rangle \langle 0 |)
\otimes D_j )  X(\Pi,\hat{\theta})\Big] &=\delta_{i,j}  \Label{c2} . 
\end{align}
Under the above condition,
the objective function $\Tr{G V_{\bm{\theta}}[\hat{\Pi}]}$
is rewritten 
with $X(\Pi,\hat{\theta})$ as
\begin{align}
\Tr \big[( G \otimes \rho) X(\Pi,\hat{\theta}) \big]\Label{o1}.
\end{align}
In this notation, the components $G_{0,0},G_{0,j}$, and $G_{j,0}$ are zero.
Therefore, we consider the minimization:
\begin{align}
S(P1):= \min_{X \in {\cal S}_{SEP}}
\big\{\Tr \big[( G \otimes \rho) X \big] \big|
\eqref{c1},\eqref{c2} \hbox{ hold.}
\big\}\Label{o1-B},
\end{align}
and 
we denote the above primal conic linear programming problem by $P1$,
and the minimum value is denoted by $S(P1)$.
Although the above discussion shows the inequality
$C[G] \ge S(P1)$,
we have the following theorem.

\begin{theorem}\Label{TH2}
We have
\begin{align}
C[G]=S(P1)\Label{CSD}.
\end{align}
\end{theorem}
The above theorem shows the tight CR bound can be calculated by solving 
the primal conic linear programming $P1$.
Since the tight CR bound equals the \textsf{tight} bound,
the primal conic linear programming $P1$ characterizes the \textsf{tight} bound.

\begin{proof}
To show the opposite inequality $C[G] \le S(P1)$,
we assume that an element 
$X \in {\cal S}_{SEP}$ satisfies 
the conditions
\eqref{c1} and \eqref{c2}.
$X$ is written as
\begin{align}
X=
\sum_{
{x\in\mathcal X}
} T(
{x}
) \otimes M_{x}
= 
\sum_{
{x\in\mathcal X_1}
} T(
{x}
) \otimes M_{x}
+
\sum_{
{x\in\mathcal X_0}
} T(
{x}
) \otimes M_{x}
,\Label{ACL}
\end{align}
where
$\mathcal X_0:= \{ x \in \mathcal X| 
\langle 0|T(x)|0\rangle=0 \}$ and $\mathcal X_1:=
\mathcal X \setminus \mathcal X_0$.
Then, for $x \in \mathcal X_0$, 
$T(x)|0\rangle=0$ and $\langle 0|T(x)=0$.
Hence, 
the second term in \eqref{ACL} does not contribute 
the conditions \eqref{c1} and \eqref{c2}.
Thus,
the first term in \eqref{ACL} satisfies
the conditions \eqref{c1} and \eqref{c2}.
Since
\begin{align}
& \Tr \Big[( G \otimes \rho) 
\sum_{
{x \in \mathcal X_1}
} T(
{x}
) \otimes M_{x} \Big]\nonumber \\
\le 
& \Tr \Big[( G \otimes \rho) 
\Big(
\sum_{
{x \in \mathcal X_1}
} T(
{x}
) \otimes M_{x} 
+
\sum_{
{x\in\mathcal X_0}
} T(
{x}
) \otimes M_{x} \Big)
\Big]\nonumber \\
=&
\Tr \big[ ( G \otimes \rho) X\big],
\end{align}
we have \eqref{CSD}.
\end{proof}

When ${\cal H}$ is finite-dimensional and 
$\dim {\cal H}=n$, 
the dimension of ${\cal B}$ is $\frac{d(d+1)}{2}\cdot n^2$.
Due to Carathéodory's theorem,
any element of ${\cal S}_{SEP}$
can be written as a convex sum of $\frac{d(d+1)n^2}{2}+1 $ extremal elements of 
${\cal S}_{SEP}$. Then, we have the following lemma.

\begin{lemma}
When $\dim {\cal H}=n$, we have
$C[G,\frac{d(d+1)n^2}{2}+1]=C[G]$.
\end{lemma}

When we replace the cone ${\cal S}_{SEP}$ by another cone
in the primal conic linear program $P1$,
we obtain another primal conic linear program.
The relation \eqref{MBP} lists various alternative cones.
For instance, 
when the condition 
$X \in {\cal S}_{SEP}$ is replaced by $X \in {\cal S}_{P}$,
we denote the primal conic linear programming problem as $P2$.
In fact, since the cone ${\cal S}_{P}$ is given as the set of
positive semi-definite matrices in ${\cal B}$,
$P2$ can be simply solved by semi-definite programming.

When the condition 
$X \in {\cal S}_{SEP}$ is replaced by 
$X \in {\cal S}({\cal R}_C\otimes {\cal H})_{PPT}\cap {\cal B}''$,
($X \in {\cal S}({\cal R}_C\otimes {\cal H})_P\cap {\cal B}''$),
we denote this primal conic linear programming problem as $P3$ ($P4$).
When we denote the minimum of the primal conic linear programming problem $Pl$
by $S(Pl)$ for $l=2,3,4$.
The relation \eqref{MBP} implies 
\begin{align}
S(P1) \ge S(P2) \ge S(P3) \ge S(P4).
\end{align}

Then, the NH bound $C^{NH}[G]$ and the \textsf{SLD} bound $C^{S}[G]$
can be calculated by solving 
$S(P2)$ and $S(P4)$, respectively, as follows.
\begin{theorem}\Label{TH99}
We have
\begin{align}
C^{NH}[G]&=S(P2) \Label{THR1}\\
C^{S}[G]&=S(P4) .\Label{THR2}
\end{align}
\end{theorem}

The above theorem shows that the \textsf{NH} and \textsf{SLD} bounds are 
the optimal values of the primal conic linear programs 
$P2$ and $P4$.
Their difference lies in the choice of their cones.
The cone of $P2$ is composed of positive semi-definite matrices in ${\cal B}
= {\cal M}_{rs}({\cal R})\otimes {\cal B}_{sa}({\cal H})$,
and the cone of $P4$ is composed of 
all positive semi-definite matrices on ${\cal R}_C\otimes {\cal H}$ that are also in the cone $\mathcal B''$.

Although the minimizations $P2,P3,P4$
are defined in cones different from the cone ${\cal S}({\cal R}_C\otimes {\cal H})_P  $,
these minimizations $P2,P3,P4$ can be considered as 
semi-definite programming
(SDP) in the cone ${\cal S}({\cal R}_C\otimes {\cal H})_P  $
because their cones are given as 
linear constraints in 
the cone ${\cal S}({\cal R}_C\otimes {\cal H})_P  $.
Hence, the \textsf{NH} and \textsf{SLD} bounds 
$C^{NH}[G]$ and $C^{S}[G]$ can be solved by semi-definite programming
while the reference \cite[(22)]{CSLA} already gave the same \textsf{SDP} form as $S(P2)$ for $C^{NH}(G)$.

\begin{figure}
         \centering
    \includegraphics[width=8cm]{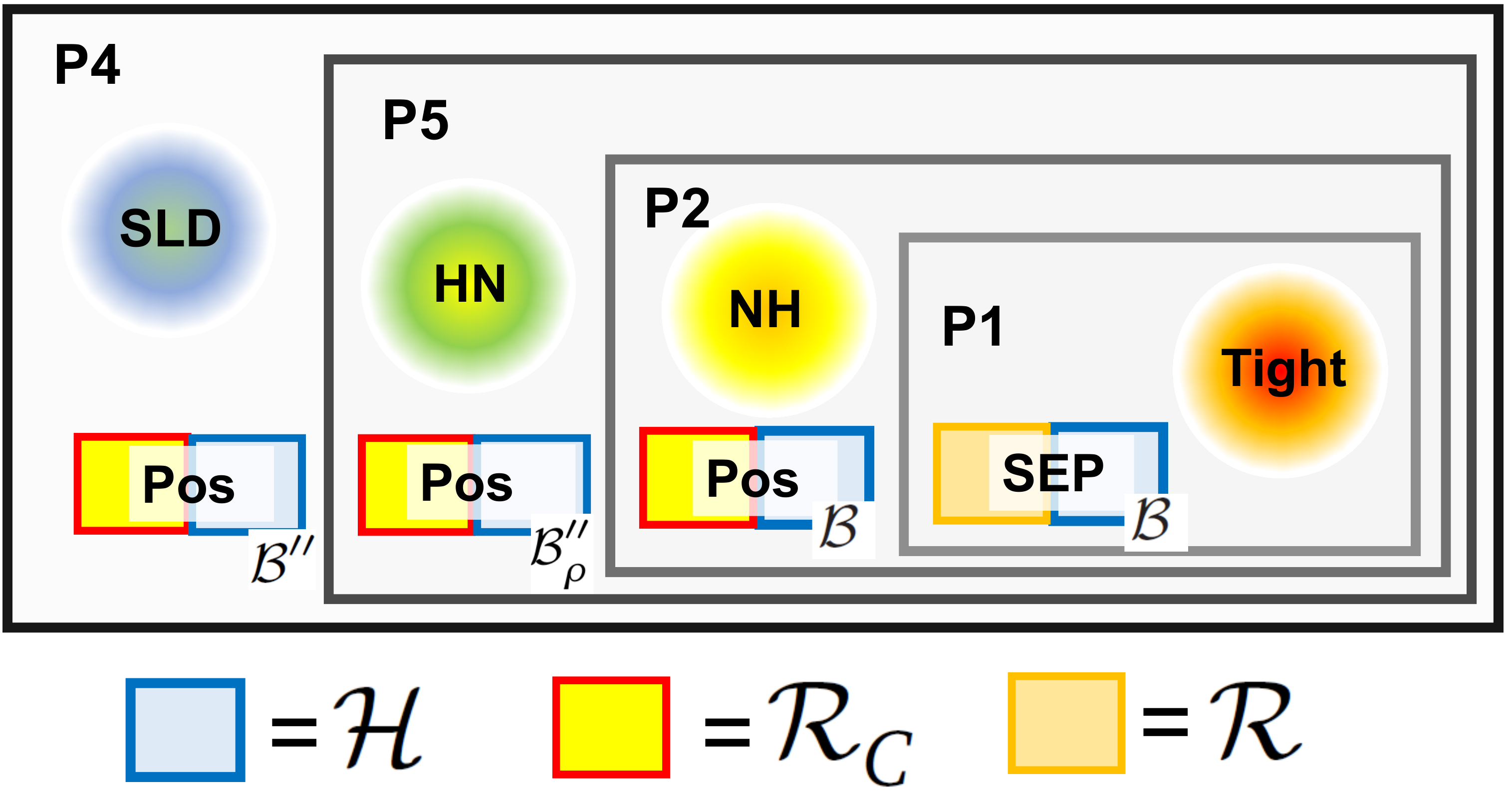}
    \caption{$P1,P2,P4$ and $P5$ are all conic linear programs on the space of operators $\mathcal B$, ${\cal B}''_\rho$, or $\mathcal B''$. The objective functions of all of these conic programs is the same, and the only difference between them lies in the constraints that their operator variables must satisfy. The operator spaces for each of these conic programs are different; P1 is optimized over a bipartite separable space on the tensor product of $\mathcal R$ and $\mathcal H$. All the other spaces use the complexification of $\mathcal R$. The possible choices of $X$ in
the programs $P1,P2,P4$ and $P5$ are contained within one another, which we depict visually as nested rectangles. Because of this, we have $S(P1)\ge S(P2)\ge S(P5)\ge S(P4)$.}
    \label{conic}
\end{figure}


\begin{proof} 
First, we show \eqref{THR1}.
For $X' \in {\cal B}'$ and 
$\vec{Z}=(Z^1, \ldots, Z^d)$ satisfying $ \Tr Z^i D_j=\delta_{j}^i $,
we define $X(X',\vec{Z}) \in {\cal B}$ as follows.
\begin{align}
X(X',\vec{Z}):= X'+ \sum_{j=1}^d ( | 0\rangle \langle j|+ | j\rangle \langle 0|) \otimes Z^j
+ | 0\rangle \langle 0| \otimes I.
\end{align}
Since any element in ${\cal B}$ can be written in the above form due to the definition of ${\cal B}$,
it is sufficient to discuss $S(P2)$ to consider the matrix with the form
$X(X',\vec{Z}) $.

Any element of ${\cal R}_C \otimes {\cal H}$
can be written as the form 
$|(y,z)\rangle:=|y\rangle \oplus (|0\rangle\otimes |z\rangle)$ with 
$|y\rangle (:= \sum_{j=1}^d |j \rangle \otimes |y_j\rangle)\in  {\cal R}_C' \otimes {\cal H}$
and $| z\rangle \in {\cal H}$, where $|y_j \rangle\in {\cal H}$.
Then, we have
\begin{align}
&\langle (y,z) |X(X',\vec{Z})| (y,z) \rangle \nonumber\\
=& \langle y |X'| y \rangle
+\|z\|^2
+\langle z |\sum_{j=1}^d Z^j y_j \rangle
+\langle \sum_{j=1}^d Z^j y_j |z\rangle \nonumber\\
=& \langle y |X'| y \rangle
+\big\| z +\sum_{j=1}^d Z^j y_j \big\|^2
-\big\| \sum_{j=1}^d Z^j y_j \big\|^2 \nonumber\\
=&\langle y |X'-  \Pi(\vec{Z})| y \rangle
+\big\| z +\sum_{j=1}^d Z^j y_j \big\|^2.
\end{align}
Considering the case with $z=-\sum_{j=1}^d Z^j y_j$, 
as pointed by \cite[(22)]{CSLA},
we find the equivalence between 
the following two conditions for $X'$ and $\vec{Z}$.
\begin{itemize}
\item[(i)] 
The inequality $\langle (y,z) |X(X',\vec{Z})| (y,z) \rangle \ge 0$ holds for any $|(y,z)\rangle\in {\cal R}_C \otimes {\cal H}
$.
\item[(ii)] 
The inequality $\langle y |X'- \Pi(\vec{Z})| y \rangle\ge 0 $ holds for any $
|y\rangle \in  {\cal R}_C' \otimes {\cal H}$.
\end{itemize}
Since the condition (i) is the conditions for $P2$ and 
the condition (ii) is the conditions for the NH bound $C^{NH}[G]$,
the desired statement is obtained.

Next, we proceed to the proof of \eqref{THR2}.
We choose $X'$ as an element of
${\cal S}({\cal R}_c'\otimes {\cal H})_P\cap {\cal B}''$.
When $\vec{Z} $ is fixed,
the optimum $X'$ under the condition $X(X',\vec{Z}) \in {\cal S}({\cal R}_C\otimes {\cal H})_P$
is the operator  $\Pi(\vec{Z})$.
In this case, the value of the objective function $\Tr [(G\otimes \rho) X]$
is $\Tr [G  \mathcal{Z}(\vec{Z})]$.
Under the condition \eqref{AMO},
we can show
 $\mathcal{Z}(\vec{Z})\ge J^{-1}$
 in the same way as the SLD Cram\'{e}r-Rao inequality.
When we choose a suitable $\vec{Z}$, the above inequality becomes an equality.
Thus, the minimum of $\Tr [G  \mathcal{Z}(\vec{Z})]$
under the condition \eqref{AMO} is $\Tr [G J^{-1}]=C^{S}[G]$.
Therefore, we obtain \eqref{THR2}.
\end{proof}

To get a relation with \textsf{HN} bound, $C^{HN}(G)$,
we introduce a linear constraint to the operator $X \in {\cal B}''$ as
\begin{align}
\Tr \big[X( ( |j \rangle \langle i |-  |i \rangle \langle j |)\otimes \rho )\big]
=0\Label{c3}
\end{align}
for $i,j=1,2, \ldots, d$.
Using this linear constraint, we define the subspace ${\cal B}_\rho''$ of ${\cal B}''$
as
\begin{align}
{\cal B}_\rho'':=\{X \in {\cal B}''|
\eqref{c3} \hbox{ holds.}
\}.
\end{align}
We consider the minimization:
\begin{align}
&S(P5):= \min_{X \in {\cal S}({\cal R}_C\otimes {\cal H})_P\cap {\cal B}''_\rho}
\big\{\Tr \big[( G \otimes \rho) X \big] \big|
\eqref{c1},\eqref{c2} \hbox{ hold.}
\big\} \nonumber \\
&= \min_{X \in {\cal S}({\cal R}_C\otimes {\cal H})_P\cap {\cal B}''}
\big\{\Tr \big[( G \otimes \rho) X\big] \big|
\eqref{c1},\eqref{c2},\eqref{c3} \hbox{ hold.}
\big\}\Label{o1-B2}.
\end{align}
Since we have the relation
\begin{align}
 {\cal S}_P
\subset {\cal S}({\cal R}_C\otimes {\cal H})_P \cap {\cal B}''_\rho
\subset {\cal S}({\cal R}_C\otimes {\cal H})_P \cap {\cal B}'',
\Label{MBPC}
\end{align}
we have the following relations
\begin{align}
S(P4) \le S(P5) \le S(P2).
\end{align}
Then, we have the following theorem.
\begin{theorem}\Label{H-bound}
We have
\begin{align}
C^{HN}[G] =S(P5) .
\Label{THR4}
\end{align}
\end{theorem}

The minimization $P5$ can be considered as an \textsf{SDP} in the cone ${\cal S}({\cal R}_C\otimes {\cal H})_P  $
because the cone is given as 
linear constraints in 
the cone ${\cal S}({\cal R}_C\otimes {\cal H})_P  $
in the same way as the minimizations $P2,P4$.
Therefore, 
the \textsf{tight}, \textsf{NH}, \textsf{SLD}, and \textsf{HN} bounds
are summarized as Fig. \ref{conic}.
Although the reference \cite[Eq. (11)]{Albarelli2019_PRL} derived an \textsf{SDP} form of the \textsf{HN} bound,
their \textsf{SDP} is different from $P5$.
They assumed the finite-dimensional system for ${\cal H}$ in the derivation \cite[Eq. (11)]{Albarelli2019_PRL}. In contrast,
we do not require this assumption to derive the equation \eqref{THR4}.

\begin{proof}
We show \eqref{THR4} by using the second expression of $S(P5)$ in \eqref{o1-B2}.
Since $X(X',\vec{Z}) \in {\cal S}({\cal R}_C\otimes {\cal H})_P \cap {\cal B}''$,
the components of $\vec{Z} $ are self-adjoint.
First, we fix $\vec{Z} $.
When the condition \eqref{c3} holds,
the condition 
$X(X',\vec{Z}) \in {\cal S}({\cal R}_C\otimes {\cal H})_P \cap {\cal B}''$
is rewritten as follows.
\begin{align}
X'_*:=&X' -
\Big( \Pi(\vec{Z}) - \Big(\sum_{1\le i,j\le d } \frac{\Tr \big[\rho [X^i, X^j]\big]}{2} |i\rangle \langle j|\Big)\otimes I \Big) \nonumber\\
\ge & \Big(\sum_{1\le i,j\le d } \frac{\Tr \big[\rho [X^i, X^j]\big]}{2} |i\rangle \langle j|\Big)\otimes I .\Label{NNM}
\end{align}
Since
$\Big( \Pi(\vec{Z}) - \sum_{1\le i,j\le d } \frac{\Tr [ \rho [X^i, X^j]]}{2} |i\rangle \langle j|\otimes I \Big)$
satisfies the condition \eqref{c3},
when $X'_*$ satisfies the condition \eqref{c3},
when ${X}'$ satisfies the condition \eqref{c3}.
The pair of the condition \eqref{c3} and the condition 
$X(X',\vec{Z}) \in {\cal S}({\cal R}_C\otimes {\cal H})_P\cap {\cal B}''$ 
is rewritten as
the pair of the condition \eqref{c3} for $X'_*$ and the condition
\eqref{NNM}.

When $X'_*$ satisfies \eqref{NNM},
we have 
$\Tr_{\cal H}\rho X'_* \ge \sum_{1\le i,j\le d } \frac{\Tr [\rho [X^i, X^j]]}{2} |i\rangle \langle j|$.
Conversely,
when a $d \times d$ symmetric matrix $V'$ satisfies the condition 
$V' \ge \sum_{1\le i,j\le d } \frac{\Tr [\rho [X^i, X^j]]}{2} |i\rangle \langle j| $,
$X'_*= V' \otimes I$ satisfies \eqref{NNM}.

Since $\Tr \big[(G\otimes \rho) X'_*\big]=\Tr \big[G (\Tr_{\cal H}\rho X'_*)\big]$,
our minimization with fixed $\vec{Z} $
is rewritten as
\begin{align}
&\min_{V': {\rm symmetric}}
\Big\{ \tr G V'
\Big|V' \ge \sum_{1\le i,j\le d } \frac{\Tr \big[\rho [X^i, X^j]  \big]}{2} |i\rangle \langle j|\Big\} \nonumber\\
=&\min_{V': {\rm symmetric}}
\big\{ \tr G V'
\big|V' \ge \sqrt{-1}{\Im} {\cal Z}(\vec{Z}) \big\}, 
\end{align}
which equals $\Tr{\big[|{G}^{\frac12} {\Im} {\cal Z}(\vec{Z}) {G}^{\frac12} |\big]}$.

Since 
\begin{align}
&\Tr \Big[(G\otimes \rho) \Big( \Pi(\vec{Z}) - \Big(\sum_{1\le i,j\le d } 
\frac{\Tr \big[\rho [X^i, X^j]\big]}{2} |i\rangle \langle j|\Big)\otimes I \Big) \Big]\nonumber \\
=& \Tr \big[(G\otimes \rho)  \Pi(\vec{Z})\big]
=\Tr{\big[G\Re {\cal Z}(\vec{Z})\big]},
\end{align}
the minimum value of the objective function with fixed $\vec{Z}$
is $\Tr{[G\Re {\cal Z}(\vec{Z})]}+\Tr{[|{G}^{\frac12} {\Im} {\cal Z}(\vec{Z}) {G}^{\frac12} |]}$.
Since the condition for $\vec{Z}$ is the same as the \textsf{HN} bound, 
we obtain \eqref{THR4}.
\end{proof}

Now, for a $d\times d $ real matrix $a$ and a self-adjoint operator $S$ on the Hilbert space ${\cal H}$,
 let us define the operator as an element of ${\cal T}$:
\begin{align}
    \Pi(a,S) := &
G \otimes \rho - \frac{1}{2} \Big(\sum_{1 \le i, j \le d} a_i^j 
( | 0\rangle \langle i|+ | i\rangle \langle 0|) \otimes D_j \Big) \nonumber \\
&- | 0\rangle \langle 0| \otimes S .
\Label{def:X.AS}
\end{align}
The dual problem $D1$ of $P1$ is given as the following maximization:
\begin{align}
S(D1):=\max_{(a,S) \in \mathbb{R}^{d\times d}\times {\cal T}_{sa}({\cal H}) } 
\Big\{\sum_{i}a_i^i+ \Tr S\Big|\Pi(a,S) \in {\cal S}_{SEP}^* \Big\}.
\Label{NYI}
\end{align}
Also, the dual problem $D2$ of $P2$ is 
 given as the following maximization:
\begin{align}
S(D2):=\max_{(a,S) \in \mathbb{R}^{d\times d}\times {\cal T}_{sa}({\cal H}) } 
\Big\{\sum_{i}a_i^i+ \Tr S\Big|\Pi(a,S) \in {\cal S}_{P}^* \Big\}.
\Label{NYIP}
\end{align}
In the same way, the dual problems $D3$ and $D4$ of $P3$ and $P4$ are given as the following maximizations:
\begin{align}
&S(D3)\nonumber \\
:=&\max_{(a,S) \in \mathbb{R}^{d\times d}\times {\cal T}_{sa}({\cal H}) } 
\Big\{\sum_{i}a_i^i+ \Tr S\Big|
\nonumber \\
&\hspace{20ex}
\Pi(a,S) \in 
({\cal S}({\cal R}_C\otimes {\cal H})_{PPT}
\cap {\cal B}'')^*
 \Big\}.
\Label{NYIP4} \\
&S(D4)\nonumber \\
:=&\max_{(a,S) \in \mathbb{R}^{d\times d}\times {\cal T}_{sa}({\cal H}) } 
\Big\{\sum_{i}a_i^i+ \Tr S\Big|
\nonumber \\
&\hspace{20ex}
\Pi(a,S) \in 
({\cal S}({\cal R}_C\otimes {\cal H})_P \cap {\cal B}'')^*
\Big\}.
\Label{NYIP5}\\
&S(D5)\nonumber \\
:=&\max_{(a,S) \in \mathbb{R}^{d\times d}\times {\cal T}_{sa}({\cal H}) } 
\Big\{\sum_{i}a_i^i+ \Tr S\Big|
\nonumber \\
&\hspace{20ex}
\Pi(a,S) \in 
({\cal S}({\cal R}_C\otimes {\cal H})_P \cap {\cal B}''_\rho)^*
\Big\}.
\Label{NYIP5}
\end{align}
Since the minimizations $P1,\ldots, P5$ are conic linear programmings, 
we have the following theorem.
\begin{theorem}\Label{TH3}
We have
\begin{align}
S(Pl)=S(Dl)
\end{align}
for $l=1,2,3,4,5$.
\end{theorem}
For its proof, see Appendix \ref{AAM}.
When ${\cal H}$ is a finite-dimensional space, 
the conic linear programing problems appearing in Theorem \ref{TH3}
are conic linear programing problems on finite-dimensional vector spaces.

Also, as shown in Appendix \ref{ApD}, the following theorem holds.
\begin{theorem}\Label{TH96}
We have
\begin{align}
S(D0)=S(D1).
\end{align}
\end{theorem}
The combination of Theorems \ref{TH2}, \ref{TH3}, and \ref{TH96} implies 
Proposition \ref{TH1}.
The combination of Theorems \ref{TH2}, \ref{TH3}, and \ref{TH96} implies 
another proof of Proposition \ref{TH1}.

\section{Calculation of \texorpdfstring{$S(P1)$}{S(P1)}}
\label{S4}
\subsection{Upper bound}
Although the minimization problem $P2$, i.e., 
the \textsf{NH} bound $C^{NH}[G]$ can be solved by semi-definite programming,
the cone in the primal conic linear programming $P1$ is 
the separable cone ${\cal S}_{SEP}$, which is different from
the set of positive semi-definite matrices.
Hence, the primal conic linear programming $P1$ is still not so easy 
even with a finite-dimensional space ${\cal H}$.
In fact, this type of problem appears in the membership problem
of the separable cone ${\cal S}_{SEP}$, which appears in the decision problem of 
the existence of entanglement \cite{Gurvits,bruss2002characterizing}
and in generalized robustness of entanglement \cite{takagi2019prl}.
Also, the reference \cite[Proposition 2]{c-value}
showed that the communication value can be calculated 
as a conic linear programming with the separable cone ${\cal S}_{SEP}$.

In this section, we consider an efficient algorithm to solve 
$S(P1)$. 
Although we present our algorithm for 
the conic linear programming $S(P1)$,
this algorithm can be applied to a general 
conic linear programming with the separable cone ${\cal S}_{SEP}$
including the problem presented in \cite[Proposition 2]{c-value}.}
To solve this problem, 
we choose a finite set ${\cal W}_R:=\{ |w_s\rangle\}_{s=1}^m$ composed of several normalized vectors in 
$\mathbb{R}^{d+1}$.
Then, we choose a subset $
{\cal S}({\cal W}_R) :=
\{\sum_{s=1}^m | w_s\rangle \langle w_s| \otimes X_s|
X_s \in {\cal T}_{sa,+}({\cal H})\}\subset
{\cal S}_{SEP}$.
Replacing ${\cal S}_{SEP}$ by 
${\cal S}({\cal W}_R)$,
we consider the minimization:
\begin{align}
S[P1,{\cal W}_R]:= 
\min_{X \in {\cal S}({\cal W}_R)}
\big\{\Tr \big[( G \otimes \rho) X \big] \big|
\eqref{c1},\eqref{c2} \hbox{ hold.}
\big\}\Label{o1-BN}.
\end{align}
The inclusion relation 
${\cal S}({\cal W}_R) \subset {\cal S}_{SEP}$ implies the relation
\begin{align}
S[P1,{\cal W}_R] \ge S(P1).
\end{align}
We illustrate the relationship between $P1$ and the \textsf{SDP} $[P1,\mathcal W_R]$ in Fig.~\ref{fig:sdpWR}.
\begin{figure}
         \centering
    \includegraphics[width=6cm]{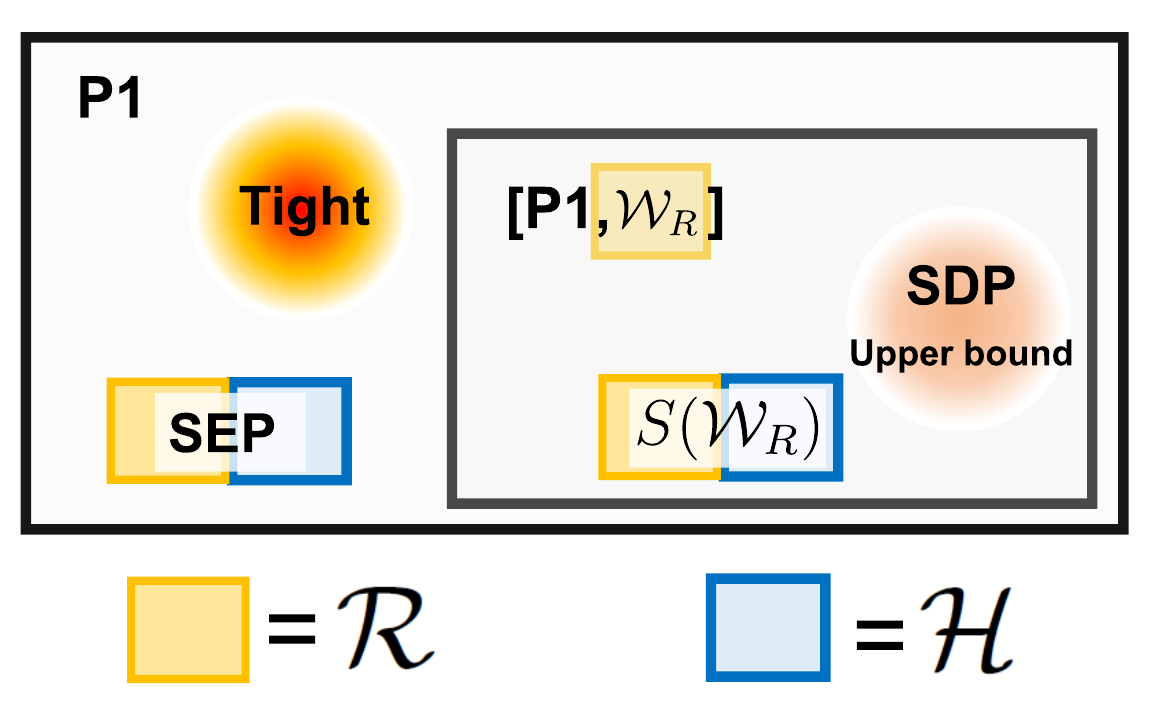}
    \caption{Relationship between $P1$ and $[P1,\mathcal W_R]$. 
    To convert the problem $P1$ to an SDP, we have introduced a subset 
    ${\cal S}({\cal W}_R)$ of ${\cal S}_{SEP}$. The inclusion relation ${\cal S}({\cal W}_R) \subset {\cal S}_{SEP}$ implies the relation $ S[P1,{\cal W}_R] \ge S(P1) $. Also, 
    due to the construction of ${\cal S}({\cal W}_R)$,
    $ S[P1,{\cal W}_R]$ can be solved by an SDP.}
    \label{fig:sdpWR}
\end{figure}
To calculate $S[P1,{\cal W}_R]$, 
we define the set of block diagonal matrices 
${\cal T}({\cal W}_R) :=
\{\sum_{s=1}^m |s\rangle \langle s| \otimes X_s|
X_s \in {\cal T}_{sa,+}({\cal H})\}$
on the vector space $\mathbb{C}^m \otimes {\cal H}$.
Identifying an element  
$\sum_{s=1}^m | w_s\rangle \langle w_s| \otimes X_s
\in {\cal S}({\cal W}_R), X_s \in {\cal T}_{sa,+}({\cal H})$
with a block diagonal matrix
 $\sum_{s=1}^m |s\rangle \langle s| \otimes X_s \in {\cal T}({\cal W}_R)$, we have another form of 
$S[P1,{\cal W}_R]$:
\begin{align}
&S[P1,{\cal W}_R]\nonumber \\
=& 
\min_{X=(X_s) \in {\cal T}({\cal W}_R)} 
\left\{ \left. \sum_{s=1}^m \langle w_s| G |w_s\rangle \Tr X_s \rho
\right| C[P1,{\cal W}_R] \hbox{ holds.} 
\right\},
\end{align}
where
the condition $C[P1,{\cal W}_R]$ is defined as
\begin{align}
\sum_{s=1}^m \langle w_s| 0\rangle \langle 0|w_s\rangle X_s&=I,
\label{CP1-complete}
\\
{\frac 1 2}
\sum_{s=1}^m 
\langle w_s|( | 0\rangle \langle i|+ | i\rangle \langle 0|) |w_s\rangle 
\Tr  [X_s D_j ]&=\delta_i^j.\Label{ACVX}
\end{align}
We call the above problem $[P1,{\cal W}_R]$.
This problem can be considered as an \textsf{SDP} with $d^2+n^2$ constraints
with block diagonal matrices, which is a typical case of sparse SDP.

We define the number $\delta({\cal W}_R)$ as
\begin{align}
\delta({\cal W}_R):=\max_{x \in \mathbb{R}^{d+1}:\|x\|=1 } 
\min_{w \in {\cal W}_R}
\| |x\rangle \langle x| - |w\rangle \langle w|\|_1.
\end{align}
We consider a sequence ${\cal W}_{R,n}$ such that $\delta({\cal W}_{R,n})\to 0$.
As shown in Theorem \ref{thm:D0-lower-bound}, we have
\begin{align}
\lim_{n\to \infty}S[P1,{\cal W}_{R,n}] =
S(P1).\Label{ASC-B}
\end{align}
Hence, we can use \textsf{SDP} for this problem.

Many existing algorithms for \textsf{SDP} work well for sparse positive semi-definite matrices.
The paper \cite{FKN} studied the calculation complexity for
generic primal-dual interior-point method for SDP.
When we apply their analysis for the calculation complexity to the case with block diagonal matrices, 
as explained in Remark \ref{Rem1},
we find that 
the calculation complexity of $S[P1,{\cal W}_R]$ is 
\begin{align}
&O( m((d^2+n^2)n^3+(d^2+n^2)^2n^2)) \nonumber \\
=&
O( m(n^6+ d^2 n^3+d^4 n^2))
\Label{CMO}.
\end{align}

\begin{remark}\Label{Rem1}
To get \eqref{CMO}, we explain how to apply the analysis by the reference 
\cite{FKN}. 
In the reference \cite{FKN}, $m$ is the number of constraint,
and $n$ is the dimension of the vector space to define the positive semi-definite matrix. 
In their analysis for generic primal-dual interior-point method for SDP,
the dominant part of the calculation complexity 
is given as the calculation of ${\bf B}'$ and ${\bf r}'$ in Eq. (8') of \cite{FKN}.
When the matrices are given as block diagonal matrices,
the number of blocks is $n_1$ and size of each block is $n_2$,
the calculation complexity of ${\bf B}'$ is
$O(mn_2^3 n_1+m^2 n_2^2n_1)$ and
the calculation complexity of ${\bf r}'$ is
$O(n_2^3 n_1+m n_2^2n_1)$.
Therefore, the total calculation complexity is given as
$O(mn_2^3 n_1+m^2 n_2^2n_1)$.
Applying this estimation to our setting, we obtain \eqref{CMO}.
\end{remark}

\subsection{Estimator attaining upper bound}

Here, we explain how to construct an estimator to attain 
the upper bound $S[P1,{\cal W}_R]$ with ${\cal W}_R=\{|w_s\rangle\}_{s=1}^m$,
where $|w_s\rangle=\sum_{j=0}^{d}w_s^j |j\rangle  $.
For the optimal solution $X^*_1, \dots, X^*_m$ to the \textsf{SDP} $[P1,\overline{\mathcal W}_R]$, 
we can define the positive semi-definite matrix 
\begin{align}
M_s:= |w_s^0|^2 X^*_s.
\label{construct-POVM}
\end{align}
The first condition \eqref{CP1-complete} implies
\begin{align}
\sum_{s=1}^m M_s=\sum_{s=1}^m |w_s^0|^2 X^*_s=
\sum_{s=1}^m \langle w_s| 0\rangle \langle 0|w_s\rangle X^*_s=I,
\end{align}
which shows that $\{M_s\}$ satisfies the condition of POVM.
For each outcome $s$, we define $\hat\theta(s) \in \mathbb{R}^d$ as 
\begin{align}
\hat\theta^i(s):= \frac{w_s^i}{w_s^0}+\theta^i
\label{construct-estimator}
\end{align}
so that the pair $(\{M_s\},  \hat\theta)$ forms an estimator.
The second condition \eqref{ACVX} implies
\begin{align}
&\sum_{s=1}^m 
\hat\theta^i(s)
\Tr M_s D_j
=
\sum_{s=1}^m 
\frac{w_s^i}{w_s^0}
\Tr |w_s^0|^2 X^*_s D_j
+\sum_{s=1}^m 
\theta^i \Tr |w_s^0|^2 X^*_s D_j
 \nonumber \\
=&
{\frac 1 2}
\sum_{s=1}^m 
\langle w_s|( | 0\rangle \langle i|+ | i\rangle \langle 0|) |w_s\rangle 
\Tr  X^*_s D_j 
+\theta^i\Tr I D_j  
=\delta_i^j+0.
\end{align}
Therefore,  the estimator $(\{M_s\},  \hat\theta)$ satisfies
the locally unbiasedness condition.
Its weighted mean square error is calculated as
\begin{align}
&\sum_{i,j}G_{i,j}
\sum_{s=1}^m 
(\hat\theta^j(s)-\theta^j) (\hat\theta^i(s)-\theta^i)
\Tr M_s D_j\\
=&
\sum_{i,j}G_{i,j}
\sum_{s=1}^m \frac{w_s^j}{w_s^0} \frac{w_s^i}{w_s^0}
\Tr |w_s^0|^2 X^*_s D_j 
=
\sum_{s=1}^m \langle w_s| G |w_s\rangle \Tr X^*_s \rho.
\label{tight-bound-POVM}
\end{align}
Therefore we have the following result.
\begin{theorem}
Let us use the optimal solution of $[P1,{\cal W}_R]$ to construct $M_s$ according to  \eqref{construct-POVM} and  $\hat \theta$ according to \eqref{construct-estimator}. Then the estimator $(\{M_s\},  \hat\theta)$ attains the upper bound $S[P1,{\cal W}_R]$.
\end{theorem}

\subsection{Lower bound}\Label{sec:lower-bound-to-D1}
Since 
the solution $S[P1,{\cal W}_R]$ of the \textsf{SDP} gives an upper bound of $S(P1)$,
we need to derive a lower bound for $S(P1)$.
However, it is not easy to derive an lower bound of $S(P1)$ from 
 the solution $S[P1,{\cal W}_R]$.
In the following, we present a lower bound of $S(P1)$
under the algorithm presented in the above section.
Although the algorithm and the upper bound 
presented in the above section
work with a general 
conic linear programming with the separable cone ${\cal S}_{SEP}$,
the following lower bound does not necessarily work with
a general 
conic linear programming with the separable cone ${\cal S}_{SEP}$
because the lower bound depends on the form of 
the conic linear programming $S(P1)$ with the separable cone ${\cal S}_{SEP}$.

To derive its lower bound, 
we denote the optimum $a$ and $S$ of the dual problem 
$[D1,{\cal W}_R]$ of $[P1,{\cal W}_R]$ by 
$a^*= \sum_{i,j} a^*_{i,j} |i\>\<j|$ and $S^*$.

Namely, we define the dual problem $[D1,{\cal W}_R]$ as 
\begin{maxi!}
{(a,S) \in \mathbb{R}^{d\times d}\times {\cal T}_{sa}({\cal H}) }
{\sum_{i}a_i^i+ \Tr S }{}{}\label{D1WR:obj}
\addConstraint{
\sum_{s=1}^m |s\>\<s| \otimes \<w_s|\Pi(a,S)|w_s\>
}{\ge 0 }{} 
\label{constraint}
\end{maxi!} 
with optimal value 
$S[D1,\mathcal W_R] := \sum_{i}(a^*)_i^i+ \Tr S^*$.

Then, we define the real number
\begin{align}
C_2(a):=\frac{1}{2} 
\Big\| \sum_{j}
\Big(\rho^{-1/2}(\sum_{j'}a_{j}^{j'}D_{j'}  )\rho^{-1/2}\Big)^2 \Big\|^{1/2} .
\label{def:C2a}
\end{align}
for a $d \times d$ matrix $a$.
Using this value, we define
\begin{align}
{X^*} &:=  \Pi(a^*, S^*) \Label{def:X*} \\
\kappa &:= - \min_{v \in {\cal H}: \|v\|=1} \min_{y : \|y\| \le C_2(a^*)}
\langle (1,y),v|X^*|(1,y),v\rangle.\Label{ZMP-B}
\end{align}

Then, we have the following lemma.

\begin{lemma}\Label{LL8Y}
A pair $(a,S)$ satisfies the condition \eqref{MML} if and only if
the relation 
\begin{align}
\|y\|^2 \rho + \sum_{j} (a y)^j D_j -S\ge 0
\end{align}
holds for $y \in \mathbb{R}^d$ with 
$\|y\| \le C_2(a)$.
\end{lemma}

\begin{proof}
In this proof, $(c_j)_j$ expresses the vector whose components are $c_1, \ldots, c_d$.
A pair $(a,S)$ satisfies the condition \eqref{MML} if and only if
any $y \in \mathbb{R}^d$ and $|\phi\rangle \in {\cal H} $
satisfy 
\begin{align}
\|y\|^2 \langle \phi| \rho |\phi\rangle +\sum_{j,j'} a_{j'}^j y^{j'}
\langle \phi| D_j |\phi\rangle-\langle \phi| S|\phi\rangle\ge 0.
\Label{AMU}
\end{align}
The LHS is rewritten as
\begin{align}
&\langle \phi| \rho |\phi\rangle
\Big\|y -
\Big(\frac{1}{2\langle \phi| \rho |\phi\rangle
} \sum_{j'}a_{j'}^j \langle \phi| D_{j'} |\phi\rangle\Big)_j
\Big\|^2 \nonumber \\
&-\langle \phi| S|\phi\rangle
-\frac{1}{4\langle \phi| \rho |\phi\rangle}
\Big\|\Big(
 \sum_{j'}a_{j'}^j \langle \phi| D_{j'} |\phi\rangle\Big)_j
\Big\|^2 .
\end{align}
That is, when $|\phi\rangle$ is fixed, 
the minimum of the LHS for $y$ is realized when $y=
\Big(\frac{1}{2\langle \phi| \rho |\phi\rangle
} \sum_{j'}a_{j}^{j'} \langle \phi| D_{j'} |\phi\rangle\Big)_j$.
Therefore, 
the range of $y \in \mathbb{R}^d$ in the condition \eqref{AMU}
can be limited to the set 
$\Big\{y \in \mathbb{R}^d \Big|
\|y\|\le
\max_{|\phi\rangle}
\Big\|
\Big(\frac{1}{2\langle \phi| \rho |\phi\rangle
} \sum_{j'}a_{j}^{j'} \langle \phi| D_{j'} |\phi\rangle\Big)_j
\Big\| \Big\}$.

Since we have
\begin{align}
&
\Big\|
\Big(\frac{1}{2\langle \phi| \rho |\phi\rangle
} \sum_{j'}a_{j}^{j'} \langle \phi| D_{j'} |\phi\rangle\Big)_j
\Big\|^2\nonumber\\
=&
\sum_{j}
\Big(\frac{1}{2\langle \phi| \rho |\phi\rangle
} \sum_{j'}a_{j}^{j'} \langle \phi| D_{j'} |\phi\rangle\Big)^2\nonumber\\
=&
\sum_{j}
\Big(\frac{1}{2\langle \phi| \rho |\phi\rangle
}  \langle \phi| \sum_{j'}a_{j}^{j'}D_{j'} |\phi\rangle\Big)^2,
\end{align}
using $|\phi'\rangle:= \rho^{1/2}|\phi\rangle $,
we have 
\begin{align}
&\max_{|\phi\rangle}
\Big\|
\Big(\frac{1}{2\langle \phi| \rho |\phi\rangle
} \sum_{j'}a_{j}^{j'} \langle \phi| D_{j'} |\phi\rangle\Big)_j
\Big\|^2\nonumber\\
=&\max_{|\phi'\rangle}
\sum_{j}
\Big(\frac{1}{2\langle \phi'|\phi'\rangle
}  \langle \phi'| \rho^{-1/2}(\sum_{j'}a_{j}^{j'}D_{j'})\rho^{-1/2} |\phi'\rangle\Big)^2 \nonumber\\
=&\max_{|\phi'\rangle}
\sum_{j}
\frac{1}{4\langle \phi'|\phi'\rangle^2
}  \Big(\langle \phi'| \rho^{-1/2}(\sum_{j'}a_{j}^{j'}D_{j'})\rho^{-1/2} |\phi'\rangle\Big)^2 \nonumber\\
=&\max_{|\phi'\rangle:\|\phi'\|=1}
\sum_{j}
\frac{1}{4}  \Big(\langle \phi'| \rho^{-1/2}(\sum_{j'}a_{j}^{j'}D_{j'})\rho^{-1/2} |\phi'\rangle\Big)^2 \nonumber\\
\le &\max_{|\phi'\rangle:\|\phi'\|=1}
\sum_{j}
\frac{1}{4}  \langle \phi'| \Big(\rho^{-1/2}(\sum_{j'}a_{j}^{j'}D_{j'})\rho^{-1/2}\Big)^2 |\phi'\rangle \nonumber\\
= &\max_{|\phi'\rangle:\|\phi'\|=1}
\frac{1}{4}  \langle \phi'| \sum_{j}
\Big(\rho^{-1/2}(\sum_{j'}a_{j}^{j'}D_{j'})\rho^{-1/2}\Big)^2 |\phi'\rangle\nonumber\\
= &
\frac{1}{4} 
\Big\| \sum_{j}
\Big(\rho^{-1/2}(\sum_{j'}a_{j}^{j'}D_{j'})\rho^{-1/2}\Big)^2 \Big\|=C_2(a)^2.
\end{align}
Hence, the range of $y \in \mathbb{R}^d$ in the condition \eqref{AMU}
can be limited to the set 
$\big\{y \in \mathbb{R}^d \big|
\|y\|\le C_2(a)\big\}$.
Thus, we obtain the desired statement.
\end{proof}

Then, for the lower bound of $S(P1)$, we have the following theorem.

\begin{theorem}
\Label{thm:lowerbound-kappa}
We have the following relation.
\begin{align}
S[D1,{\cal W}_R]=&S[P1,{\cal W}_R]
\ge S(P1)=S(D1) \nonumber \\
\ge &
\underline{S}[D1,{\cal W}_R]
:=S[D1,{\cal W}_R]
- n \kappa
 \Label{eq:theorem-lower-bound-to013},
\end{align} 
where $n$ is the dimension of Hilbert space ${\cal H}$.
In particular, \eqref{eq:theorem-lower-bound-to-OD0} holds with equality in the limit $\delta \to 0$.
\end{theorem}

When we calculate a lower bound of $S(P1)$ from 
$S[P1,{\cal W}_R]$,
we have knowledge for $a^*$ and $X^*$ in a real numerical calculation.
Hence, it is possible 
to calculate $\kappa$ by applying Theorem \ref{thm:lowerbound-kappa}.

\begin{proof}
It is sufficient to show that 
the pair $(a^*, S^*- \kappa I)$ satisfies the condition \eqref{MML}
because in this case, the corresponding objective function of $[D1,\mathcal W_R]$ is
$\tr a^*+\tr (S^*-\kappa I)$.
Due to Lemma \ref{LL8Y}, it is sufficient to show that 
the pair $(a^*, S^*-\kappa I)$ satisfies 
\begin{align}
\|y\|^2 \rho +\sum_{j} (a^* y)^j D_j -S^*+\kappa I\ge 0\Label{AA8TE}
\end{align}
for $ y \in \mathbb{R}^d$ with $\|y\| \le C_2(a^*)$.

For $ y \in \mathbb{R}^d$ with $\|y\| \le C_2(a^*)$ and 
$v\in {\cal H}$ with $\|v\|=1$,
we have
\begin{align}
&\langle v|\Big(
\|y\|^2 \rho +\sum_{j} (a^* y)^j D_j -S^*+\kappa  I
\Big) |v\rangle \nonumber \\
=&\langle (1,y),v|X^*|(1,y),v\rangle+\langle v| \kappa I|v \rangle\nonumber \\
=&\langle (1,y),v|X^*|(1,y),v\rangle+\kappa \ge 0 ,
\end{align}
which implies \eqref{eq:theorem-lower-bound-to013}.
\end{proof}

\subsection{Theoretical Evaluation of \texorpdfstring{$\kappa$}{kappa}}
\label{ssec:evaluate-kappa}
When we derive a lower bound of $S(P1)$ 
in our numerical calculation, 
we can calculate $\kappa$ numerically, and can directly use Theorem \ref{thm:lowerbound-kappa}.
However, when we estimate the calculation complexity
to get the tight CR bound within additive error $\epsilon$,
we need to evaluate $\kappa$ theoretically.
For this aim, we define 
\begin{align}
\delta(\mathcal{W}_R) := \max_{x \in \mathbb{R}^{d+1} : \|x\|=1} 
\min_{|w\> \in \mathcal{W}_R} \| |x\>\<x| - |w\>\<w| \|_1.
\Label{def:delta-thm6}
\end{align}
Then, the following lemma gives an upper bound of $\kappa$.
\begin{lemma}\Label{LLP1}
The quantity $\delta(\mathcal{W}_R)$ gives an upper bound of $\kappa$.
\begin{align}
\kappa \le \|X^*\|  (1+  C_2(a^*)^2 ) \delta(\mathcal{W}_R) \Label{AA8}.
\end{align} 
\end{lemma}

\begin{proof}
For $ y \in \mathbb{R}^d$ with $\|y\| \le C_2(a^*)$
and $\delta=\delta(\mathcal{W}_R)$, 
we have
\begin{align}
&\langle (1,y)|X^*|(1,y)\rangle \nonumber\\
\ge & (1+\|y\|^2) \langle w|X^*|w\rangle
-\|X^*\|  \| |(1,y)\rangle \langle (1,y)| \nonumber\\
&- (1+\|y\|^2) |w\rangle \langle w| \|_1  I\nonumber\\
\ge & (1+\|y\|^2) \langle w|X^*|w\rangle
-\|X^*\|  (1+\|y\|^2) \delta I\nonumber\\
\ge & (1+\|y\|^2) \langle w|X^*|w\rangle
-\|X^*\|  (1+ C_2(a^*)^2) \delta I \nonumber\\
\ge &
-\|X^*\|  (1+ C_2(a^*)^2) \delta I ,
\end{align}
which implies \eqref{AA8}.
\end{proof}

Combining Theorem \ref{thm:lowerbound-kappa} and Lemma \ref{LLP1}, we have the following theorem.

\begin{theorem}
\Label{thm:D0-lower-bound}
\begin{align}
&S[D1,{\cal W}_R]=S[P1,{\cal W}_R] 
\ge S(P1)=S(D1) \nonumber \\
\ge &
S[D1,{\cal W}_R]
- n \|X^*\|  (1+  C_2(a^*)^2 ) \delta(\mathcal{W}_R)
 \Label{eq:theorem-lower-bound-to-OD0}.
\end{align} 
\end{theorem}


Considering the structure of $\kappa$, we have another evaluation of 
$\kappa$. 
For this aim,
we define
\begin{align}
\delta(s,\mathcal{W}_R):= 
\max_{x=(\cos \theta , ( \sin \theta )z ),  z \in S^{d-1}}
\min_{|w\> \in \mathcal W_R} \| |x\>\<x| - |w\>\<w| \|_1 ,
\end{align}
with $\theta= \tan^{-1}s$.
Then, the following lemma gives another upper bound of $\kappa$.
\begin{lemma}\Label{LAH1}
Using the quantity $\delta(s,\mathcal{W}_R)$, we have the following upper bound of $\kappa$.
\begin{align}
\kappa \le \|X^*\|  
\max_{s \in[0, C_2(a^*) ]}(1+s^2 ) 
\delta( s,{\cal W}_R) .
\Label{AA9}
\end{align} 
\end{lemma}

\begin{proof}
For $ y \in \mathbb{R}^d$ with $\|y\| \le C_2(a^*)$,
we have
\begin{align}
&\langle (1,y)|X^*|(1,y)\rangle \nonumber\\
\ge & (1+\|y\|^2) \langle w|X^*|w\rangle \nonumber\\
& -\|X^*\|  \| |(1,y)\rangle \langle (1,y)|- (1+\|y\|^2) |w\rangle \langle w| \|_1 I
\nonumber\\
\ge & (1+\|y\|^2) \langle w|X^*|w\rangle
-\|X^*\|  (1+\|y\|^2) \delta(\|y\|, {\cal W}_R) I\nonumber\\
\ge & (1+\|y\|^2) \langle w|X^*|w\rangle
-\|X^*\|  \max_{s \in[0, C_2(a^*) ]}(1+s^2 ) 
\delta( s,{\cal W}_R)  I \nonumber\\
\ge & 
-\|X^*\|  \max_{s \in[0, C_2(a^*) ]}(1+s^2 ) 
\delta( s,{\cal W}_R)  I .
\end{align}
\end{proof}

\section{Construction of \texorpdfstring{${\cal W}_R$}{WR}}
\label{S5}
\subsection{Quantum \texorpdfstring{$t$}{t}-design}
To implement the proposed method, we need to construct 
the subset ${\cal W}_R \subset \mathbb{R}^{d+1}$.
For this aim, we like to pack $m$ points uniformly on a $d$-sphere of radius 1, such that any point on the sphere is as close to some point as possible.
As one possible choice, we randomly generate pure states 
$|w_1\rangle \langle w_1|, \ldots, |w_l\rangle \langle w_l|$
on the system ${\cal H}$
subject to the Haar measure.
Another choice for ${\cal W}_R$
is a quantum $t$-design.

A subset ${\cal W}_H$ of normalized vectors on the finite-dimensional 
Hilbert space ${\cal H}$ is called a quantum $t$-design on ${\cal H}$
when
\begin{align}
\sum_{w \in {\cal W}_H}\frac{1}{|{\cal W}_H|}
|w\rangle \langle w|^{\otimes t}=
\int |x\rangle \langle x|^{\otimes t} \mu_{\cal H}(dx),
\end{align} 
where $ \mu$ is the Haar measure on the set of pure states on ${\cal H}$ \cite{t-design} \cite[Section 4.5]{H-group2}.
A subset ${\cal W}_R$ of normalized vectors on ${\cal R}$ is 
called a quantum $t$-design on ${\cal R}$
when
\begin{align}
\sum_{w \in {\cal W}_R}\frac{1}{|{\cal W}_R|}
|w\rangle \langle w|^{\otimes t}=
\int |x\rangle \langle x|^{\otimes t} \mu_{\cal R}(dx),
\end{align} 
where $ \mu$ is the Haar measure on the set of pure states on ${\cal R}$.

For a good choice of ${\cal W}_R$ and ${\cal W}_H$, we can consider 
a quantum $t$-design
because a quantum $t$-design has a symmetric property.
For example, the paper \cite{XA} discusses approximate construction of 
quantum $t$-designs. 


\subsection{Spherical \texorpdfstring{$t$}{t}-design}\Label{S4-B}
A quantum $t$ design has not been sufficiently studied, but
a spherical $t$-design has been well studied \cite{Delsarte}.
The references \cite{BALADRAM,Ba1,Ba2} study a spherical $t$-design.
Fortunately, 
a quantum $t$-design can be constructed from a spherical $t$-design \cite{t-design}.

A subset ${\cal V}$ of $S^{d-1}$ is called 
a spherical $t$-design when 
the relation
\begin{align}
\sum_{x \in S^{d-1}} f(x)
=\int f(x) \mu_{S^{d-1}}(dx)
\end{align}
holds any polynomial $f$ with degree $t$, where
$\mu_{S^{d-1}}$ is the Haar measure on $S^{d-1}$.

Given a spherical $2t$-design ${\cal V} \subset S^{2d-1}$,
we define a subset ${\cal V}_H$ of pure states on ${\cal H}=\mathbb{C}^d$ as follows.
For $x \in S^{2d-1}$,
we define the normalized vector $w(x)\in {\cal H}$ as
$w(x)_j:= x_{2j-1}+x_{2j}i$.
As shown below, the set ${\cal V}_H:=\{w(x)\}_{x \in {\cal V}}$ is a quantum $t$-design. 
For any sequences $e_1,\ldots,e_t$ and $f_1,\ldots,f_t$, we have
\begin{align}
& \sum_{x \in {\cal V}}\frac{1}{|{\cal V}|}
\Tr \big[|w(x)\rangle \langle w(x)|^{\otimes t}
|e_1,\ldots,e_t\rangle \langle f_1,\ldots,f_t|\big]\nonumber\\
=&
\sum_{x \in {\cal V}}\frac{1}{|{\cal V}|}
\prod_{j=1}^t\prod_{j'=1}^t w(x)_{e_j} \overline{w(x)}_{f_{j'}}\nonumber\\
=&
\int
\prod_{j=1}^t\prod_{j'=1}^t w(x)_{e_j} \overline{w(x)}_{f_{j'}}
\mu_{S^{2d-1}}(dx) \nonumber\\
=&
\int
\Tr\big[ |w(x)\rangle \langle w(x)|^{\otimes t}
|e_1,\ldots,e_t\rangle \langle f_1,\ldots,f_t|\big]
\mu_{S^{2d-1}}(dx) \nonumber\\
=&
\int \Tr \big[|x\rangle \langle x|^{\otimes t}
|e_1,\ldots,e_t\rangle \langle f_1,\ldots,f_t|\big]
 \mu_{\cal H}(dx).
\end{align} 
Hence, we have
\begin{align}
 \sum_{x \in {\cal V}}\frac{1}{|{\cal V}|}
 |w(x)\rangle \langle w(x)|^{\otimes t}
=\int |x\rangle \langle x|^{\otimes t}
 \mu_{\cal H}(dx).\,
\end{align} 
which shows that 
the set ${\cal V}_H:=\{w(x)\}_{x \in {\cal V}}$ is a quantum $t$-design. 
In the same way, 
given a spherical $t$-design ${\cal V} \subset S^{d-1}$,
the set ${\cal V}$ is a quantum $t$-design on ${\cal R}=\mathbb{R}^d$. 

\subsection{Construction reflecting structure of \texorpdfstring{$\kappa$}{kappa}}\Label{SSD}
Since the lower bound of $S(D1)$ depends on the value  $\kappa$, 
we may construct ${\cal W}$ by reflecting the structure of $\kappa$ as follows.
We choose a subset ${\cal S} \subset S^{d-1}$, which can be a spherical $t$-design.
Then, we define a subset 
${\cal S}(\phi):=\{(\cos \phi, \sin \phi y)\}_{y \in \mathcal S}$.
We choose $\phi_0>0$ as $\tan \phi_0= C_2(a^*)$.
Using $k$, we choose ${\cal W}_R$ as
$\cup_{j=0}^k{\cal S}(\frac{\phi_0 j}{k})$.

We can modify this construction as follows.
That is, 
to construct ${\cal S}(\frac{\phi_0 j}{k}) $,
depending on $j$, we choose a subset ${\cal S}_j \subset S^{d-1}$, which also may be a spherical $t$-design.
We define a subset 
${\cal S}(\frac{\phi_0 j}{k}):=\{(\cos \frac{\phi_0 j}{k}, \sin \frac{\phi_0 j}{k} y)
\}_{y \in {\mathcal S}_j}$.
We choose $\phi_0>0$ as $\tan \phi_0= C_2(a^*)$.
Using $k$, we choose ${\cal W}_R$ as
$\cup_{j=0}^k{\cal S}(\frac{\phi_0 j}{k})$.

Indeed, when $d=2$, the choice of ${\cal S}$ is very easy
because we can choose ${\cal S}$
as $\{ (\cos \frac{2\pi l}{N}, \sin \frac{2\pi l}{N}) \}_{l=1}^{N}$.
In this case, we find that $\delta(\mathcal S) = \| (1,0) - (\cos 2\pi/N , \sin 2\pi/N)\|
= \sqrt{2-2\cos 2\pi/N}$.

Appendix \ref{ACAS} discusses the theoretical evaluation of $\kappa$ in these choices
while it is better to directly calculate $\kappa$
from the obtained data $a^*$ and $X^*$ in a real numerical calculation.

\begin{widetext}
\subsection{Another construction}\Label{XNHY}
First, we construct the real discrete subset ${\cal D}_{n,d}
\subset {\cal S}_d$, where ${\cal S}_d$ is the $d$-dimensional sphere in $\mathbb{R}^{d+1}$.
We define ${\cal D}_{n,1}$ as
\begin{align}
{\cal D}_{n,1}:=
\Big\{
\Big(\cos \frac{2\pi j}{n},\sin \frac{2\pi j}{n}\Big)\Big\}_{j=0}^{n-1}
\end{align}
Then, we inductively define ${\cal D}_{n,d}$ as
\begin{align}
{\cal D}_{n,d}:=
\Big\{
\Big( \cos \frac{2\pi j}{n} v ,\sin \frac{2\pi j}{n}\Big)\Big|
v \in {\cal D}_{n,d-1,R},
j=0, \ldots, n-1
\Big\}.
\end{align}
Then, 
$ |{\cal D}_{n,d}|=n^d$.
We have
\begin{align}
F({\cal D}_{n,1})
:=& \min_{\psi_1\in {\cal S}_2 } \max_{\psi_2\in {\cal D}_{n,2,R}}
|\langle \psi_1|\psi_2\rangle |
= \cos \frac{\pi }{n}.
\end{align}
In general, we have
\begin{align}
F({\cal D}_{n,d})
=&\min_{\psi_1\in {\cal S}_d } \max_{\psi_2\in {\cal D}_{n,d}}
|\langle \psi_1|\psi_2\rangle|=
\cos^d \frac{\pi }{n}.
\end{align}
This can be inductively shown as follows.
\begin{align}
F({\cal D}_{n,d})
=& \min_{\psi_1\in {\cal S}_d } \max_{\psi_2\in {\cal D}_{n,d}}
|\langle \psi_1|\psi_2\rangle |\nonumber \\
=& \min_{\theta}\min_{\psi_1\in {\cal S}_{d-1} } 
\max_{j=0,\ldots, n-1}
\max_{\psi_2\in {\cal D}_{n,d-1}}
\Big|
\langle \psi_1|\psi_2\rangle \cos \theta \cos \frac{2\pi j}{n}
+
\sin \theta \sin \frac{2\pi j}{n}\Big|
\nonumber \\
=& \min_{\theta}\min_{\psi_1\in {\cal S}_{d-1} } 
\max_{j=0,\ldots, n-1}
\max_{\psi_2\in {\cal D}_{n,d-1}}
\Big|
\langle \psi_1|\psi_2\rangle 
\Big(\cos \theta \cos \frac{2\pi j}{n}
+
\frac{1}{\langle \psi_1|\psi_2\rangle }
\sin \theta \sin \frac{2\pi j}{n}
\Big)\Big|
\nonumber \\
=& \min_{\theta}\min_{\psi_1\in {\cal S}_{d-1} } 
\max_{\psi_2\in {\cal D}_{n,d-1}}
|\langle \psi_1|\psi_2\rangle |
\Big(\max_{j=0,\ldots, n-1}
\Big|
\cos \theta \cos \frac{2\pi j}{n}
+
\frac{1}{\langle \psi_1|\psi_2\rangle }
\sin \theta \sin \frac{2\pi j}{n}
\Big|\Big)
\nonumber \\
\ge & \min_{\theta}\min_{\psi_1\in {\cal S}_{d-1} } 
\max_{\psi_2\in {\cal D}_{n,d-1}}
|\langle \psi_1|\psi_2\rangle |
\Big(\max_{j=0,\ldots, n-1}
\max \Big(
\Big|
\cos \theta \cos \frac{2\pi j}{n}
+
\sin \theta \sin \frac{2\pi j}{n}
\Big|,
\notag\\
&\quad\quad 
\Big|
\cos \theta \cos \frac{2\pi j}{n}
-
\sin \theta \sin \frac{2\pi j}{n}
\Big|
\Big)
\Big)
\nonumber \\
= & \min_{\theta}\min_{\psi_1\in {\cal S}_{d-1} } 
\max_{\psi_2\in {\cal D}_{n,d-1}}
|\langle \psi_1|\psi_2\rangle |
\Big(\max_{j=0,\ldots, n-1}
\max \Big(
\Big|
\sin (\frac{2\pi j}{n}+\theta )
\Big|,
\Big|
\sin (\frac{2\pi j}{n}-\theta)
\Big|
\Big)
\Big)
\nonumber \\
= & 
\Big(\min_{\theta}
\max_{j=0,\ldots, n-1}
\max \Big(
\Big|
\sin (\frac{2\pi j}{n}+\theta )
\Big|,
\Big|
\sin (\frac{2\pi j}{n}-\theta)
\Big|
\Big)
\Big(\min_{\psi_1\in {\cal S}_{d-1} } 
\max_{\psi_2\in {\cal D}_{n,d-1}}
|\langle \psi_1|\psi_2\rangle|\Big)
\nonumber \\
\ge & \cos \frac{\pi}{n} F({\cal D}_{n,d-1})
= \cos^d \frac{\pi}{n}.
\end{align}

Therefore,
\begin{align}
&\delta({\cal D}_{n,d}) 
= \max_{ |y\> \in \mathcal R : \|x\|=1} 
\min_{|x\> \in \mathcal {\cal D}_{n,d}} 
\| |y\>\< y| - |x\>\< x| \| \nonumber \\
=&2\sqrt{1-F({\cal D}_{n,d})^2}=2\sqrt{1-\cos^{2d} \frac{\pi}{n}}.
\end{align}

Thus, when $\delta({\cal D}_{n,d}) =\delta$,
we have 
\begin{align}
1-\frac{\delta^2}{4}=\cos^{2d}\frac{\pi}{n}
=(1-\sin^{2}\frac{\pi}{n})^d
\ge (1-\frac{\pi^2}{n^2})^d\ge 1-\frac{d \pi^2}{n^2}.
\end{align}
Hence,
$|{\cal D}_{n,d}|= n^d \le
\frac{(2\pi)^d d^{d/2}}{\delta^d}
$.
\end{widetext}

\section{Calculation complexity of tight CR bound within additive error \texorpdfstring{$\epsilon$}{epsilon}}
\label{S6}
Now, we evaluate the calculation complexity of tight CR bound within additive error $\epsilon $ under the choice of ${\cal W}_R$ given in Section \ref{XNHY}.

For further evaluation, we introduce the quantity: 
\begin{align}
C_1&:= \frac{1}{2}
\Big\| \sum_j (\rho^{-1/2}(D_{j}\rho)\rho^{-1/2})^2 \Big\|^{1/2} \Label{AAS}, 
\end{align}

Then, we have the following lemma;
\begin{lemma}\Label{ADS}
\begin{align}
C_2(a) \le \|a\| C_1.
\end{align}
\end{lemma}

\begin{proof}
We define $A_j=\rho^{-1/2}(D_{j}\rho)\rho^{-1/2}$. Then, since
$\|a\|^2 I \ge a^\dagger  a$, 
we have
\begin{align}
& \|a\|^2 \sum_{j=1}^d A_j^2 =
\|a\|^2
(A_1, \ldots, A_d)I
\left(
\begin{array}{c}
A_1 \\
\vdots \\
A_d
\end{array}
\right) \nonumber \\
\ge &
(A_1, \ldots, A_d) a^\dagger  a
\left(
\begin{array}{c}
A_1 \\
\vdots \\
A_d
\end{array}
\right)
=\sum_{j=1}^d (\sum_{j'=1}^d a_{j}^{j'}A_{j'})^2 .
\end{align}
Then, we have
\begin{align}
\|a\|^2 C_1^2
= \|a\|^2 \Big\| \sum_{j=1}^d A_j^2 \Big\|
\ge
\Big\| \sum_{j=1}^d (\sum_{j'=1}^d a_{j}^{j'}A_{j'})^2 \Big\|
=C_2(a)^2.
\end{align}
\end{proof}


When the pair $a^*,S^*$ is the optimal solution of 
$S[D1,{\cal W}_R]$ with a weight matrix $W$,
we define $\xi:=n \|\Pi(a^*,S^*)\|(1+  \|a^*\|^2 C_1^2) $.
Then, Theorem \ref{thm:D0-lower-bound} and Lemma \ref{ADS} guarantee that 
the error $\epsilon$ is upper bounded by $\xi \delta(W_R)$.
Using this fact, we have the following theorem


\begin{theorem}
\Label{thm:complexity-OD0}
Suppose that we choose
${\cal W}_R$ as ${\cal D}_{n,d} $ with 
$\epsilon \xi^{-1} = \delta({\cal D}_{n,d})$.
Then, the calculation complexity of 
the tight CR bound for probe states of size $n$ within additive error $\epsilon $ is 
\begin{align}
O \Big( \frac{\xi^d d^{d/2}}{\epsilon^d}
(n^6+ d^2 n^3+d^4 n^2)\Big).
\end{align}
\end{theorem}
\begin{proof}
We consider how $|\mathcal{W}_R|$ relates to the additive error of the CR bound.
From Theorem \ref{thm:D0-lower-bound}, our lower bound to $O(D0)$ depends on the covering radius of $|\mathcal{W}_R
$.
When the relation $\epsilon \ge \xi \delta$ holds, 
Lemma \ref{ADS} and \eqref{eq:theorem-lower-bound-to-OD0} 
guarantee that the error of CR bound is upper bounded by $\epsilon$.

Since $\epsilon \xi^{-1} = \delta$, we have 
$m=|\mathcal{W}_R| 
\le \frac{(2\pi)^d d^{d/2}}{\delta^d}
=\frac{(2\pi \xi)^d d^{d/2}}{\epsilon^d}
$.
Hence, substituting the above value to $m$ in \eqref{CMO},
the overall complexity of the algorithm is 
\begin{align}
O \Big( \frac{\xi^d d^{d/2}}{\epsilon^d}
(n^6+ d^2 n^3+d^4 n^2)\Big).
\end{align}
\end{proof}


Since Theorem \ref{thm:complexity-OD0} evaluates 
the calculation complexity of 
the tight CR bound within additive error $\epsilon $ by using the value $x$,
we need to evaluate the quantity $\| \Pi(a^{**},S^{**})\|$, where $(a^{**},S^{**})$ denotes the optimal solution of $D0$.
The following theorem gives its upper bound.
\begin{theorem}
\Label{norm}
We denote the SLD Fisher information matrix by $J$.
We have
\begin{align}
&\| \Pi(a^{**},S^{**})\| \nonumber \\
\le &
\|G\|
+ \frac{d}{2 
\sqrt{\tr G^{1/2} J^{-1}G^{1/2}}}
+ \frac{d}{4 
\tr G^{1/2} J^{-1} G^{1/2}},\Label{XMR4}
\end{align}
and
\begin{align}
\| a^{**} \|\le 
\tr a^{**} \le 
 \frac{d}{2 
\sqrt{\tr G^{1/2} J^{-1} G^{1/2}}}
\Label{XMR6}.
\end{align}
\end{theorem}

Finally, we evaluate 
the calculation complexity of 
the tight CR bound within additive error $\epsilon $ by 
Theorems \ref{thm:complexity-OD0} and \ref{norm}.
Since $ a^{**},S^{**}$ are close to 
$a^{*}({\cal W}_R),S^{*}({\cal W}_R)$,
$\|a^{*}({\cal W}_R)\| $
and
$\| \Pi(a^{*}({\cal W}_R),S^{*}({\cal W}_R))\| $ 
have the same order as 
$\|a^{**}\|$ and $\| \Pi(a^{**},S^{**})\| $, respectively.
We assume that $G$ is the identity matrix.
Combining Theorems \ref{thm:complexity-OD0} and \ref{norm}, we find that 
\begin{align}
\xi=
O\Big(n C_1^2 (\frac{d}{2 \sqrt{\tr  J^{-1}}})^2
(\frac{d}{2 
\sqrt{\tr J^{-1}}}
+ \frac{d}{4 
\tr J^{-1}})\Big).
\end{align}
Hence,
the calculation complexity is upper bounded as
\begin{align}
O\bigg( \frac{
\Big(
n C_1^2 (\frac{d}{2 \sqrt{\tr  J^{-1}}})^2
(\frac{d}{2 
\sqrt{\tr J^{-1}}}
+ \frac{d}{4 
\tr J^{-1}})\Big)
^{d}   d^{d/2}}{  \epsilon^{d}}
(n^6+ d^2 n^3+d^4 n^2)\bigg).
\end{align}

For simplicity, we consider the case when $J$ is the identity matrix.
The value $\tr  J^{-1}$ is $d$. The above value is simplified as
\begin{align}
&O\bigg( \frac{
(n C_1^2 d^{3/2})^{d} d^{d/2} }{  \epsilon^{d}}
(n^6+ d^2 n^3+d^4 n^2)\bigg)\nonumber \\
=&
O\bigg( \frac{1}{ \epsilon^{d}}n^{d} C_1^{2 d}d^{2d}
(n^6+ d^2 n^3+d^4 n^2)\bigg).
\end{align}

\section{Numerical lower bound for \texorpdfstring{$S(D1)$}{S(D1)}}
\label{S7}
Here we numerically find a lower bound for $S(D1)$.
Here $n$ denotes the dimension of $\mathcal H$, and $d$ is the number of parameters that we estimate.
Here, we consider the metrology problem with 
probe state
\begin{align}
\rho = I_n /n,
\end{align}
where $I_n$ denotes an identity matrix of size $n$.
We set the number of parameters to estimate to be and $d=2$.
Now let us define the algorithm \texttt{MakeRandomDs}$(n)$ that generates size $n$ random traceless matrices $D_1$ and $D_2$. 
\newline

\noindent{\bf Algorithm }\texttt{MakeRandomDs}\\
\noindent{$[D_1, D_2] = $\texttt{MakeRandomDs}$(n)$}
\begin{enumerate}
\item For $i=1,2$:
\item \quad Make a random matrix $M$ of size $n$, with each matrix element chosen independently from the uniform distribution on $[0,1]$. 
\item \quad Set $D = M^\dagger  M.$
\item \quad Set $D_i = D - \tr D / n.$
\end{enumerate}

We can interpret this as corresponding to a quantum model with true parameters $\theta_1$ and $\theta_2$, where if $x$ is in the neighborhood of $\theta = (\theta_1,\theta_2)$, 
we have $\rho( (x_1,x_2) )   \approx I_n / n + D_1 (x_1 - \theta_1) + D_2 (x_2 - \theta_2).$
Physically, this could correspond to quantum parameter estimation on a set of states that are close to maximally mixed, which describes the scenario where noise dominates the quantum system.

Now we describe how we make the set $\mathcal W_R$.
First we make the algorithm \texttt{circlepoints}$(N)$ which takes as input a positive integer $N$, and returns a set \textsf{calS}.
\newline

\noindent {\bf Algorithm \texttt{circlepoints}}\newline
\noindent \textsf{calS} = \texttt{circlepoints}$(N)$
\begin{enumerate}
\item Set \textsf{calS} = a matrix with $N$ rows and two columns, with every entry equal to zero.
\item For $j=1:N$
\item \quad     Set $\theta = 2\pi j/N$
\item \quad     Set  \textsf{calS}$(j,1) = \cos(\theta)$
\item \quad     Set  \textsf{calS}$(j,2) = \sin(\theta)$
\end{enumerate}

Next, we make $\mathcal W_R$ using the following algorithm. \footnote{We export information about $\mathcal W_R$ in a csv file \texttt{WR.csv} available on request.}
\newline

\noindent {\bf Algorithm \texttt{makeWR}}\newline
\noindent $\mathcal W_R = \texttt{makeWR}(N)$
\begin{enumerate}
\item Set $N=70, k = 100$ and $\phi_0 = 1.2$.
\item Set $\mathcal W_R = \emptyset$
\item for $\textsf{idx} = 1:(k+1)$
\item \quad      Set $\textsf{currN} = \max(\lceil(\textsf{idx}/k)^{1/4} N\rceil,20)$
\item \quad      Set \textsf{calS} = \texttt{circlepoints}(\textsf{currN})
\item \quad 	    Set $j = \textsf{idx}-1$
\item \quad 	    Set $\phi = \phi_0 j/k$
\item \quad 	    Set {\bf 1} as a column vector of ones that has the same number of rows as \textsf{calS}
\item \quad     Set $W$ as the matrix with three columns $[\cos(\phi) {\bf 1}, \sin(\phi) \textsf{calS}]$ with $N$ rows.
\item \quad     Add every row of $W$ into the set $\mathcal W_R$.
\end{enumerate}
We find that $\mathcal W_R$ has 5750 points.

For each $n=3,4,5,6,7,8,9,10,11,12,13,14,15,16,17$, we run Algorithm \texttt{MakeRandomDs} 50 times, and use the same $\mathcal W_R$. 
For each $(\rho,D_1,D_2)$ that we obtain, we 
\begin{enumerate}
\item calculate the \textsf{NH} bound, the 
\textsf{HN} bound, $S[D1,\mathcal W_R]$ 
and $\underline S[D1,\mathcal W_R]$.
\end{enumerate}
To find $S[D1,\mathcal W_R]$ and $\underline S[D1,\mathcal W_R]$, we first solve $[D1,\mathcal W_R]$ and find its optimal solution $(a^*,S^*)$. 
The optimal value is $S[D1,\mathcal W_R] = \tr a^* + \tr S^*$, and this is our upper bound to $S(D1)$. 
To evaluate $\underline S[D1,\mathcal W_R]$, which is our lower bound to $S(D1)$,
we first evaluate $X^* = \Pi(a^*,S^*) \eqref{def:X*}.$ 
Next, we evaluate $c_2 = C_2(a^*)$ \eqref{def:C2a} and $\phi_0^* = \tan^{-1}(c_2).$
Our next step is to numerically approximate $\kappa$ as defined in \eqref{ZMP-B}. 

We implement our approximation of $\kappa$ in Theorem \ref{thm:lowerbound-kappa} by the following method.
\begin{enumerate}
\item Set $\textsf{myrange}  = \{-c_2 + 2c_2 j/999 |  j=0,\dots, 999\}$. (We can implement this in MatLab using \textsf{linspace}$(-c_2,c_2,1000)$).
\item Set $\kappa = - \infty$.
\item For x in \textsf{myrange}:
\item \quad For y in \textsf{myrange}:
\item \quad \quad Set $w = (1,x,y)$ as a column vector.
\item \quad \quad If $\|(x,y) \| \le c_2$:
\item \quad \quad \quad Set $m = \lambda_{\rm min}(\<w| X^* |w\>)$, where $\lambda_{\rm min}(\cdot)$ gives the minimum eigenvalue of a Hermitian matrix.
\item \quad \quad \quad If $m < -\kappa$, set $\kappa= - m$.
\end{enumerate}
This subroutine evaluates the smallest eigenvalue of the matrices $\<(1,y)|  X^*|(1,y)\>$ amongst the no more than $10^6$ points of $y \in \mathbb R^2$, and $\|y\|\le c_2$.
Then we set our lower bound to $S(D1)$ to be 
$\underline S[D1,\mathcal W_R] = \tr a^* + \tr S^* - n \kappa$.
If $\underline S[D1,\mathcal W_R] > S(P2)$, we set $\textsf{isgap}=1$.
Otherwise, we set $\textsf{isgap}=0$.


Since $(D_1,D_2)$ is generated randomly by \texttt{MakeRandomDs}, 
$\textsf{isgap}$ is a binary random variable. 
We define $f_n$ as the probability that $\textsf{isgap}$ takes the value $1$.
Since $f_n$ is an unknown parameter, 
we estimate $f_n$ from
our 50 numeric experiments that independently generates
$(D_1,D_2)$, which determines $\textsf{isgap}$ in the above method.
Let $g_n$ be the number of times $\textsf{isgap} =1$ out of the 50 experiments.
We plot the values of $g_n$ along with their error bars in in Fig.~\ref{fig:frac}.

\begin{figure}
         \centering
    \includegraphics[width=8cm]{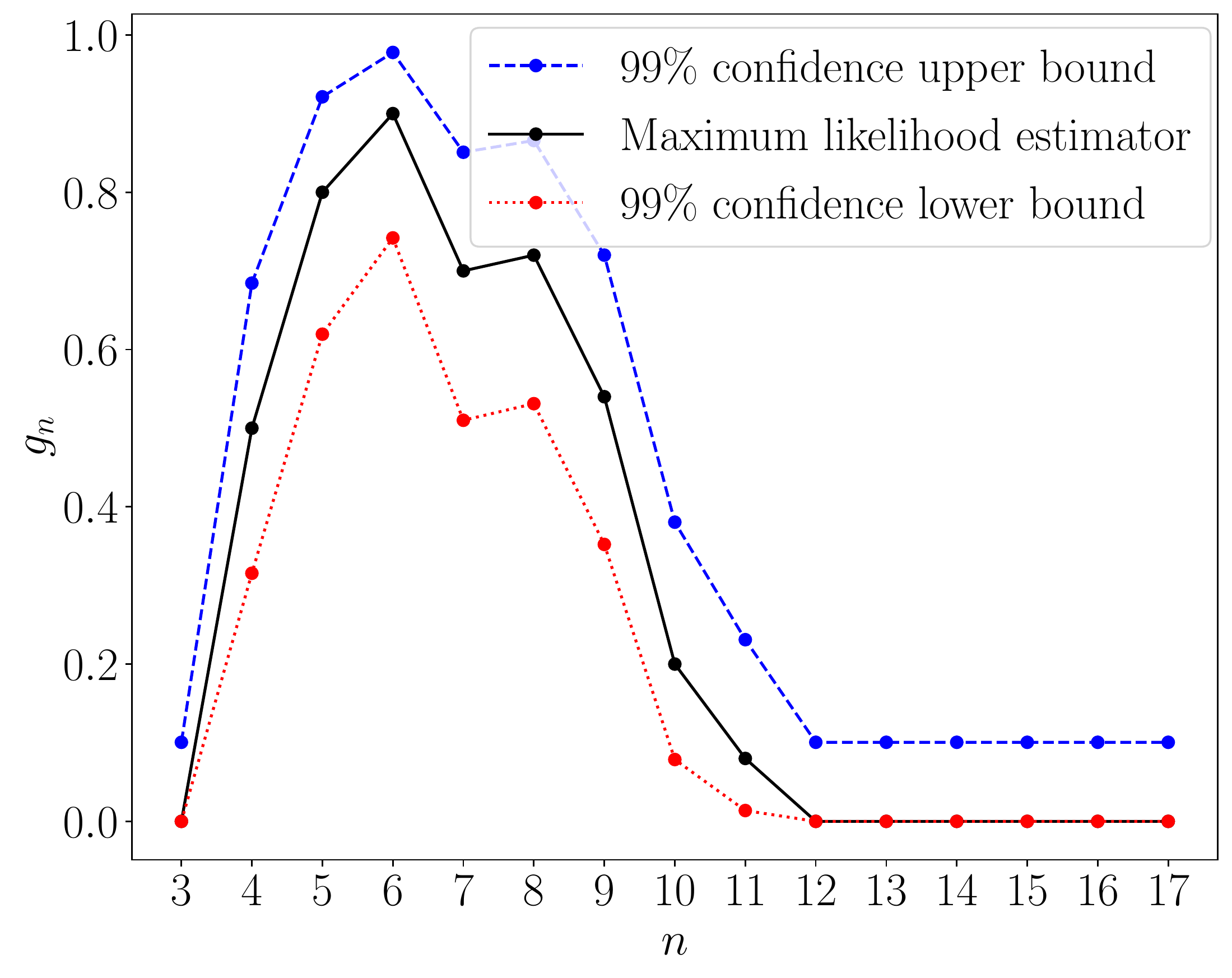}
    \caption{For each value of $n={\rm dim}(\cal H)$, we randomly generate $50$ independent pairs of derivatives of probe states $(D_1,D_2)$ according to Algorithm \texttt{MakeRandomDs}. 
    Let $\beta_n$ count how many of the 50 independent pairs $(D_1,D_2)$ have $\underline S[D1,\mathcal W_R]/S(P2) > 1$. Then set $g_n = \beta_n / 50$. Here $g_n$ is a maximum likelihood estimator of $f_n$ which is the asymptotic fraction of experiments which would have $\underline S[D1,\mathcal W_R]/S(P2) > 1$. 
        When $g_n$ is close to 1, this indicates that it is very likely that all random derivatives of probe states have a $C[G] > S(P2)$. Hence, the numerical evidence show that there is very often a gap between the \textsf{tight} bound and the \textsf{NH} bound, demonstrating the benefit of using the \textsf{tight} bound over the \textsf{NH} bound.
    We calculate the 99\% confidence interval of our estimator $g_n$ according to the Clopper–Pearson method (using the function \textsf{binofit} in \textsf{MatLab}.)
    }
    \label{fig:frac}    
\end{figure}

For fixed $n$, we plot the largest values of $\underline S[D1,\mathcal W_R] / S(P2)$ among our 50 experiments in Fig.~ \ref{fig:violation}. 
Values of $\underline S[D1,\mathcal W_R] / S(P2)$ that are strictly larger than 1 illustrate a positive gap between $S(P1)$ and $S(P2)$.
We also plot the corresponding values of $S(P4) / S(P2)$ in Fig.~ \ref{fig:violation}.

As expected, when $n=3$, $\underline S[D1,\mathcal W_R]$ is never greater than the \textsf{NH} bound given by $S(P2)$. However, for $n=4,5,6,7,8,9$, in the majority of our numerical experiments, $\underline S[D1,\mathcal W_R]$ exceeds the \textsf{NH} bound.

\begin{figure}
         \centering
    \includegraphics[width=8cm]{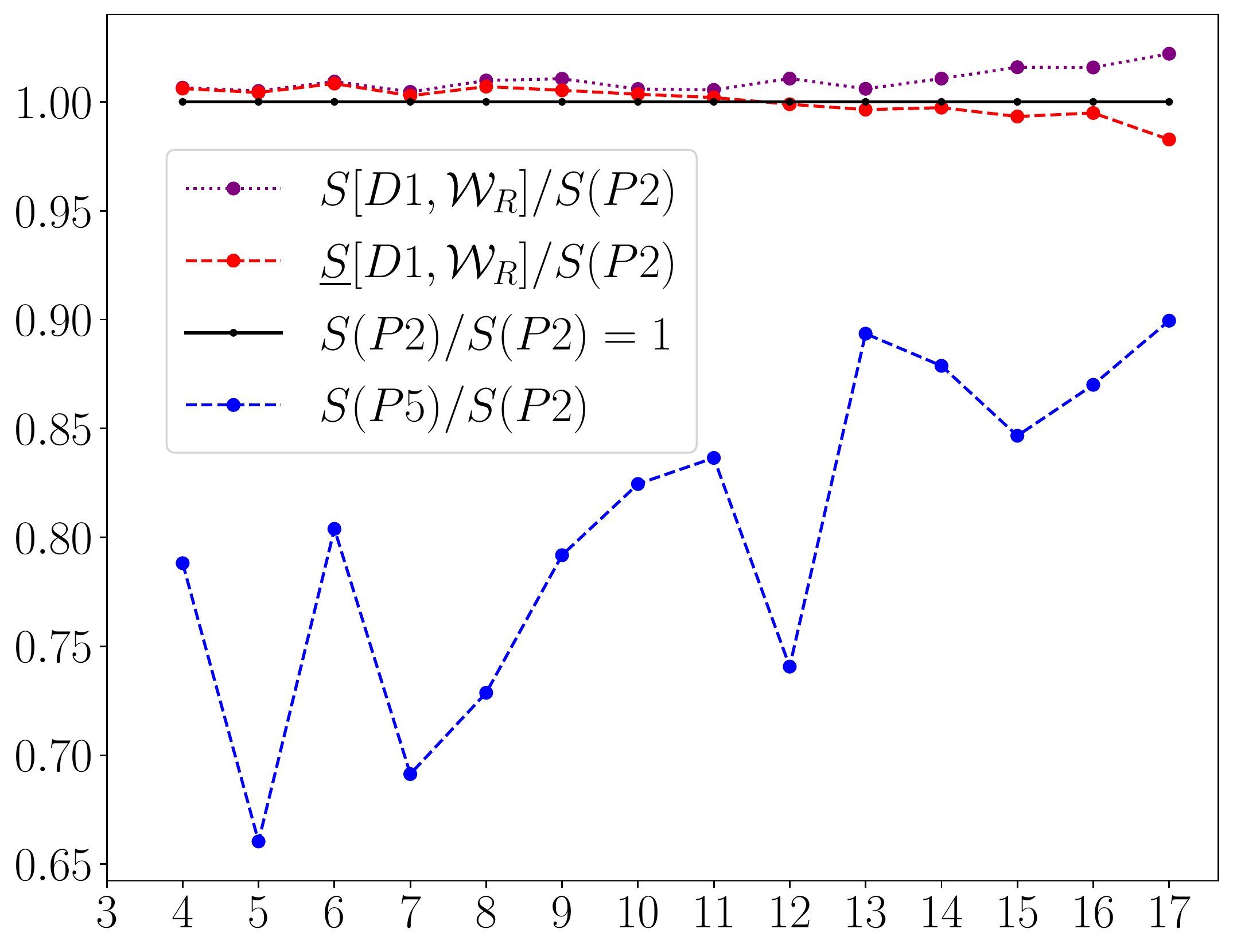}
    \caption{For each of our 50 numerical experiments for $n={\rm dim}(\mathcal H)=4,5,\dots,16,17$, we find the maximum $\underline{S}[D1,\mathcal W_R]/S(P2)$ and plot it. We also plot the corresponding values of $S[D1,\mathcal W_R]/S(P2)$ and $S(P5)]/S(P2)$. 
    In Fig.~\ref{fig:frac}, we see that there is very often a gap between the \textsf{tight} bound and the \textsf{NH} bound. Here, this figure shows us what the magnitude of this gap is. Since the gap is quite small, we can see that there is numerical evidence that the \textsf{NH} bound is a good approximation to the \textsf{tight} bound.
    }
    \label{fig:violation}
\end{figure}

\section{Applications}
\label{S-apps}
\subsection{Learning parameters of Hamiltonian models}
\label{S-apps:Hamiltonian-learning}

When one has a physical system on many qubits, the underlying Hamiltonian can always be written as
$H = \sum_{k} a_k P_k$ 
where $P_k$ are multiqubit Pauli operators. The Pauli operators $P_k$ are known, but often, we do not know the precise values of $a_k$.
In a physical system with $n$ qubits, there are up to $4^n$ coefficients for $a_k$ to determine. 
However, in realistic physical systems comprising of qubits, 
the predominant interactions are between pairs of qubits, and the interaction strength decreases as the separation of the qubits increases. 
When the interaction between the qubits is dominated by nearest-neighbor interactions, and the qubits are arranged in a 2D or 3D array, the vast majority of the coefficients $a_k$ have absolute values close to zero.
In fact in this case, the number of unknown coefficients $a_k$ that we like to estimate will be polynomial in $n$.
This can arise for instance in physical systems where only spatially adjacent
qubits have a non-negligible interaction. 
For example, we could have Heisenberg model with qubits interacting according to low-dimensional graph \cite{ouyang2019computing}, and the goal is to find the coefficients of the non-negligible interaction terms.

In general, given the Hamiltonian $H = \sum_{k=1}^d a_k P_k$, we like to estimate the parameters $a_1, \dots, a_d$ using a initial $n$-qubit probe state $\rho$.
After we initialize the physical system as the probe state $\rho$, we allow the physical system to naturally evolve with time according to the quantum channel $\mathcal N_{a_1, \dots , a_d}$. 
In the noiseless setting, the quantum channel $\mathcal N_{a_1, \dots, a_d}$ describes evolution of quantum systems according to the Schr\"odinger equation, and is given by 
\begin{align}
\mathcal N_{a_1, \dots, a_d} ( \rho)
= \exp(-i \sum_{k=1}^d a_k P_k) \rho \exp(i \sum_{k=1}^d a_k P_k).
\end{align}
In more realistic scenarios however, the quantum channel $\mathcal N_{a_1, \dots, a_d}$ that models the evolution of the initial probe state, is not necessarily a unitary channel. In general, $\mathcal N_{a_1, \dots, a_d}$ can be described using Lindblad’s formalism \cite{Lin76}, where the channels are generated by exponentiating Liouville operators. Such models have been considered for instance in Ref \cite{yang2018quantum}.
We can see that in both settings with and without noise, the channel $\mathcal N_{a_1,\dots,a_d}$ is differentiable with respect to the parameters $a_1, \dots, a_d$.
In the language of multiparameter quantum estimation, we can define the model
\begin{align}
\mathcal M = \{ \mathcal N_{a_1, \dots , a_d} ( \rho ) : a_1,\dots , a_d \in \mathbb R \},
\end{align}
which is differentiable set of quantum states.
The task at hand is then to find the optimal quantum estimator that lets us learn the values of $a_1, \dots, a_d$ while also minimizing the sum of the mean squared errors of $a_1, \dots, a_d$.

Matsumoto \cite{KM} showed that when the model $\mathcal M $ is a set of pure states,
then the corresponding \textsf{tight} bound is in fact equivalent to the \textsf{HN} bound. 
Since the noiseless setting corresponds to having $\mathcal N_{a_1, \dots, a_d}$ as a unitary channel, having the initial probe state to be a pure state would make the model $\mathcal M$ comprise of only pure states. 
In this case, we can simply evaluate $S(P4)$, the \textsf{HN} bound, and know that $S(P4)$ is equal to the \textsf{tight} bound. 

Now, if we consider the case where $\mathcal N_{a_1, \dots, a_d}$ is a non-unitary channel which is generated by exponentiating Liouville operators, even if the initial probe state $\rho$ is pure, the model $\mathcal M$ need not be a set of only pure states. In this scenario, we have no guarantee that the $\textsf{tight}$ bound is equal to the \textsf{HN} bound.
To find the optimal quantum estimator, we can nonetheless use our method for finding the asymptotically for the \textsf{tight} bound. 

As a concrete example, consider a two-qubit Hamiltonian $H_{a,b} = a X\otimes X + b Y \otimes Y$, where $X$ and $Y$ are the usual single qubit Pauli operators. Such a Hamiltonian $H_{a,b}$ models an XY Heisenberg interaction between a pair of qubits,
and such interactions are ubiquitous in nature. Being able to estimate $a$ and $b$ is essential to characterize such Hamiltonians. We also consider having some dissipative jump operators $L_j$ that describe noise in the physical system. 
In this case, the Liouville operator is given by $\mathcal L_{a,b}$, where 
\begin{align}
\mathcal L_{a,b}(\rho) = -i[H_{a,b} ,\rho] + \sum_{j} \gamma_j
\left(  L_j \rho L_j^\dagger   - \frac{1}{2}\{L_j^\dagger L_j ,\rho\} \right),
\label{liouville}
\end{align}
$[H_{a,b},\rho]=H_{a,b}\rho-\rho H_{a,b}$, $\{L_j^\dagger L_j ,\rho\}  = L_j^\dagger L_j \rho 
+\rho L_j^\dagger L_j $, and $\gamma_j$ are real numbers.
The quantum channel is 
\begin{align}
\mathcal N_{a, b}(\rho) = e^{\mathcal L_{a,b}}( \rho) ,
\end{align}
and the corresponding model is 
$\mathcal M = \{   \exp(\mathcal L_{a,b}( \rho) ) : a,b \in \mathbb R \}$. 
Our method can address such problems. 
As Section \ref{S7} illustrates, we are able to calculate tight two-sided bounds on the \textsf{tight} bound for such models with two-parameter estimation. We are also able to obtain corresponding near-optimal quantum estimators for the \textsf{tight} bound.
Hence, for this example, we are able to determine the optimal uncorrelated measurement strategy to estimate $a$ and $b$ with the minimum mean-square error,
which will be useful in characterizing interactions in spin-systems.

\subsection{3D field sensing}
The mathematical structure of the Hamiltonian learning problem in Section \ref{S-apps:Hamiltonian-learning} also allows us to discuss the 
multiparameter quantum metrology of field sensing.

In the research area of `field sensing' \cite{toth2014quantum}, we have a classical field that interacts with a physical system comprises of qubits, and the classical field interacts with each qubit in exactly the same way. The classical field is a 3D-vector $(x, y,z)$ with three real components.
The interaction Hamiltonian on $n$ qubits is given by 
\begin{align}
\sum_{j=1}^n (x X_j + y Y_j + z Z_j ),
\end{align}
where 
$X_j$, $Y_j$ and $Z_j$ denote multiqubit Pauli operators that apply the Pauli $X,Y$ and $Z$ respectively on the $j$th qubit, and apply the identity operator on all other qubits. 
For example when we have only two qubits, the Hamiltonian is 
\begin{align}
H = 
x (X \otimes I + I \otimes X)
+
y (Y \otimes I + I \otimes Y)
+
z (Z \otimes I + I \otimes Z)
.
\label{field-sensing-Hamiltonian}
\end{align}
Just as in Section \ref{S-apps:Hamiltonian-learning}, we consider the evolution operator to be generated by a Liouville operator $\mathcal L$, to take into account the effects of noise in the quantum system.
Namely, we consider $\mathcal L$ of the same form as \eqref{liouville}, except that we replace the Hamiltonian $H$ with what we consider in \eqref{field-sensing-Hamiltonian}.
In this case, the model is 
\begin{align}
\mathcal M =  \{  e^{\mathcal L}( \rho ): x,y,z \in \mathbb R \},
\end{align} 
where $\rho$ is a fixed two-qubit density matrix. We could for example consider $\rho$ as a GHZ state.
Then we can numerically calculate upper and lower bounds on the \textsf{tight} bound for such models, and furthermore also also obtain corresponding near-optimal quantum estimators for the \textsf{tight} bound. Hence, it is possible to use our algorithm for determining optimal estimators in multiparameter quantum metrology for the extremely practical problem of field sensing.

\section{Discussion}
\label{S8}

To summarize our results, we have presented a unified viewpoint to understand existing bounds for the Cramer-Rao bound from the viewpoint of conic linear programming. However in general, existing bounds such as the NH bound and HCR bound, are not the tight Cramer-Rao bound. Therefore, we proposed an algorithm to calculate the tight bound by using conic linear programming. Therefore, the main contribution of our paper is to propose an algorithm to calculate the tight bound. Our method also enables us to find the optimal POVM for the parameter estimation.

In more detail, we have characterized three types of bounds, 
\textsf{HN} bound, \textsf{NH} bound, and the \textsf{tight} bound
as the solution of conic linear program with different cones in the same linear space.
Their size relationship is characterized as the inclusion relation
among their corresponding cones.
Since the conic linear programing corresponding to 
the \textsf{HN} and \textsf{NH} bounds
are given as an \textsf{SDP}, they can be efficiently solved.
However, since
the cone corresponding to the \textsf{tight} bound is 
the separable cone ${\cal S}_{SEP}$, i.e., 
the cone composed of positive semi-definite operators with separable form over a bipartite system,
it is hard to solve the corresponding conic linear programing.
In the second part of this paper, 
we have tackled the problem, i.e., how to efficiently solve this 
conic linear programing.
For this aim, we have proposed
a method to convert the conic linear programing with 
the separable cone ${\cal S}_{SEP}$ to 
an \textsf{SDP} with constraints labeled by unit vectors from a real vector space.
We have derived 
upper and lower bounds on the \textsf{tight} bound 
by using the solution of the above converted SDP. 
From the optimal solution of the SDP, we give a concrete way to derive a corresponding uncorrelated measurement strategy which is asymptotically optimal.
In addition, we have proposed the method to choose 
a strategic subset of the above unit vectors from design theory. 
Then, we have given the calculation complexity of
the \textsf{tight} bound
within an additive error of $\epsilon$
when the above method is applied.
Then, we have numerically applied our method to our examples.
In these examples, we have numerically shown that 
the \textsf{tight} bound is strictly larger than the \textsf{NH} bound,
which shows the existence of 
the gap between the \textsf{tight} bound and the \textsf{NH} bound. 

In fact, several problems of quantum information can be written with 
the separable cone ${\cal S}_{SEP}$.
For example, 
the communication value can be calculated 
as a conic linear programming with the separable cone ${\cal S}_{SEP}$
\cite[Proposition 2]{c-value}.
Also, the existence of entanglement can be written as 
a problem by using the separable cone ${\cal S}_{SEP}$ \cite{Gurvits,bruss2002characterizing}.
We can expect that 
our method to approximately solve conic linear programing with 
the separable cone ${\cal S}_{SEP}$
can be extended to the above problems, and other similar problems \cite{takagi2019prl,takagi2019prx}.
Such extensions are interesting for future study.

Further, after completing this research, the reviewer informed us about
the method discussed in the reference \cite{Doherty}
to calculate the conic linear programing with separable cone.
The comparison between our method and this approach is an interesting topic, but beyond the scope of this paper.
Also, as another application of our method includes the simultaneous estimation of phase and dephasing which was discussed in \cite{Vidrighin}.
There is a possibility that our presented approach provides a new insight into this problem.
This kind of research is another interesting future study.

Another open problem is whether our conic programming framework can also be used to unify multiparameter quantum metrology problems where instead of minimizing the mean square error, one minimizes a different regret function. For instance in the reference \cite{PhysRevLett.126.120503}, minimization of an information regret function was considered. It will be interesting to see if our framework also applies in that setting.

\section*{Acknowledgement}
MH is supported in part by the National Natural Science Foundation of China (Grant No. 62171212) and
Guangdong Provincial Key Laboratory (Grant No. 2019B121203002).
YO is supported in part by NUS startup
grants (R-263-000-E32-133 and R-263-000-E32-731), 
the Quantum Engineering Programme grant NRF2021-QEP2-01-P06,
and the National Research Foundation, Prime Minister’s Office, Singapore and the Ministry of Education, Singapore under the Research Centres of Excellence program.
YO also acknowledges support from EPSRC (Grant No. EP/W028115/1). 

\appendices

\section{Applications of \texorpdfstring{$C[G]=S(D0)$ \eqref{AML-B}}{C[G]=S(D0)}}\Label{Ap1}
\subsection{One-parameter case}
To consider how to choose $a$ and $S$, 
we focus on the one-parameter case based on the idea by
\cite{hayashi97,hayashi97-2}.
Let $L$ be the SLD, and $J$ be the SLD Fisher Information.
We simplify 
$(AB+BA)/2$
by $A\circ B$ 
Assume that $g(x,x)=x^2$.
The condition \eqref{MML} can simplified as
\begin{align}
0 \le& x^2 \rho - a x L\circ \rho -S \nonumber \\
=&(x- \frac{a}{2}L) \rho(x- \frac{a}{2}L)- \frac{a^2}{4}L \rho L-S.
\end{align}
Hence, we have 
\begin{align}
- \frac{a^2}{4}L \rho L\ge S
\end{align}
Now, we choose $S= - \frac{a^2}{4}L \rho L$.
The value to be minimized is
$a-\Tr \frac{a^2}{4}L \rho L= a-\frac{a^2}{4}J
= -( \frac{\sqrt{J}}{ 2}a- \frac{1}{\sqrt{J}})^2+ 1/J$.
When $a= \frac{2}{J}$, the maximum $1/J$ is realized.

Indeed, the LHS of \eqref{MGT} is zero when 
$M$ is the spectral decomposition of the Hermitian matrix $\frac{a}{2}L$.
Hence, we can find that 
this measurement achieves the minimum value 
$1/J$.

However, in the general case, there is no simultaneous spectral decomposition of all SLDs.
Hence, the problem is very difficult.

\subsection{qubit case}
Next, 
we focus on the qubit case based on the idea by
\cite{hayashi97,hayashi97-2}.
We choose the coordinate such that 
the SLD Fisher information matrix is identical.
By applying suitable orthogonal transformation on the parameter space,
the matrix $G$ is a diagonal matrix such that
$G_{i,j} =g_j \delta_{i,j}$
with $g_i \ge 0$.
Let $L_i$ be the SLD of the $i$-th parameter.
Under the above condition, 
in the qubit case, we can check the following condition, 
as shown later:
\begin{align}
\frac{1}{2}(L_i \rho L_j+L_j \rho L_i)
=I-\rho.\Label{MBF}
\end{align}
Since only the condition \eqref{MBF} is essential in the discussion of this section,
the discussion of this section can be extended to the case of GPT 
(general probability theory)
based on Lorentz cone \cite[Appendix F]{Arai}.

We set $a_i^j= a \sqrt{g_j}\delta_{i}^j$
and
$S=- \frac{a^2}{4}(I- \rho)$.
The condition \eqref{MML} can simplified as
\begin{align}
&\sum_j g_j (x^j)^2 \rho -a \sum_{j}\sqrt{g_j} x^j L_j\circ \rho -S \nonumber \\
=&
\Bigg( \sqrt{\sum_j g_j (x^j)^2} - \frac{a}{2\sqrt{\sum_j g_j (x^j)^2}} 
\sum_{j}\sqrt{g_j} x^j L_j
\Bigg) 
\rho \nonumber\\
&\cdot\Bigg( \sqrt{\sum_j g_j (x^j)^2} - \frac{a}{2\sqrt{\sum_j g_j (x^j)^2}} 
\sum_{j}\sqrt{g_j} x^j L_j\Bigg) \nonumber \\
\ge& 0.
\end{align}
Hence, the objective value of the dual problem is
$\tr a+\Tr S
=a \sum_j \sqrt{g_j}- \frac{a^2}{4}
=-\frac{1}{4}(a -2 \sum_j \sqrt{g_j})^2
+( \sum_j \sqrt{g_j})^2$.
Hence, we find that 
$( \sum_j \sqrt{g_j})^2$ is a lower bound for CR bound.
Moreover, we expect that 
$( \sum_j \sqrt{g_j})^2$ is the solution with $a=2 \sum_j \sqrt{g_j}$.

This expectation can be checked by constructing 
a locally unbiasd estimator to realize the equality.

We denote the spectral decomposition of $L_j$ 
as $\sum_{k} s_{j,k}E_{j,k}$.
Define the vector $x(j,k) \in \mathbb{R}^d$ as follows.
It has only non-zero element in the $j$-th entry.
Its $j$-th entry is $s_{j,k}\frac{\sum_{j'} \sqrt{g_{j'}}}{\sqrt{g_j}}$
Then, we define the locally unbiasd estimator as follows.
It takes the outcomes in the set
$\{ x(j,k) \}_{j,k}$.
The POVM $\{M_{j,k}\}$ is defined as
$M_{j,k}:= \frac{\sqrt{g_j}}{\sum_{j'} \sqrt{g_{j'}}} E_{j,k}$.
Here, we 
multiply the probabilistic factor 
$\frac{\sqrt{g_j}}{\sum_{j'} \sqrt{g_{j'}}}$
in the POVM element.
Hence,
to satisfy the locally unbiasedness condition,
we need to multiply its inverse
$\frac{\sum_{j'} \sqrt{g_{j'}}}{\sqrt{g_j}}$
in the measurement outcome. 
Hence, we can check that this POVM satisfies the condition for 
the locally unbiasedness condition.

In fact, when $x= x(j,k)$,
\begin{align}
&\sqrt{\sum_j g_j (x(j,k)^j)^2}\nonumber\\
& - \frac{2 \sum_{j'} \sqrt{g_{j'}}}{2\sqrt{\sum_j g_j (x(j,k)^j)^2}} 
\sum_{j}\sqrt{g_j} x(j,k)^j L_j
\nonumber\\
=&\Bigg(  
\Big(\sqrt{g_j} s_{j,k}\frac{\sum_{j'} \sqrt{g_{j'}}}{\sqrt{g_j}}
\Big)
 -  \Big(\sum_{j'} \sqrt{g_{j'}}\Big)\sgn s_{j,k} L_j
\Bigg) \nonumber\\
=&(\sgn s_{j,k}) \Big(\sum_{j'} \sqrt{g_{j'}}\Big)
 \Big(  s_{j,k}-  L_j \Big) .
\end{align}
Hence, we find that
\begin{align}
0=&
\sum_{j,k}
\Tr \Bigg(\sum_j g_j (x(j,k)^j)^2 \rho \nonumber \\
&-a \sum_{j}\sqrt{g_j} x(j,k)^j L_j\circ \rho -S \Bigg)
M_{j,k}.
\end{align}
This equation shows the lower bound $( \sum_j \sqrt{g_j})^2$ is attainable.

\begin{proof}[Proof of \eqref{MBF}]
We choose $\sigma_i$ as
\begin{align}
\sigma_1=
\left(
\begin{array}{cc}
0 & 1 \\
1 & 0
\end{array}
\right), \quad
\sigma_2=
\left(
\begin{array}{cc}
0 & -i \\
i & 0
\end{array}
\right), \quad
\sigma_3=
\left(
\begin{array}{cc}
1 & 0 \\
0 & -1
\end{array}
\right).
\end{align}
Without loss of generality,
we can assume that $
\rho= \frac{1}{2}(I+\alpha(\mu) \sigma_3) $ with $0\le a \le 1$.

Notice that the SLD Fisher information matrix is identical.
The condition \eqref{MBF} is covariant.
That is, if the condition \eqref{MBF} under a coordinate,
it holds even with another coordinate converted from orthogonal transformation.
Hence, it is sufficient to check 
the condition \eqref{MBF} only under a specific coordinate.

Now, we consider the case that
\begin{align}
L_1=\sigma_1,\quad
L_2=\sigma_2,\quad 
L_3=\frac{1}{\sqrt{1-\alpha(\mu)^2}}(-\alpha(\mu) I+\sigma_3).\Label{NVR}
\end{align}
In this case, the SLD Fisher information matrix is identical, and 
the condition \eqref{MBF} holds.
\end{proof}

\section{Linear programming with general cone}\label{Ap2}
We consider conic linear programming  (EP) with a general cone.
Let ${\cal X}$ and ${\cal Z}$ be topological real vector spaces.
Let ${\cal P}\subset {\cal X}$ be a cone. 
Let $A$ be a linear function from ${\cal X}$ to ${\cal Z}$.
Given $c \in {\cal X}^*$ and $b \in {\cal Z}$, we consider the following minimization.

\begin{align}
EP:= \min_{x \in {\cal P}}\{ c(x) | 
Ax=b\}.
\end{align}
As the duality problem, we consider the following maximization.
\begin{align}
EP^*:= \max_{w \in {\cal Z}^*  } \{ w(b)|
-A^* w +c \in {\cal P}^*\} 
\end{align}
When the pair of $x \in {\cal P}$ and $w \in {\cal Z}^* $ satisfy the constraints, 
we have $c(x) -w(b)=(-A^* w +c)(x)\ge 0 $.
Hence, the inequality
\begin{align}
EP \ge EP^*
\end{align}
holds.
The difference $EP - EP^*$ is called the duality gap,
and its existence/non-existence is not trivial in general. 
To discuss this problem, we define two sets:
\begin{align}
{\cal F}:=& \{( c(x),Ax )\}_{x \in {\cal P}}\\
{\cal E}:= &   \mathbb{R} \times \{b\}.
\end{align}
Since 
${\cal F} \cap {\cal E}\subset  \overline{{\cal F}} \cap {\cal E}$
and $\overline{{\cal F}} \cap {\cal E}$ is a closed set, 
we have the inclusion relation
\begin{align}
\overline{{\cal F} \cap {\cal E}} \subset  \overline{{\cal F}} \cap {\cal E}
 \Label{AMP9}.
 \end{align}
The opposite inclusion relation is not trivial, and is related to 
the duality gap as follows.

\begin{theorem}[\protect{\cite[Proposition 4.2]{H2}}]\Label{THUU}
When 
\begin{align}
 \overline{{\cal F} \cap {\cal E}}=\overline{{\cal F}} \cap {\cal E},
 \Label{AMP}
 \end{align}
 we have
\begin{align}
EP=EP^*
\end{align}
\end{theorem}

When ${\cal X}$ and ${\cal Z}$ are finite-dimensional, 
the relation \eqref{AMP9} holds.
Hence, there is no duality gap.

\section{Proof of Theorem \ref{TH3}}\Label{AAM}
This appendix aims to show the non-existence of the duality gap in the respective problems
$Pl$ for $l=1,2,3,4,5$ by using Theorem \ref{THUU}
based on the notations given in Appendix \ref{Ap2}.
Since the condition \eqref{AMP} holds in the finite-dimensional case,
 Theorem \ref{TH3} immediately follows from Theorem \ref{THUU}.
This point is in contrast with the proof of Proposition \ref{TH1} by the paper \cite{hayashi97-2}
because the paper \cite{hayashi97-2} had to manage 
a difficulty in the possibility of the existence of the duality gap 
even with a finite-dimensional space ${\cal H}$.
The difficulty in \cite{hayashi97-2} was caused by handling the space of POVMs 
with $\Omega=\mathbb{R}^{d}$
as an infinite-dimensional linear vector space.
Hence, in our proof, we encounter 
the difficulty related to the potential existence of the duality gap 
only with an infinite-dimensional space ${\cal H}$.

\noindent {\bf Step 1:} Preparation and structure of this proof.

We define ${\cal X}$ as the set of bounded operators on
$\mathbb{R}^{d+1}\otimes {\cal H}$.
We define ${\cal Z}$ as the tensor product of 
$\mathbb{R}^{d^2}$ and
the set of bounded operators on ${\cal H}$.
Then, the linear map $A$ is defined as the left sides of 
\eqref{c1} and \eqref{c2}.
The element $b \in {\cal Z}$  
is defined as the right sides of \eqref{c1} and \eqref{c2}.
The element $c \in {\cal X}^*$  
is defined as \eqref{o1}.
Then, 
for $l=1$,
we choose ${\cal P}$ as ${\cal S}_{SEP}$.
For $l=2$,
we choose ${\cal P}$ as ${\cal S}_{P}$.
For $l=3$,
we choose ${\cal P}$ as ${\cal S}({\cal R}_C \otimes {\cal H})_{PPT}\cap {\cal B}''$.
For $l=4$,
we choose ${\cal P}$ as ${\cal S}({\cal R}_C \otimes {\cal H})_{P}\cap {\cal B}''$.
For $l=5$,
we choose ${\cal P}$ as ${\cal S}({\cal R}_C \otimes {\cal H})_{P}\cap {\cal B}''_\rho$.
We consider the operator norm in the space ${\cal F}$
when ${\cal H}$ is infinite-dimensional.

We denote the SLD Fisher information matrix by $J_{i,j}$,
and its inverse matrix by $J^{i,j}$.
We define $L^j:=\sum_{i}J^{i,j}L_i$.
Hence, we have
\begin{align}
\Tr L^j D_j  = \delta_i^j \Label{AMY8}.
\end{align}

Due to Theorem \ref{THUU}, it is sufficient to show \eqref{AMP}.
In the following, we show \eqref{AMP}.
Given an element $s \in \mathbb{R}$ with the condition
$(s,b)\in \overline{{\cal F}} \cap {\cal E}$,
we choose a sequence $\{X_n\}$ in ${\cal X}$
such that 
$ \{(c(X_n),AX_n)\}$ converges to $(s,b)$.
It is sufficient to show that 
there exists a sequence $\{\dot{X}_n\}$ in ${\cal P}$ 
such that $A\dot{X}_n=b$ and 
$c(\dot{X}_n)$ converges to $s$, which shows that 
$\overline{{\cal F}} \cap {\cal E} \subset 
\overline{{\cal F} \cap {\cal E}}$.

\noindent {\bf Step 2:} Definition of a new element $\dot{X}_n \in {\cal P}$.

We define the $d \times d$ matrix $a^{(n)}$
with component
$a^{(n),i}_j:=
\Tr (
\frac{1}{2}(|0\rangle \langle i| +|i\rangle \langle 0 |)
\otimes D_j )  X_n $
and the operator
$B_n:= 
\Tr_{\cal R} 
|0\rangle \langle 0| \otimes I_{\cal H}  X_n $ on ${\cal H}$.
We define 
\begin{align}
\hat{X}_n:=(I_{\cal R}\otimes B_n^{-1/2}) X_n 
(I_{\cal R}\otimes B_n^{-1/2})
\in {\cal P}
\end{align}
for any $l=1,2,3,4,5$.
Then, we have 
$\Tr_{\cal R} 
|0\rangle \langle 0| \otimes I_{\cal H}  \hat{X}_n=I_{\cal H} $.
The convergence $ AX_n \to b$ implies the relations
$a^{(n)} \to I_{{\cal R}}$ and 
$B_n \to I _{\cal H}$
in the sense of the operator norm.
We define the $d \times d$ matrix $\hat{a}^{(n)}$ by the components,
$\hat{a}^{(n),i}_j:=
\Tr (
\frac{1}{2}(|0\rangle \langle i| +|i\rangle \langle 0 |)
\otimes D_j  )  \hat{X}_n$.
Since $B_n \to I _{\cal H}$
in the sense of the operator norm,
$B_n^{-1/2} D_jB_n^{-1/2}\to D_j$
in the sense of the trace norm.
Thus,
we have 
\begin{align}
\hat{a}^{(n)} \to I_{\cal R} \Label{A9}, \\
c(\hat{X}_n)\to s.\Label{A8}
\end{align}

We define the vector $ \dot{a}^{(n),i} \in \mathbb{R}^{d}$ as
$ \dot{a}^{(n),i}_j:= \hat{a}^{(n),i}_j -\delta_i^j$.
We define $\epsilon_n:= \max_i \|\dot{a}^{(n),i}\|$.
The relation \eqref{A9} implies 
\begin{align}
\epsilon_n \to 0.\Label{A7}
\end{align}

For $i=1, \ldots, d$,
we denote the spectral decomposition of 
the operator 
\begin{align}
Y_{n,i}:=
-\frac{1}{\epsilon_n}\sum_j \dot{a}^{(n),i}_j L^j\Label{MYU}
\end{align}
by $\sum_{u} x_u E_u$.

Then, we define the matrix
$X_{n,i}:= 
\sum_u \sum_{i'',i'} \hat{\theta}^{i'}(u) \hat{\theta}^{i''}(u)
|i''\rangle \langle  i'| \otimes E_u \in {\cal P}$,
where 
$\hat{\theta}^{i'}(u) = \delta_{i',i}(1-\delta_{i',0})u + \delta_{i',0}$ i.e.,
\begin{align}
\hat{\theta}^{i'}(u)=
\left\{
\begin{array}{cc}
 u  & \hbox{ for }i'=i \\
 1  & \hbox{ for }i'=0 \\
0  & \hbox{ otherwise. }
\end{array}
\right.
\end{align}

Then, we have
\begin{align}
X_{n,i}= &
\sum_{i'',i'} 
|i''\rangle \langle  i'| \otimes 
\big(\delta_{i',i}(1-\delta_{i',0})Y_{n,i} + \delta_{i',0}I_{\cal H}\big)\nonumber \\
&\cdot \big(\delta_{i'',i}(1-\delta_{i'',0})Y_{n,i} + \delta_{i'',0}I_{\cal H}\big)
\Label{MML-B}.
\end{align}
We define 
\begin{align}
\dot{X}_n:=(1-d\epsilon_n)
T_n \hat{X}_n
T_n+\epsilon_n \sum_{i=1}^d X_{n,i}
\end{align}
where
\begin{align}
T_n:=
\frac{1}{1-d\epsilon_n}
(I- |0\rangle \langle 0| \otimes I_{\cal H})
+|0\rangle \langle 0| \otimes I_{\cal H}.
\end{align}
Since $X_{n,i}$ belongs to ${\cal P}$,
$\hat{X}_{n,i}$ belongs to ${\cal P}$ for any $l=1,2,3,4$.
Hence, $T_n \hat{X}_n T_n$ belongs to ${\cal P}$ for any $l=1,2,3,4$.
Thus, $\dot{X}_n$ belongs to ${\cal P}$ for any $l=1,2,3,4$.


\noindent {\bf Step 3:} Asymptotic behavior of new element $\dot{X}_n\in {\cal P}$.

For $i=1, \ldots, d$, the relations \eqref{AMY8}, \eqref{MYU}, 
and \eqref{MML-B} imply
\begin{align}
& \Tr \Big[
(\frac{1}{2}(|0\rangle \langle i'| +|i'\rangle \langle 0 |)
\otimes D_j )  X_{n,i} \Big]
= 
-\frac{1}{\epsilon_n}\dot{a}^{(n),i}_j\delta_{i,i'}\nonumber\\
&\Tr_{\cal R} \big[
(|0\rangle \langle 0|
\otimes I_{\cal H} ) X_{n,i} \big]= I_{\cal H}\nonumber\\
& \Tr\big[
(|i'\rangle \langle i''| 
\otimes \rho ) X_{n,i}\big] = \delta_{i,i''}\delta_{i,i'}
\frac{1}{\delta_n}
\sum_{k,k'}J^{k,k'} \dot{a}^{(n),i}_k \dot{a}^{(n),i}_{k'}.
\end{align}
Hence, 
\begin{align}
& \Tr \Big[
(\frac{1}{2}(|0\rangle \langle i'| +|i'\rangle \langle 0 |)
\otimes D_j )  \dot{X}_{n} \Big]\nonumber\\
&= 
\frac{1-d\epsilon_n}{1-d\epsilon_n}
\hat{a}^{(n),i}_j 
-\frac{\epsilon_n}{\epsilon_n}
\sum_i \dot{a}^{(n),i}_j\delta_{i,i'}=\delta_{i',j}\Label{MK1},
\\
&\Tr_{\cal R} \big[
(|0\rangle \langle 0|
\otimes I_{\cal H} ) \dot{X}_{n} \big]= (1-d\epsilon_n) I_{\cal H}
+d\epsilon_n I_{\cal H}= I_{\cal H}\Label{MK2},
\\
& \Tr \big[
(|i'\rangle \langle i''| 
\otimes \rho)  \dot{X}_{n} \big]\nonumber \\
&= 
\frac{(1-d\epsilon_n)^2}{1-d\epsilon_n}
\Tr \big[
(|i'\rangle \langle i''| \otimes \rho ) \hat{X}_{n} \big]
\nonumber\\
&\quad +\epsilon_n \sum_{i}
\delta_{i,i''}\delta_{i,i'}
\frac{1}{\delta_n}
\sum_{k,k'}J^{k,k'} \dot{a}^{(n),i}_k \dot{a}^{(n),i}_{k'} \nonumber\\
 &= 
(1-d\epsilon_n)
\Tr \big[
(|i'\rangle \langle i''| \otimes \rho ) \hat{X}_{n}\big] \nonumber\\
&\quad +\epsilon_n \sum_{i}
\delta_{i,i''}\delta_{i,i'}
\frac{1}{\delta_n}
\sum_{k,k'}J^{k,k'} \dot{a}^{(n),i}_k \dot{a}^{(n),i}_{k'} .\Label{MK3}
\end{align}

Using \eqref{MK3}, we have
\begin{align}
&c(\dot{X}_{n})=
\sum_{i',i''} G_{i',i''}
\Tr\big[
(|i'\rangle \langle i''| 
\otimes \rho)  \dot{X}_{n} \big]\nonumber\\
= &
\sum_{i',i''} G_{i',i''}
\Big(
(1-d\epsilon_n)
\Tr
(|i'\rangle \langle i''| \otimes \rho ) \hat{X}_{n} \nonumber\\
&+\epsilon_n \sum_{i}
\delta_{i,i''}\delta_{i,i'}
\frac{1}{\delta_n}
\sum_{k,k'}J^{k,k'} \dot{a}^{(n),i}_k \dot{a}^{(n),i}_{k'} \Big)\nonumber\\
=&
(1-d\epsilon_n)c(\hat{X}_{n} )
+
\frac{1}{\epsilon_n}
\sum_{i} G_{i,i}
\sum_{k,k'}J^{k,k'} \dot{a}^{(n),i}_k \dot{a}^{(n),i}_{k'} .\Label{MK4}
\end{align}
Since 
\begin{align}
0 \le &
\frac{1}{\epsilon_n}
\sum_{i} G_{i,i}
\sum_{k,k'}J^{k,k'} \dot{a}^{(n),i}_k \dot{a}^{(n),i}_{k'} 
\le
\frac{d}{\epsilon_n}
\| G \| \|J\| \max_i \|\dot{a}^{(n),i}\|^2
\nonumber\\
=& d \| G \| \|J\|\epsilon_n\to 0,
\end{align}
using \eqref{A8}, \eqref{A7}, and \eqref{MK4},
we have 
\begin{align}
c(\dot{X}_{n}) \to \lim_{n \to \infty}c(\hat{X}_{n} )=s.\Label{MUL}
\end{align}
Since
$\dot{X}_n$ belongs to ${\cal P}$, the relations 
\eqref{MK1} and \eqref{MK2}
guarantee
$(c(\dot{X}_{n}),b) \in {\cal F} \cap {\cal E}$.
Thus, \eqref{MUL} implies $(s,b)\in \overline{{\cal F} \cap {\cal E}}$.
Hence, we obtain \eqref{AMP} for any $l=1,2,3,4$.

\noindent {\bf Step 4:} Case of $l=5$.

Now, we discuss the case of $l=5$.
We define the $d \times d$ antisymmetric matrix 
$e^{n}$ as
\begin{align}
e^{n,i}_j:=&
\Im \Tr \big[(|i\rangle \langle j| \otimes \rho)\dot{X}_n\big] \nonumber \\
=&
\frac{1}{(1-d\epsilon_n)^2}
\Im \Tr \big[(|i\rangle \langle j| \otimes \rho)\hat{X}_n\big]\nonumber \\
=&
\frac{1}{(1-d\epsilon_n)^2}
\Im \Tr \big[(|i\rangle \langle j| \otimes B_n^{-1/2} \rho B_n^{-1/2}) X_n\big]
\end{align}
for $i,j=1, \ldots, d$.
Since $B_n\to I_{\cal H}$ in the sense of the operator norm,
$B_n^{-1/2} \rho B_n^{-1/2} \to \rho$ in the sense of the trace norm.
Since $X_n \in {\cal P}$, 
$\Im \Tr (|i\rangle \langle j| \otimes \rho){X}_n =0$.
Since $\epsilon_n \to 0$,
the above discussions imply $e^{n,i}_j \to 0$.
That is, we have 
\begin{align}
\Tr G |e^{(n)}|\to 0.\Label{NDK}
\end{align}
We define 
the operator $ \tilde{X}_n$ as
\begin{align}
\tilde{X}_n:=\dot{X}_n+ (|e^{(n)}|- \sqrt{-1}e^{(n)})\otimes I_{\cal H}.
\end{align}
Since $|e^{(n)}|- \sqrt{-1}e^{(n)} \ge 0$ and
$\dot{X}_n \in {\cal S}({\cal R}_C \otimes {\cal H})_{P}\cap {\cal B}''$,
the operator $\tilde{X}_n$ belongs to ${\cal P}=
{\cal S}({\cal R}_C \otimes {\cal H})_{P}\cap {\cal B}''_\rho$.

Since the $(0,0)$ and $(0,j)$ components of $ |e^{(n)}|- \sqrt{-1}e^{(n)}$ 
are zero for $j=1, \ldots, d$,
the relation \eqref{MK1} implies
\begin{align}
& \Tr \Big[
(\frac{1}{2}(|0\rangle \langle i'| +|i'\rangle \langle 0 |)
\otimes D_j )  \tilde{X}_{n} \Big]\nonumber\\
=& \Tr \Big[
(\frac{1}{2}(|0\rangle \langle i'| +|i'\rangle \langle 0 |)
\otimes D_j )  \dot{X}_{n} \Big]=\delta_{i',j},\Label{MK1B}
\end{align}
and
the relation \eqref{MK2} implies
\begin{align}
&\Tr_{\cal R}\big[ 
(|0\rangle \langle 0|
\otimes I_{\cal H} ) \tilde{X}_{n} \big]
=\Tr_{\cal R} \big[
(|0\rangle \langle 0|
\otimes I_{\cal H} ) \dot{X}_{n} \big]
= I_{\cal H}\Label{MK2B}.
\end{align}
Since 
$c(\tilde{X}_n)=c(\dot{X}_n)+\Tr G |e^{(n)}|$,
the combination of \eqref{MUL} and \eqref{NDK} implies 
\begin{align}
c(\tilde{X}_n)\to s\Label{MULB}.
\end{align}
Since
$\dot{X}_n$ belongs to ${\cal P}$, the relations 
\eqref{MK1B} and \eqref{MK2B}
guarantee
$(c(\dot{X}_{n}),b) \in {\cal F} \cap {\cal E}$.
Thus, \eqref{MULB} implies $(s,b)\in \overline{{\cal F} \cap {\cal E}}$.
Hence, we obtain \eqref{AMP} for $l=5$.

\section{Proof for Theorem \ref{TH96}}\Label{ApD}
The reference \cite[Theorem 6]{hayashi97-2} spends long pages to show 
Proposition \ref{TH1} because the conic linear programming $P0$ is given as
a conic linear programming in an infinite-dimensional space.
However, due to Theorems \ref{TH2} and \ref{TH3},
we can show Proposition \ref{TH1} by showing the relation $S(D0)=S(D1)$.

\begin{theorem}
We have
\begin{align}
S(D0)=S(D1).
\end{align}
\end{theorem}

\begin{proof}
The condition for $D0$ can be rewritten as
\begin{align*}
\min_{y \in {\cal H}}
\min_{x \in {\mathbb R}^d} 
\left(
( x^T G x)  y^\dagger \rho y  - \sum_{i,j=1}^d a_i^j x^j y^\dagger D_j  y  -y^\dagger S y 
 \right)  &\ge 0 \Label{eq:constraint-minimization} \\
 \mbox{subject to $y^\dagger y$}&=1 . 
\end{align*}
For $x \in {\mathbb R}^d $, we define 
$v(x) \in {\cal R}$ as
$v(x):= \sum_{j=1}^d x_j |j\rangle+|0\rangle$.
Then, the condition \eqref{eq:constraint-minimization} is rewritten as
\begin{align}
0 \le &\Tr \Big[ | v(x)\rangle \langle v(x)| \otimes | y\rangle \langle y|
 \nonumber \\
& \Big(G \otimes \rho - \frac{1}{2}\sum_{1 \le i, j \le d} a_i^j 
( | 0\rangle \langle i|+ | i\rangle \langle 0|) \otimes D_j
- | 0\rangle \langle 0| \otimes S \Big) \Big]
\end{align}
for $ y \in {\cal H}$ and $x \in {\mathbb R}^d$.
This condition is rewritten as
\begin{align}
0 \le &\Tr \Big[
| v\rangle \langle v| \otimes | y\rangle \langle y| 
 \nonumber \\
& \Big(G \otimes \rho - \frac{1}{2}\sum_{1 \le i, j \le d} a_i^j 
( | 0\rangle \langle i|+ | i\rangle \langle 0|) \otimes D_j
- | 0\rangle \langle 0| \otimes S \Big) \Big]
 \Label{aa1}
\end{align}
for $ y \in {\cal H}$ and $v \in {\cal R}$ to satisfy $\langle v|0\rangle\neq 0 $.

When \eqref{aa1} holds 
for $ y \in {\cal H}$ and $v \in {\cal R}$ to satisfy $\langle v|0\rangle\neq 0 $,
it holds for $ y \in {\cal H}$ and $v \in {\cal R}$
because of the following reason.
We choose $v \in {\cal R}$ such that $\langle v|0\rangle= 0 $.
Then, we choose a sequence $ v_n$ such that $v_n \to v$ and 
$\langle v_n|0\rangle\neq 0 $.
Since 
\eqref{aa1} holds with $v_n$, 
by taking the limit, \eqref{aa1} holds even with $v$. 

Therefore, the condition for $D0$ can be rewritten as the condition that
the inequality \eqref{aa1} holds with any $ y \in {\cal H}$ and any $v \in {\cal R}$.
This condition is equivalent to \eqref{NYI}.
Hence, we conclude that the two dual problems $D0$ and $D1$ are equivalent.
\end{proof}

\section{Evaluation of \texorpdfstring{$\kappa$}{kappa} based on 
\texorpdfstring{${\cal W}_R$ constructed in Section \ref{SSD}}{WR constructed in Section VC}}
\Label{ACAS}
Since the lower bound of $S(D1)$ depends on the value  $\kappa$, 
it is better to construct ${\cal W}$ by reflecting the structure of $\kappa$.
Hence, we choose ${\cal W}_R$ as follows.
We choose a subset ${\cal S} \subset S^{d-1}$.
We define a subset 
${\cal S}(\phi):=\{(\cos \phi, \sin \phi y)\}_{y \in \mathcal S}$.
We choose $\phi_0>0$ as $\tan \phi_0= C_2(a^*)$.
Using $k$, we choose ${\cal W}_R$ as
$\cup_{j=0}^k{\cal S}(\frac{\phi_0 j}{k})$.

Indeed, when $d=2$, the choice of ${\cal S}$ is very easy
because we can choose ${\cal S}$
as $\{ (\cos \frac{2\pi l}{N}, \sin \frac{2\pi l}{N}) \}_{l=1}^{N}$.
In this case, we find that $\delta(\mathcal S) = \| (1,0) - (\cos 2\pi/N , \sin 2\pi/N)\|
= \sqrt{2-2\cos 2\pi/N}$.

In this case, to evaluate the quantity $\kappa$, we employ Lemma \ref{LAH1}.
In this method, we need to evaluate the quantity
$\delta( s,{\cal W}_R)$.
For this aim, we define 
\begin{align}
\bar{\delta}({\cal S})&:= 
\max_{x \in \mathbb{R}^{d} : \|x\|=1} 
\min_{|w\> \in \mathcal {\cal S}} \| |x\>- |w\> \| \\
\bar{\delta}(s,\mathcal{W}_R)&:=
\max_{x=(\cos \tan^{-1}s , (\sin \tan^{-1}s)z ),  
z \in S^{d-1}}
 \min_{|w\> \in \mathcal W_R} \| |x\>- |w\> \|.
\end{align}

To consider the relation between 
$\bar{\delta}(s,\mathcal{W}_R)$
and $\delta(s,\mathcal{W}_R)$, we focus on 
the relations
\begin{align}
&\| |0\rangle \langle 0| - 
(\cos \theta|0\rangle +\sin \theta|1\rangle ) 
(\cos \theta \langle 0| +\sin \theta\langle 1| ) \|_1\nonumber \\
=&
2|\sin \theta| \Label{BB1}\\
&\| |0\rangle  - 
(\cos \theta|0\rangle +\sin \theta|1\rangle )  \|
=\sqrt{2 (1-\cos \theta)} \cong|\sin \theta|.\Label{BB2}
\end{align}
Then, we define function $\eta: \sqrt{2 (1-\cos \theta)}\mapsto 2|\sin \theta|$.
When $t>0$ is small, 
the relations \eqref{BB1} and \eqref{BB2} implies 
$\eta(t)\cong 2t$. 
This function connects two quantities
$\bar{\delta}(s,\mathcal{W}_R)$ and $\delta(s,\mathcal{W}_R)$
as $\delta(s,\mathcal{W}_R)=\eta (\bar{\delta}(s,\mathcal{W}_R))$.
Then, we have the following lemma.
\begin{lemma}\Label{ACRT}
Under the above construction of ${\cal W}_R$,
for $s=\tan \theta$, we have
\begin{align}
&\bar{\delta}(s,{\cal W}_R) \nonumber \\
\le &\min_{j \in \{0,1,\ldots, k\}} 
\Big(
 (\cos \theta - \cos \frac{\phi_0 j}{k})^2\nonumber \\
& + \Big(
 | \sin \theta | \bar{\delta}(\mathcal S) +  
\Big| \sin \theta -  \sin \frac{\phi_0 j}{k} \Big|
 \Big)^2 
\Big)^{\frac{1}{2}}.
\end{align}
\end{lemma}
This lemma implies that
\begin{align}
&\delta(s,{\cal W}_R)\nonumber \\
\le & \min_{j \in \{0,1,\ldots, k\}} 
\eta\Big(
\Big(
 (\cos \theta - \cos \frac{\phi_0 j}{k})^2\nonumber \\
& + \Big(
 | \sin \theta | \bar{\delta}(\mathcal S) +  
\Big| \sin \theta -  \sin \frac{\phi_0 j}{k} \Big|
 \Big)^2 \Big)^{\frac{1}{2}} \Big).\Label{AMZ}
\end{align}
Combining \eqref{AMZ} and Lemma \ref{LAH1}, we have
\begin{align}
&\kappa \nonumber \\
\le & \|X^*\|  
\max_{s \in[0, C_2(a^*) ]}(1+s^2 ) 
\min_{j \in \{0,1,\ldots, k\}} 
\eta\Big(
\Big(
 (\cos \theta - \cos \frac{\phi_0 j}{k})^2\nonumber \\
& + \Big(
 | \sin \theta | \bar{\delta}(\mathcal S) +  
\Big| \sin \theta -  \sin \frac{\phi_0 j}{k} \Big|
 \Big)^2 \Big)^{\frac{1}{2}} \Big),
\end{align} 
where $s=\tan \theta$.
\begin{proof}
Using the definition of the set $\mathcal S$, we get
\begin{align}
&\bar{\delta}(s,{\cal W}_R) \nonumber \\
=&
\max_{x=(\cos \theta, \sin \theta z),  z \in S^{d-1}}
 \min_{0 \le j \le k} \min_{y \in \mathcal S} 
 \Big\| |x\>- (\cos \frac{\phi_0 j}{k}, \sin \frac{\phi_0 j}{k} y) \Big\|
 \notag\nonumber \\
 =&
\max_{z \in S^{d-1}}
 \min_{0 \le j \le k} \min_{y \in \mathcal S} \Big\| (\cos \theta, \sin \theta z) - (\cos \frac{\phi_0 j}{k}, \sin \frac{\phi_0 j}{k} y) \Big\| \notag\nonumber \\
 =&
\max_{\substack{z \in \mathbb R^d \\ \|z\|=1}}
 \min_{0 \le j \le k} \min_{y \in \mathcal S} \Big\| (\cos \theta, \sin \theta z) - (\cos \frac{\phi_0 j}{k}, \sin \frac{\phi_0 j}{k} y)\Big\| \notag\nonumber \\
 \le &
 \min_{0 \le j \le k} 
\max_{\substack{z \in \mathbb R^d \\ \|z\|=1}}
 \min_{y \in \mathcal S} \Big\| (\cos \theta, \sin \theta z) - (\cos \frac{\phi_0 j}{k}, \sin \frac{\phi_0 j}{k} y) \Big\| \notag \nonumber \\
 =&
 \min_{0 \le j \le k} 
\max_{\substack{z \in \mathbb R^d \\ \|z\|=1}}
 \min_{y \in \mathcal S} 
\Big(
 (\cos \theta - \cos \frac{\phi_0 j}{k})^2 \nonumber \\
& + \Big\| \sin\theta z -  \sin \frac{\phi_0 j}{k} y \Big\|^2 
\Big)^{\frac{1}{2}},
 \end{align}
where the inequality follows from the max-min inequality.

Now note that for $\|y\| = 1$,
\begin{align}	
&\Big\| \sin\theta z -  \sin \frac{\phi_0 j}{k} y \Big\|\nonumber \\
=&
\Big\| \sin\theta z  - \sin \theta y +  \sin \theta y -  \sin \frac{\phi_0 j}{k} y 
\Big\|
\notag\nonumber \\
\le&
\| \sin\theta z  - \sin \theta y \| +  
\Big\| \sin \theta y -  \sin \frac{\phi_0 j}{k} y \Big\|
\notag\nonumber \\
=&
| \sin \theta | \| z - y \| +  
\Big| \sin \theta -  \sin \frac{\phi_0 j}{k} \Big|.
\end{align}
Hence 
\begin{align}
&\bar{\delta}(s,{\cal W}_R) \nonumber \\
\le &
 \min_{0 \le j \le k} 
\max_{\substack{z \in \mathbb R^d \\ \|z\|=1}}
 \min_{y \in \mathcal S} 
\Big(
 (\cos \theta - \cos \frac{\phi_0 j}{k})^2
\nonumber \\
& + \Big(
 | \sin \theta | \| z - y \| +  
\Big| \sin \theta -  \sin \frac{\phi_0 j}{k} \Big|
 \Big)^2 
\Big)^{\frac{1}{2}}
\nonumber \\
=&
 \min_{0 \le j \le k} 
\Big(
 (\cos \theta - \cos \frac{\phi_0 j}{k})^2\nonumber \\
& + \Big(
 | \sin \theta | \bar{\delta}(\mathcal S) +  
\Big| \sin \theta -  \sin \frac{\phi_0 j}{k} \Big|
 \Big)^2 \Big)^{\frac{1}{2}}.
\end{align}
\end{proof}

In the previous choice of ${\cal W}_R$, we use the same subset of ${\cal S}^{d-1}$
to construct ${\cal S}(\phi) $.
We modify this construction. That is, 
to construct ${\cal S}(\frac{\phi_0 j}{k}) $,
depending on $j$, we choose a subset ${\cal S}_j \subset S^{d-1}$.
We define a subset 
${\cal S}(\frac{\phi_0 j}{k}):=\{(\cos \frac{\phi_0 j}{k}, \sin \frac{\phi_0 j}{k} y)
\}_{y \in {\mathcal S}_j}$.
We choose $\phi_0>0$ as $\tan \phi_0= C_2(a^*)$.
Using $k$, we choose ${\cal W}_R$ as
$\cup_{j=0}^k{\cal S}(\frac{\phi_0 j}{k})$.

Then, we can show the following lemma in the same way as Lemma \ref{ACRT}.
\begin{lemma}\Label{ACRT2}
Under the above construction of ${\cal W}_R$,
for $s=\tan \theta$, we have
\begin{align}
&\bar{\delta}(s,{\cal W}_R)\nonumber \\
\le &\min_{j \in \{0,1,\ldots, k\}} 
\Big( (\cos \theta - \cos \frac{\phi_0 j}{k})^2\nonumber \\
& + \Big(
 | \sin \theta | \bar{\delta}({\mathcal S}_j) +  
\Big| \sin \theta -  \sin \frac{\phi_0 j}{k} \Big|
 \Big)^2 
 \Big)^{\frac{1}{2}}.
\end{align}
\end{lemma}
This lemma implies that
\begin{align}
&\delta(s,{\cal W}_R)\nonumber \\
\le & \min_{j \in \{0,1,\ldots, k\}} 
\eta\Big(
\Big(
 (\cos \theta - \cos \frac{\phi_0 j}{k})^2\nonumber \\
& + \Big(
 | \sin \theta | \bar{\delta}(\mathcal{S}_j) +  
\Big| \sin \theta -  \sin \frac{\phi_0 j}{k} \Big|
 \Big)^2  \Big)^{\frac{1}{2}}
\Big).\Label{AMZ6}
\end{align}
Combining \eqref{AMZ6} and Lemma \ref{LAH1}, we have
\begin{align}
& \kappa \nonumber \\
\le & \|X^*\|  
\max_{s \in[0, C_2(a^*) ]}(1+s^2 ) 
\min_{j \in \{0,1,\ldots, k\}} 
\eta\Big(
\Big(
 (\cos \theta - \cos \frac{\phi_0 j}{k})^2 \nonumber \\
& + \Big(
 | \sin \theta | \bar{\delta}(\mathcal{S}_j) +  
| \sin \theta -  \sin \frac{\phi_0 j}{k} |
 \Big)^2 \Big)^{\frac{1}{2}} \Big),
\end{align} 
where $s=\tan \theta$.

\section{Proof of \texorpdfstring{Theorem \ref{norm}}{our upper bound on the norm of Pi(a**,S**) and the norm of a**}}
\subsection{Proof of \texorpdfstring{\eqref{XMR4}}{our upper bound on the norm of Pi(a**,S**)}}
\begin{lemma}\Label{LL4}
For an Hermitian matrix $X$, we have
\begin{align}
\Big\| \Big(\frac{1}{2}(X \rho+\rho X)\Big)^2\Big\|
\le \Tr X \rho X.
\end{align}
\end{lemma}
\begin{proof}
We choose a normalized vector $|\phi\rangle$ to satisfy
\begin{align}
\Big\| \Big(\frac{1}{2}(X \rho+\rho X)\Big)^2\Big\|
=\langle \phi \Big(\frac{1}{2}(X \rho+\rho X)\Big)^2|\phi\rangle.
\end{align}
We choose another normalized vector $|\psi\rangle$
as a constant times of $\Big(\frac{1}{2}(X \rho+\rho X)\Big)|\phi\rangle$.
Then, we have 
\begin{align}
&\langle \phi \Big(\frac{1}{2}(X \rho+\rho X)\Big)^2|\phi\rangle
=
|\langle \psi \Big(\frac{1}{2}(X \rho+\rho X)\Big)|\phi\rangle|^2 \nonumber \\
=&
\Big|\Tr \Big(\frac{1}{2}(X \rho+\rho X)\Big)
 |\phi\rangle \langle \psi| \Big|^2.
\end{align}
We apply Schwartz inequality to the inner product
$\langle X,Y\rangle:= \Tr \Big[\frac{1}{2}(X \rho+\rho X)\Big]
 Y^\dagger$. Then, we have
\begin{align}
&\Big| \Tr \Big[\frac{1}{2}(X \rho+\rho X)\Big]
 |\phi\rangle \langle \psi| \Big|^2 \nonumber\\
\le & \Tr \Big[\Big(\frac{1}{2}(X \rho+\rho X)\Big)X^\dagger \Big] \nonumber\\
&\cdot  \Tr \Big[\Big(\frac{1}{2}( |\psi\rangle \langle \phi| \rho+\rho |\psi\rangle \langle \phi|)\Big) |\phi\rangle \langle \psi|\Big]
\nonumber\\
\le & \Tr \big[X \rho X\big]
\frac{1}{2}
\Tr \big[
 |\psi\rangle \langle \phi| \rho|\phi\rangle \langle \psi|+
 |\phi\rangle \langle \psi|\rho |\psi\rangle \langle \phi|) \big]
\nonumber\\
= & \Tr \big[X \rho X\big]
\frac{1}{2}
(\langle \phi| \rho|\phi\rangle +
 \langle \psi|\rho |\psi\rangle)
\le \Tr \big[X \rho X\big].
\end{align}

\end{proof}

We define
\begin{align}
S(a):=\argmax_{S} \{\Tr S| (a,S) \hbox{ satisfies the condition \eqref{MML}.}\}
\end{align}
\begin{lemma}\Label{LL5}
Then, for a real number $t$, we have 
\begin{align}
S(t a)=t^2 S(a).\Label{ACC}
\end{align}
Also, we have
\begin{align}
-S(a) \ge 0. \Label{AMY}
\end{align}

\end{lemma}

\begin{proof}
The condition \eqref{MML} with $ x=0$ implies \eqref{AMY}.

Eq \eqref{ACC} can be shown as follows.

The condition \eqref{MML} for $(ta,S)$ is written as follows.
Any vector $x \in \mathbb{R}^d$ satisfies
\begin{align}
(y^T G y) \rho - \sum_{i,j} a_i^j t x^i D_j -S \ge 0.\Label{MML4}
\end{align}
By using a vector $y=t^{-1}x \in \mathbb{R}^d$,
\eqref{MML4} is rewritten as
\begin{align}
t^2 (y^T G y)\rho - \sum_{i,j} a_i^j t^2 y^i D_j -S \ge 0.\Label{MML6}
\end{align}
That is, it is rewritten as
\begin{align}
(y^T G y)\rho - \sum_{i,j} a_i^j  y^i D_j -t^{-2}S \ge 0,\Label{MML7}
\end{align}
which implies \eqref{ACC}.
\end{proof}

In the following, we assume that $G$ is the identity matrix.
\begin{lemma}\Label{LL6}
For a real number $t$, we have 
\begin{align}
-\Tr S(a) \ge
\frac{1}{d}\tr a^T J a.
\end{align}
\end{lemma}

\begin{proof}
The condition \eqref{MML} for $(a,S(a))$ is written as follows.
\begin{align}
0 \le &(x^i  -\frac{1}{2} \sum_{j=1}^d a_i^j L_j) \rho (x^i  -
\frac{1}{2} \sum_{j=1}^d a_i^j L_j) \nonumber\\
&+\sum_{i'\neq i}
\big( (x^{i'})^2\rho+ \sum_{j} a_{i'}^j t^2 x^{i'} D_j \big)\nonumber\\
&-\frac{1}{4} (\sum_{j=1}^d a_i^j L_j) \rho(\sum_{j'=1}^d a_i^{j'} L_{j'}) 
-S(a). \Label{MML17}
\end{align}
We take the diagonalization as $\frac{1}{2} \sum_{j=1}^d a_i^j L_j
=\sum_{k}c_k |\phi_k\rangle \langle \phi_k|$.
We choose $x^{i'}=0$ for $i'\neq i$ and $x^i= c_k$ in
\eqref{MML17}.
Then, we have
\begin{align}
0 \le &(c_k  -\frac{1}{2} \sum_{j=1}^d a_i^j L_j) \rho (c_k^i  -
\frac{1}{2} \sum_{j=1}^d a_i^j L_j) \nonumber\\
&-\frac{1}{4} (\sum_{j=1}^d a_i^j L_j) \rho(\sum_{j'=1}^d a_i^{j'} L_{j'}) 
-S(a)\Big)\Label{MML17B}.
\end{align}
Hence, we have
\begin{align}
 &\langle \phi_k|\Big(
-\frac{1}{4} (\sum_{j=1}^d a_i^j L_j) \rho(\sum_{j'=1}^d a_i^{j'} L_{j'}) 
-S(a)\Big)|\phi_k\rangle \nonumber\\
\stackrel{(a)}{=} &\langle \phi_k|\Big(
(c_k  -\frac{1}{2} \sum_{j=1}^d a_i^j L_j) \rho (c_k^i  -
\frac{1}{2} \sum_{j=1}^d a_i^j L_j) \nonumber\\
&-\frac{1}{4} (\sum_{j=1}^d a_i^j L_j) \rho(\sum_{j'=1}^d a_i^{j'} L_{j'}) 
-S(a)\Big)|\phi_k\rangle\nonumber\\
\stackrel{(b)}{\ge} & 0,
 \end{align}
where $(a)$ follows from the fact that $c_k$ is the 
eigenvalue of $\frac{1}{2} \sum_{j=1}^d a_i^j L_j$ with the vector $\phi_k$.
Also, $(b)$ follows from \eqref{MML17B}.

Taking the sum with respect to $k$, we have
\begin{align}
&-\Tr \frac{1}{4} (\sum_{j=1}^d a_i^j L_j) \rho(\sum_{j'=1}^d a_i^{j'} L_{j'}) 
-\Tr S(a) \nonumber\\
=&
\sum_k \langle \phi_k|\Big(
-\frac{1}{4} (\sum_{j=1}^d a_i^j L_j) \rho(\sum_{j'=1}^d a_i^{j'} L_{j'}) 
-S(a)\Big)|\phi_k\rangle \ge 0.\Label{AC9}
\end{align}
Since \eqref{AC9} holds for $i=1, \ldots, d$, we have
\begin{align}
-\Tr S(a)\ge \frac{1}{d}\sum_{i=1}^d
\Tr \frac{1}{4} (\sum_{j=1}^d a_i^j L_j) \rho(\sum_{j'=1}^d a_i^{j'} L_{j'}) 
= \frac{1}{d} a^t J a.
\end{align}
\end{proof}

\begin{lemma}\Label{LL7}
\begin{align}
\|X(a,S(a))\| \le 
\|G \| +
(-d\Tr S(a))^{1/2}
-\Tr S(a).
\Label{XMR7}
\end{align}
\end{lemma}

\begin{proof}
We have
\begin{align}
&\Big\| \frac{1}{2}\sum_{1 \le i, j \le d} a_i^j 
( | 0\rangle \langle i|+ | i\rangle \langle 0|) \otimes D_j\Big\|^2
\nonumber\\
=&
\Big\|\Big( \frac{1}{2}\sum_{1 \le i, j \le d} a_i^j 
( | 0\rangle \langle i|+ | i\rangle \langle 0|) \otimes D_j \Big)^2 \Big\|
 \nonumber \\
=&
\left\|\frac{1}{4}
\left(
| 0\rangle \langle 0| \otimes \Big(
\sum_{i=1}^d (\sum_{j=1}^d a_i^j D_j)^2\Big)
+\sum_{i'=1}^d
| i'\rangle \langle i'| \otimes (\sum_{j=1}^d a_{i'}^j D_j)^2
\right)
\right\|\nonumber \\
\stackrel{(a)}{=}&\Big\| \sum_{i=1}^d
 (\sum_{j=1}^d a_i^j D_j)^2\Big\|
\le \sum_{i=1}^d \Big\| 
 (\sum_{j=1}^d a_i^j D_j)^2\Big\|\nonumber\\
\stackrel{(b)}{\le}
 & \sum_{i=1}^d 
\Tr \Big[(\sum_{j=1}^d a_i^j L_j)\rho(\sum_{j=1}^d a_i^j L_j) \Big] 
= \tr a^T J a \stackrel{(c)}{\le} -d\Tr S(a),\Label{ZMP}
\end{align}
where $(a)$, $(b)$, and $(c)$ follow from 
the inequality
$ (\sum_{j=1}^d a_{i'}^j D_j)^2
\le \sum_{i=1}^d (\sum_{j=1}^d a_i^j D_j)^2$ for $i'=1, \ldots, d$,
Lemma \ref{LL4}, and Lemma \ref{LL6}, respectively.

Hence,
\begin{align}
&\|X(a,S(a))\| \nonumber \\
\le & 
\|G \otimes \rho\| +\Big\| \frac{1}{2}\sum_{1 \le i, j \le d} a_i^j 
( | 0\rangle \langle i|+ | i\rangle \langle 0|) \otimes D_j\Big\|\nonumber \\
&+\| | 0\rangle \langle 0| \otimes S(a) \| \nonumber \\
\le &
\|G \| +\Big\| \frac{1}{2}\sum_{1 \le i, j \le d} a_i^j 
( | 0\rangle \langle i|+ | i\rangle \langle 0|) \otimes D_j\Big\|
+\| S \| \nonumber \\
\stackrel{(a)}{\le} &
\|G \| +
(-d\Tr S(a))^{1/2}
-\Tr S(a),
\Label{XMR2}
\end{align}
where $(a)$ follows from \eqref{AMY} and \eqref{ZMP}.
\end{proof}

\begin{lemma}\Label{LL8}
We have
\begin{align}
 - \Tr S(a^{**})\le 
 \frac{d}{4 \tr J^{-1}}
 \Label{XMRN}
\end{align}
and
\begin{align}
 \tr a^{**} \le 
 \frac{d}{2 \tr J^{-1}}
 \Label{XMRN2}.
\end{align}

\end{lemma}
\begin{proof}
We have
\begin{align}
&\max_{(a,S)} \{\tr a+ \Tr S| (a,S) \hbox{ satisfies the condition \eqref{MML}.}\} \nonumber\\
=&\max_{a}  \tr a+ \Tr S(a) \nonumber\\
=&\max_t \max_{a} \{\tr ta+ \Tr S(ta)| \tr a=1\}\nonumber\\
\stackrel{(a)}{=} &
\max_t \max_{a} \{t + t^2\Tr S(a)| \tr a=1\}\nonumber \\
=& \max_{a} \{  \max_t t + t^2\Tr S(a)| \tr a=1\}\nonumber \\
=& \max_{a} \{  \frac{-1}{4\Tr S(a)}| \tr a=1\},\Label{MM0}
\end{align}
where $(a)$ follows from Lemma \ref{LL5}.

In particular, when $t$ equals $t_a:=-\frac{1}{2 \tr S(a)}$,
the maximum $\max_t t + t^2\Tr S(a)$ is realized.
Then, since $\Tr S(t_a a)\le 0$, due to Lemma \ref{LL5},
we have 
\begin{align}
-\Tr S(t_a a)= -t_a^2 \Tr S( a)=\frac{-1}{4\Tr S(a)}.\Label{MM1}
\end{align}

Also, due to \eqref{MM0},
using 
\begin{align}
a_*:=\argmax_{a} \Big\{  \frac{1}{4(\Tr S(a))^2}\Big| \tr a=1\Big\},
\Label{MM2}
\end{align}
 we have
\begin{align}
a^{**}=t_{a_*} a_*.\Label{MM3}
\end{align}

Hence, we have
\begin{align}
& - \Tr S(a^{**})\stackrel{(a)}{=}
 - \Tr S(t_{a_*} a_*)\stackrel{(b)}{=}
\frac{-1}{4\Tr S(a_*)} \nonumber\\
\stackrel{(c)}{=}&
\max_{a} \Big\{  \frac{-1}{4\Tr S(a)}\Big| \tr a=1\Big\} 
\stackrel{(d)}{\le}
\max_{a} \Big\{ \frac{d}{4 \tr [a^T J a]}\Big| \tr a=1\Big\}
,\Label{XMR3}
\end{align}
where $(a)$, $(b)$, and $(c)$ follow from 
\eqref{MM3}, 
\eqref{MM1}, 
and
\eqref{MM2}, respectively.
In addition, $(d)$ follows from the inequality 
\begin{align}
\frac{-1}{4\Tr S(a)} \le 
\frac{d}{4 \tr [a^T J a]},\Label{MLY}
\end{align}
which is shown by Lemma \ref{LL6}.

We consider the inner product $\langle a, b\rangle:=\tr [a b^T]$.
Schwartz inequality for $ J^{1/2}a$ and $J^{-1/2}$
implies that
\begin{align}
\tr J^{-1} \tr [a^T J a]\ge | \tr [J^{-1/2}J^{1/2}a]|^2=|\tr a|^2.
\end{align}
The equality holds when $a$ is a constant times of $J^{-1}$.
Therefore, we have the relation
\begin{align}
\max_{a} \Big\{ \frac{d}{4 \tr [a^T J a]}\Big| \tr a=1\Big\}
=\frac{d}{4 \tr J^{-1}}.\Label{XMR}
\end{align}
The combination of \eqref{XMR3} and \eqref{XMR}
yields \eqref{XMRN}.

Also, \eqref{XMRN2} is shown as
\begin{align}
\tr a^{**}
\stackrel{(a)}{=} t_{a_*}
\stackrel{(b)}{=}-\frac{1}{2 \tr S(a_*)}
\stackrel{(c)}{\le} \frac{d}{2 \tr [a_*^T J a_*]}
\stackrel{(d)}{\le} \frac{d}{4 \tr J^{-1}},
\end{align}
where 
$(b)$, $(c)$, and $(d)$ follows from the definition of $t_a$,
\eqref{MLY}, and \eqref{XMR}, respectively.
In addition,
$(a)$ follows from \eqref{MM3} and the relation $\Tr a_*=1$.
\end{proof}

\begin{proof}[Proof of \eqref{XMR4}]
When $G$ is the identity matrix, 
\eqref{XMR4} follows from Lemmas \ref{LL7} and \ref{LL8}.

When $G$ is a general positive semidefinite matrix,
we choose a new parameter $\eta= G^{1/2}\theta$.
Under the new parameter $\eta$, the weight matrix is the identity matrix,
and the Fisher information matrix is $ G^{-1/2}J G^{-1/2}$. 
Since $ (G^{-1/2}J G^{-1/2})^{-1}= G^{1/2}J^{-1} G^{1/2}$, we obtain 
\eqref{XMR4} with a general weight matrix $G$.
\end{proof}

\subsection{Proof of \texorpdfstring{\eqref{XMR6}}{our upper bound on the norm of a**}}

\begin{lemma}\Label{LL88}
When $G$ is the identity matrix,
$a^{**}$ is a positive semi-definite matrix.
\end{lemma}

\begin{proof}
Assume that $(a,S)$ satisfies the condition \eqref{MML}. 
We choose the polar decomposition of $a$ as
$a= |a| o$ with an orthogonal matrix $o$. 
Then, we have
\begin{align}
\|y\|^2 \rho +\sum_{j} (|a|o y)^j D_j -S\ge 0
\end{align}
for $y \in \mathbb{R}^d$.
We choose $y'=oy$. Then, we have
\begin{align}
\|y'\|^2 \rho +\sum_{j} (|a| y')^j D_j -S\ge 0.
\end{align}
Hence, $(|a|,S)$ also satisfies the condition \eqref{MML}
and $\tr |a| \ge \tr a$.
In fact, the equality of the above inequality holds if and only if $a$ is 
a positive semi-definite matrix.
Therefore, we can conclude that 
the optimizer $a^{**}$ is a positive semi-definite matrix.

\end{proof}

\begin{proof}[Proof of \eqref{XMR6}]
Assume that $G$ is the identity matrix.
Theorem \ref{norm} and Lemma \ref{LL88} guarantee that $\|a^{**}\|
\le \tr a^{**} \le 
 \frac{d}{2 
\sqrt{\tr J^{-1}}}$.
Hence, we obtain 
\eqref{XMR6}.

When $G$ is a general positive-semi definite matrix,
we choose a new parameter $\eta= G^{1/2}\theta$.
Then, in the same way as \eqref{XMR4}, we obtain \eqref{XMR6}.
\end{proof}

\bibliographystyle{quantum}

\onecolumn
\appendix

\end{document}